\documentclass[sn-mathphys,numbered]{sn-jnl}

\usepackage{graphicx}%
\usepackage{multirow}%
\usepackage{amsmath,amssymb,amsfonts}%
\usepackage{amsthm}%
\usepackage{mathrsfs}%
\usepackage[title]{appendix}%
\usepackage{xcolor}%
\usepackage{textcomp}%
\usepackage{manyfoot}%
\usepackage{booktabs}%
\usepackage{listings}%
\usepackage[ruled, vlined]{algorithm2e}
\usepackage{cleveref}
\usepackage[latin1]{inputenc}
\usepackage{tikz}
\usetikzlibrary{shapes,arrows}
\usepackage{mathtools}
\usepackage{color,soul}
\usepackage{setspace}
\usepackage{pgfplots}
\pgfplotsset{compat=newest}
\hypersetup{urlcolor=blue, citecolor=red}
\usepackage{float}
\usepackage{caption}
\usepackage{subcaption}
\usepackage{breqn}
\usepackage{enumitem}

\tikzstyle{decision} = [rectangle, draw, fill=green!20, 
    text width=10em, text badly centered, node distance=3cm, inner sep=2pt]
    \tikzstyle{block} = [rectangle, draw, fill=blue!20, 
    text width=8em, text centered, rounded corners, minimum height=4em]
    \tikzstyle{line} = [draw, -latex']
    \tikzstyle{cloud} = [draw, ellipse,fill=red!20, node distance=3cm,
    minimum height=2em]
\newcommand{\sig}{\operatorname{Sig}}

\theoremstyle{thmstyleone}%
\newtheorem{theorem}{Theorem}[section]
\newtheorem{proposition}[theorem]{Proposition}%
\newtheorem{lemma}[theorem]{Lemma}%
\newtheorem{assumption}{Assumption}
\newtheorem{corollary}{Corollary}

\theoremstyle{thmstyletwo}%

\theoremstyle{thmstylethree}%
\newtheorem{remark}{Remark}
\newtheorem{definition}[theorem]{Definition}%
\newtheorem{proof-app}{Proof}

\begin{document}

\title[Learning 2-complexes]{Towards Stratified Space Learning: 2-complexes}


\author{\fnm{Yossi} \sur{Bokor Bleile}}\email{yossib@math.aau.dk}

\affil{\orgdiv{Department of Mathematical Sciences}, \orgname{Aalborg University}, \orgaddress{\street{Fredrik Bajers Vej 7K}, \city{Aalborg {\O}st}, \postcode{9220}, \country{Denmark}}}


\abstract{In this paper, we consider a simple class of stratified spaces -- 2-complexes. We present an algorithm that learns the abstract structure of an embedded 2-complex from a point cloud sampled from it. We use tools and inspiration from computational geometry, algebraic topology, and topological data analysis and prove the correctness of the identified abstract structure under assumptions on the embedding.}




\maketitle
\section{Introduction}

    Recent developments in technology have led to a dramatic increase in the quantity and complexity of data we can collect. These increases require new methods to enable efficient discovery and modelling of the structures underlying them. As the dimension in which we can observe data increases, it becomes more important to be able to reduce the dimensionality of large amounts of data. Some of the difficulties can be addressed by expanding the class of structures we can identify. In \cite{fods-graph}, the authors removed the assumption that the dimension is constant and presented an algorithm for learning the simplest class of stratified spaces -- graphs. A stratified space is a space described by gluing together (manifold) pieces, called strata. There are no restrictions placed upon each stratum's dimension, and the gluing can give rise to a variety of interesting and complex local structures. We extend their work to the identification of the abstract structure underlying a 2-complex. 

    As observed in \cite{fods-graph}, manifold learning can be used to detect and model structures underlying data sets. A variety of approaches and algorithms exist to learn manifold structures from (noisy) samples, see \cite{deymanifoldreconstruction}, \cite{deycurveandsurface}, \cite{deydimensiondetection}. These methods often place assumptions on the manifold and the sampling procedure, generally in the form of restrictions on curvatures, as well as on the density of the sample and the type of noise. The assumptions on curvature are not satisfied by data sets arising in many applications, in particular geospatial data sets arising from person and vehicle movement in transportation networks. We make a second step towards in expanding the set of allowable underlying structures to include stratified spaces of dimension $2$. \cite{stratlearning} focuses on an algorithm to identify when two points have been sampled from the same stratum of a stratified space, but does not present a method for detecting what the dimension of this piece is, or what the global structure is. \cite{nandaintersection} present an algorithm for detecting samples of two intersecting manifolds, which is a first approximation of splitting a space into stratified pieces, and it comes with experimental verification but no theoretical guarantees. In \cite{metricgraph}, the focus is on reconstructing a metric on a graph, with the input consisting of intrinsic distance on the metric graph, the associated theoretical guarantees are about the lengths of the edges in the metric, instead of relating to the geometry of the embedding. In particular, they do not need to consider vertices of degree $2$, as in their setting these are points on an edge. \cite{chazal-sampling} presents a method for sampling and reconstructing compact sets in Euclidean space, with a similar focus on samples with bounded Hausdorff noise and a sufficient density. They guarantee a homotopy equivalence under sufficient conditions but do not present a method for learning the stratified structure. \cite{infering-local-homology} present an algorithm using persistent homology to assess the local homology of a sample at a particular point, using Delauney triangulations, which comes with a great computational cost. 

\subsection{Contribution}
    This paper describes an algorithm for learning the abstract structure underlying an embedded $2$-complex, and provides theoretical guarantees in terms of the geometric embedding that the sample has come from. In particular, the algorithm can be used to learn the number of cells of each dimension, and how they piece together. The output of this algorithm can then be used as a starting point to learn the particular embedding the sample came from. Previous work has focussed on using persistent homology to approximate the local homology at a point, which comes with significant computational overheads. We avoid this by approximating the local homology at each point using a fixed approximation scale, relating to the geometric conditions on the embedding. The algorithm easily works in parallel, which significantly reduces the run time on large data sets. While the algorithm only applies to $2$-complexes, many data sets arising from applications are $2$-dimensional and it provides a foundation for further developments to increase/remove the dimensionality assumption. We acknowledge that from certain perspectives, moving from graphs to $2$-complexes is a small step, yet there are many technical and geometric details involved in guaranteeing the accuracy of the structure learnt even for $2$-complexes, and this is the limiting factor for removing the dimensionality assumption at this stage.

    This article begins with \Cref{sec:defs-notations}, containing definitions of the main objects and tools we use throughout the article. After this, \Cref{sec:geom} consists of geometric lemmas used in \Cref{sec:local}, which considers the local geometry and topology we use to partition the sample $P$. Finally, \Cref{sec:alg} presents algorithms for recovering the abstract structure. \Cref{sec:alg} contains a sequence of lemmas (\Cref{lem:lmedge-disconnected,lem:lmedge-connected,lem:lmtriangle-disconnected,lem:lmtriangle-share-1-vertex,lem:lmtriangle-share-1-edge-with-vertices,lem:lmtriangle-share-2-vertices,lem:share-edge-with-vertex-and-vertex,lem:share-edge-with-vertices-and-vertex,lem:share-three-vertex,lem:share-edge-with-vertex-and-2-vertices,lem:share-edge-with-2-vertices,lem:share-2-edge-with-vertex-and-vertex,lem:share-edge-vertex-vertex-vertex,lem:share-edge-with-vertex-edge-vertex-vertex,lem:share-edge-vertex-vertex-vertex,lem:share-edge-edge-vertex-vertex-vertex,lem:share-edge-edge-edge-vertex-vertex-vertex}), which cover cases used in \Cref{thm:recover}, also known as the `Big Theorem' of this article.

   \section{Definitions and Notations}\label{sec:defs-notations}

    We begin with some definitions and notations we use throughout this article. We begin with the following definition of complex, following Definition 2.4 \cite{carlsson_tda_pattern}. 

    \begin{definition}[Abstract Complex, Definition 2.4 \cite{carlsson_tda_pattern}]
        An \emph{abstract simplicial complex} $X$ consists of a pair $(V(X), \Sigma(X))$, with $V(X)$ a finite set, and $\Sigma(X)$ a subset of the power set of $V(X)$, such for all $\sigma \in \Sigma(X)$ and $\emptyset \neq \tau \subseteq \sigma$, we have $\tau \in \Sigma(X)$. We call $V(X)$ the vertices, and $\Sigma(X)$ the \emph{simplices} of $X$.
    \end{definition}

    For ease of notation and to avoid confusion later in this paper, we will use the following specialised definition for abstract simplicial complexes with top dimension 2.
    
    \begin{definition}[Abstract $2$-Complex]\label{def:complex}
        An \emph{abstract $2$-complex} $X$ consists of 
        \begin{enumerate}
            \item a set $V = V(X)$ of vertices, 
            \item a set $E = \{ \, \sigma \in \Sigma(X) \, | \, \sigma \text{ contains 2 unique elements} \}$ of edges, 
            \item a set $T = \{ \, \sigma \in \Sigma(X) \, | \, \sigma \text{ contains 3 unique elements} \}$ of triangles, 
        \end{enumerate}
        and an incidence operator $\mathcal{I}$, which acts as follows: for any pair of cells $\sigma, \tau \in X$

        \[
            \mathcal{I} \left(\sigma, \tau \right) = \begin{cases} 1 & \text{ if } \sigma \subsetneq \tau \\ 0 & \text{ otherwise} \end{cases}
        \]
    \end{definition}

    We restrict ourselves to linear embeddings of $2$-complexes $X$ in $\mathbb{R}^n$ for some $n \geq 3$.

    \begin{definition}[Linear embedding of  $2$-complex]\label{def:linearembedding}
        Fix $n \geq 3$, then a linear embedding of a $2$-complex $X$ in $\mathbb{R}^n$, $(X, \Theta)$, consists of an abstract $2$-complex $X$ and a map 
        \[
            \Theta: X \to \mathbb{R}^n
        \] 
        
         such that 

         \begin{enumerate}
             \item on vertices $v \in V$, $\Theta$ is injective, 
             \item on edges $\{u,v\} \in E$, $\Theta$ is defined by linear interpolation on $\Theta(u)$ and $\Theta(v)$: $\Theta(\{\, u, v \, \}) = \overline{uv}$ is the line segment between $\Theta(u)$ and $\Theta(v)$,
             \item on triangles $\{\, u, v, w \, \} \in E$, $\Theta$ is defined by linear interpolation on $\Theta(u)$, $\Theta(v)$ and $\Theta(w)$:
             $\Theta(\{\, u, v, w \, \}) = \triangle u v w$ is the triangle with vertices $\Theta(u), \, \Theta(v)$ and $\Theta(w)$, and $\Theta(u), \, \Theta(v), \, \Theta(w)$ are no co-linear, 
            \item for any two cells $\sigma, \tau$ of $X$, we have $\Theta(\sigma) \cap \Theta(\tau) = \Theta(\sigma \cap \tau)$.
            
         \end{enumerate}

        We restrict our attention to embedded  $2$-complexes $|X|_{\Theta}$ such that
         \begin{enumerate}[resume] 
            \item if a vertex $v$ is in the boundary of precisely two edges $\{v,u_1\}$ and $\{v,u_2\}$, then $\angle u_1 v u_2 \neq \pi$,
            \item if an edge $\{v_0,v_1\}$ is in the boundary of precisely two triangles $\{v_0, v_1, u_1\}$ and $\{v_0, v_1 , u_2\}$, then $v_0, v_1, u_1, u_2$ are not co-planar.
         \end{enumerate}

         We denote the image of $\Theta$ in $\mathbb{R}^n$ by $|X|_{\Theta}$.
    \end{definition}

    We often talk about the \emph{boundary} of a cell. 

    \begin{definition}[Cell boundary]
        Let $X$ be an abstract $2$-complex, and take a cell $\sigma \in X$. Then the \emph{boundary} of $\tau$, $\partial \tau$, consists of the cells $\sigma \in X$ such that $\mathcal{I}(\sigma, \tau) = 1$.
    \end{definition}

    An important property of a cell $\sigma \in X$, is whether it is \emph{locally maximal} or not. 

    \begin{definition}[Locally maximal cell]\label{def:locallymaximal}
        Let $\sigma$ be a cell in a  $2$-complex. We say $\sigma$ is \emph{locally maximal} if there is no cell $\tau \in X, \tau \neq \sigma$ with $\sigma \subset \tau$. That is, there is no cell $\tau$ with $\sigma$ in the boundary of $\tau$.
    \end{definition}

    \begin{remark}
        Consider two cells $\sigma, \tau$ in a complex $X$, we say \emph{$\sigma$ is a face of $\tau$} if $\sigma$ is in the boundary of $\tau$, and we say \emph{$\sigma$ is a co-face of $\tau$} if $\tau$ is in the boundary of $\sigma$.
    \end{remark}

    We can represent the incidence relationships of cells in $X$ in a weighted graph $B$.

    \begin{definition}[Incidence graph]\label{def:incidencegraph}
        Take an abstract $2$-complex $X$. The \emph{incidence graph} $B$ of $X$ is the weighted graph with
        \begin{enumerate}
            \item a weight $0$ node $n_v$ for each vertex $v$ of $X$,
            \item a weight $1$ node $n_e$ for each edge $e=\{u,v\}$ of $X$,
            \item a weight $2$ node $n_t$ for each triangle $t=\{u,v,w\}$ of $X$,
            \item an edge between a weight $2$ node $n_t$ and weight $1$ node $n_e$ if $e \subset t$,
            \item an edge between a weight $2$ node $n_t$ and weight $0$ node $n_v$ if $v \in t$,
            \item an edge between a weight $1$ node $n_e$ and weight $0$ node $n_v$ if $v \in e$.
        \end{enumerate}
    \end{definition}
    
    Abusing notation, we usually write $|X|$ instead of $|X|_{\Theta}$ or $(X, \Theta)$, use $v$ to denote both the abstract vertex and its embedded location $\Theta(v)$, $\overline{u v}$ to denote both the abstract edge and the embedded image $\Theta(\{ \, u, v \, \})$, and $\triangle u v w$ to denote both the abstract triangle and the embedded image $\Theta( \{\, u, v, w \, \})$. Whether we are referring to an element of the abstract  $2$-complex or its image in $\mathbb{R}^n$ should be clear from the context.

    We use the following conventions in this article. Given two points $x, y \in \mathbb{R}^n$, $\lVert x - y \rVert$ is the standard Euclidean distance between $x$ and $y$, for a point $x \in \mathbb{R}^n$ and a set $Y \subset \mathbb{R}^n$, we set 

    \[
        d(x,Y ) := \inf_{y \in Y} \lVert x - y \rVert, 
    \]

    and for two sets $X, Y \subset \mathbb{R}^n$, we set
   
    \begin{align*}
        d (X,Y ) &:= \min \left\{ \inf_{x \in X} d(x, Y) ,  \inf_{y \in Y} d(y, X) \right \}, \\ 
        d_H(X,Y ) &:= \max \left\{ \sup_{x \in X} d(x, Y) ,  \sup_{y \in Y} d(y, X) \right \}.
    \end{align*}
    
    We also consider thickenings of a subset $X$: we let 
    
    \begin{equation*}
        X^{\alpha}:= \{ p \in \mathbb{R}^n \mid d(p,H) \leq \alpha\}.
    \end{equation*} 

    In proofs towards the end of this article, we use the \emph{weak feature size} of $X$ to allow us to construct isomorphism, which was introduced in \cite{chazal_lieutier_wfs} as the infimum of the positive critical values of the distance function of $X$.
    
    At various moments in the algorithm, we consider the \emph{diameter} of a set of points $X$. The \emph{diameter of $X$}, $\mathcal{D}(X)$, is the maximum distance between any pair of points $x, y \in X$:
        \[
            \mathcal{D}(X) := \operatorname{max}_{x, y \in X} \lVert x - y \rVert.
        \]
        
    We use $B_r(p)$ to denote the ball of radius $r$ centred at a point $p \in \mathbb{R}^n$, by $\partial B_R(p)$ we mean the boundary of such a ball, and let \[ \mathbb{S}^k = \{ x \in \mathbb{R}^n \mid \lVert x \rVert = 1 \} \]  denote the standard $k$-sphere. We also regularly consider points in a \emph{spherical shell}.

     \begin{definition}
        Fix $a < b$, and let $y$ be a point in $\mathbb{R}^n$. The \emph{spherical shell of radii $a$ and $b$} centered at $p$, $S_{a}^{b}(p)$ is the set 
        \[
            \left\{ q \in \mathbb{R}^n \mid a \leq \lVert q - p \rVert \leq b \right\}.
        \]
    \end{definition}

    We consider dihedral angles between two half-planes.

    \begin{definition}\label{def:dihedral}
        Let $H_1, H_2$ be two half-planes with a common boundary line $L$. Then, the \emph{dihedral angle} $\alpha$ between $H_1$ and $H_2$ is the angle formed by two vectors $v_1 \in H_1$ and $v_2 \in H_2$ originating from the same point $x \in L$ such that both $v_1$ and $v_2$ are perpendicular to $L$.
    \end{definition}
    We work with $\varepsilon$-samples $P$ of the embedded  $2$-complex $|X|$.

    \begin{definition}[$\varepsilon$-sample]\label{def:esample}
        Let $|X| \subset \mathbb{R}^n$ be an embedded  $2$-complex. An \emph{$\varepsilon$-sample} $P$ of $|X|$ is a finite subset of $\mathbb{R}^n$ such that $d_{H} (|X|, P) \leq \varepsilon$.
    \end{definition}

    In \cite{fods-graph} the authors use the threshold graph on a set of points, which we will also use.
    
    \begin{definition}[Threshold graph, Definition 3.1 \cite{fods-graph}]\label{def:threshold}
		Let $P \subset \mathbb{R}^N$ be a finite collection of points, and fix $r> 0$. The \emph{graph at threshold $r$ on $P$}, $\mathfrak{G}_r(P)$, is the graph with vertices $p \in P$, and edges $(p,q)$ if $\|p-q\| \leq r$.
	\end{definition}
    
    The objects we consider in this article are $2$-dimensional, and so we also use \emph{\v{C}ech} complexes.
    
    \begin{definition}[\v{C}ech Complex]\label{def:CC}
        Let $P \subset \mathbb{R}^n$ be a finite set of points. The \emph{\v{C}ech complex at scale $\delta$}, $\check{\mathcal{C}}_{\delta}(P)$ is the complex with j-cells $\{v_i\}_{i=0}^{j}$ such that the intersection $\bigcap_{i=0}^j B_\delta(v_i)$ is non-empty.
    \end{definition}

    Now, we formalise the aim of this article. Given an $\varepsilon$-sample $P$ of some linearly embedded  $2$-complex $|X|$, we want to recover the abstract  structure of the $2$-complex $X$. To do this, we need to learn the number of vertices, the number of edges, and the number of triangles, as well as the incidence relations between them. We achieve this by first deciding for each $p \in P$ if it is near a cell that is not locally maximal, or far away from all cells which are not locally maximal. This partitions $P$ into two subsets which intuitively are $P_{NLM}$ containing samples $p$ near non-locally maximal cells, and $P_{LM}$ containing samples $p$ only near locally maximal cells. Rigorous definitions of $P_{NLM}$ and $P_{LM}$ are in \Cref{def:PNMPLM}. Part of this process involves approximating the local homology at each $p \in P$ using a radius $r$. This requires a choice of scale at which to approximate $|X|$ from $P$. Unlike in \cite{fods-graph}, the relationship between clusters in $P_{NLM}$ and $P_{LM}$ to vertices, edges and triangles is not direct. We can, however, still infer the incidence operator. 

    \begin{remark}
        In this paper, we use \emph{local homology} in very restrictive settings. It is a very generally construction: for a space $X$, the local homology of $X$ at a point $x \in X$ is the relative homology $H\left( X, X \setminus \{ x \} \right)$.
    \end{remark}
    
\section{Geometric Lemmas}\label{sec:geom}
    We provide some geometric lemmas as motivation for the definitions of local structures and the geometric assumptions we place on the embeddings of a  $2$-complex. There are two parts to the definition of the local structure of a point cloud $P$ at a sample $p$: the first is a topological condition relating to the homology of the samples in a spherical shell around $p$, and the second relates to the geometry of these samples. The geometric lemmas in this section allow us to distinguish between points near cells that are not locally maximal and those that are only near locally maximal cells when the topological structure of $P$ at $p$ does not, see \Cref{sec:local}. The proofs of the lemmas in this section can be found in Appendix \ref{sec:geom-lems-proofs}.

    We begin with a helpful lemma that bounds the distance between a point in a spherical shell within $\varepsilon$ of a ray and the point in the ray in the middle of the shell.

     \begin{lemma}\label{lem:sample-sphere-distance}
        Let $L \subset \mathbb{R}^n$ be a ray originating at a point $z$, and fix \[ R\geq 14\varepsilon >0. \] Let $P \subset \mathbb{R}^n$ have $d_H(P,L) \leq \varepsilon$ and take $p \in \mathbb{R}^n$ with \[\lVert p - z \rVert \leq \frac{R}{2}.\] Let $x$ be the point in $L$ with $\lVert x- p \rVert = R$. Then for all $q  \in S_{R-\varepsilon}^{R+\varepsilon}(p) \cap P$ \[ \lVert q - x \rVert \leq \sqrt{2} \varepsilon.\]
    \end{lemma}

    Next, \Cref{lem:antipodal-flat}, which motivates part 3 in \Cref{def:maxstruct}. The lemma considers the distances between triples of points in $S_{R-\varepsilon}^{R+\varepsilon}(p) \cap H^{\varepsilon}$ for some point $p \in H^{\varepsilon}$, where $H^{\varepsilon}$ is the thickening of a plane $H$ by $\varepsilon$, with $\varepsilon > 0$.
    
    \newcommand{\xZxO}{2R^2\left(1+\sqrt{1 - \frac{\left( \frac{\sqrt{3}R}{2} - \varepsilon - \qxUpper \right)^2}{R^2}}\right)-2\varepsilon^2}
    \newcommand{\qxUpperSquare}{\varepsilon^2 + \left(R - \sqrt{R^2-2R\varepsilon}\right)^2}
    \newcommand{\qxUpper}{\sqrt{\varepsilon^2 + \left(R - \sqrt{R^2-2R\varepsilon}\right)^2}}
    \newcommand{\qtwo}{2\sqrt{(R^2-\varepsilon^2)\left(1-\frac{\left(-2R^2+\left(\sqrt{3(R^2-\varepsilon^2)}-(2+2\sqrt{2}) \varepsilon\right)^2\right)^2}{4(R^2-\varepsilon^2)}\right)}}
    
    \begin{lemma}\label{lem:antipodal-flat}
        Consider an affine $2$-hyperplane $H \subset \mathbb{R}^n$ and fix \[ R \geq 14 \varepsilon \geq 0. \] Let $P \subset \mathbb{R}^n$ be such that $d_H(P,H) \leq \varepsilon$, and take $p$ with $d(p, H) \leq \varepsilon$. Then, for all $q_1 \in S_{R-\varepsilon}^{R+\varepsilon}(p) \cap P$, there exists $q_2 \in S_{R-\varepsilon}^{R+\varepsilon}(p) \cap P$ with 
        \begin{align*}
            \lVert q_2 - q_1 \rVert &\geq 2\sqrt{R^2-\varepsilon^2} - (1+ \sqrt{2})\varepsilon.
        \end{align*}
    \end{lemma}
    
    Now that we have a geometric property to test if a point $p$ and the samples in $S_{R-\varepsilon}^{R+\varepsilon}(p)$ are from a subset of a plane. We want to understand what conditions need to be placed on points near an edge in two triangles to guarantee this property does not hold. In particular,  \Cref{lem:antipodal-not-flat} motivates part 4 of \Cref{def:notmaxstruct}. 
    
    \newcommand{\AngleBoundTriangles}{\arccos\left(\frac{(R+2\varepsilon)^2 + \left( \frac{3R}{2} -\varepsilon\right)^2 - \left( 2 \sqrt{R^2 - \varepsilon^2} - \left( 2 +   2\sqrt{2}\right) \varepsilon \right)^2}{2(R+2\varepsilon) \left( \frac{3R}{2} -\varepsilon \right)}\right)}
    
    For ease of reading, we let 
     \begin{dmath*}
        \Psi(\varepsilon, R) = \AngleBoundTriangles.
    \end{dmath*}
    The following lemma motivates the conditions we place on the dihedral angle between two triangles with a common boundary edge $\overline{uv}$ (of degree 2). This allows us to guarantee that the geometry of the samples in $S_{R - \varepsilon}^{R + \varepsilon}(p)$ for a sample $p$ \emph{near} $\overline{uv}$ is not the same as the geometry of samples in $S_{R - \varepsilon}^{R + \varepsilon}(p)$ when $p$ is near a triangle but \emph{far away} from its boundary.

    \begin{lemma}\label{lem:antipodal-not-flat}
        Consider two affine $2$-half-planes $H_1, H_2 \subset \mathbb{R}^n$ whose boundaries are equal, say $L$, and fix $R \geq 14 \varepsilon > 0$. Let $\alpha$ be the dihedral angle between $H_1$ and $H_2$. Let $P$ be a set of points such that $d_H(P, H_1 \cup H_2) \leq \varepsilon$. Further, take $p$ such that $d(p, H_1) \leq \varepsilon$. If \[d(L,p) \leq \frac{R}{2} - 2\varepsilon\] and \[\alpha \in \left( 0, \Psi \left(\varepsilon,R\right)\right)\] then there exist $q_1 \in S_{R-\varepsilon}^{R+\varepsilon}(p) \cap P$ such that for all $q_2 \in S_{R-\varepsilon}^{R+\varepsilon}(p) \cap P$
        \begin{equation*}
            \lVert q_2 - q_1 \rVert < 2\sqrt{R^2- \varepsilon^2} - \left( 1+ \sqrt{2} \right)\varepsilon.
        \end{equation*}        
    \end{lemma}
    
    Next, we investigate the geometry of points near a ray and half-plane, to develop a test for points near not locally maximal cells.
   
    There are several local structures that have the same topological structure: they consist of two connected components with no $1$-cycles. In \cite{fods-graph}, the authors used the angle between the centroids of the two connected components to distinguish between points near a degree 2 vertex and points near the \emph{interior} of an edge. Unfortunately, this test is not sufficient after introducing triangles. If we first check for the presence of triangles, we can again use the inner-product test. To test for the presence of triangles, we examine the diameters of the two connected components. 
    
    So, we first bound the diameter of a set of samples only near a line. 

    \begin{lemma}[Diameter of points near ray]\label{lem:diamline}
        Let $L \subset \mathbb{R}^n$ be a ray originating at a point $z$, and fix $R > 14\varepsilon >0$. Let $P \subset \mathbb{R}^n$ have $d_H(P,L) \leq \varepsilon$ and take $p \in \mathbb{R}^n$ with $d(L,p) \leq \varepsilon$ and $\lVert p - z \rVert \leq \frac{R-\varepsilon}{2}$. Then $\left(S_{R-\varepsilon}^{R+\varepsilon}(p) \cap P \right)^{\frac{3\varepsilon}{2}}$ has $1$ connected component $c$, and the diameter is less than $2\sqrt{2}\varepsilon$.
    \end{lemma}

    The previous lemma bounds the diameter of a connected component containing points with $\varepsilon$ of an edge, that are within $S_{R-\varepsilon}^{R+\varepsilon}(p)$ for a sample $p$ near a vertex in the boundary of this edge. We need to guarantee that if $p$ is near the interior of an edge, it does not fail the diameter test. To ensure this, we obtain the following as a corollary of \Cref{lem:diamline}.
    
    \begin{corollary}\label{cor:diamedge}
        Let $L \subset \mathbb{R}^n$ be a line, and fix $R > 3\varepsilon >0$. Let $P \subset \mathbb{R}^n$ have $d_H(P,L) \leq \varepsilon$ and take $p \in \mathbb{R}^n$ with $d(L,p) \leq \varepsilon$ and \[ \lVert p - z \rVert \leq \frac{R-\varepsilon}{2}. \]  Then $\left(S_{R-\varepsilon}^{R+\varepsilon}(p) \cap P \right)^{\frac{3\varepsilon}{2}}$ has $2$ connected components $c_1, c_2$, and their diameters are less than $2\sqrt{2}\varepsilon$.
    \end{corollary}
    
    \begin{proof}
        First note that $S_{R-\varepsilon}^{R+\varepsilon}(p) \cap L$ consists of two connected components, $C_1, C_2$, and the distance between them is $R-\varepsilon$.  Hence, we can apply \Cref{lem:diamline}, to $C_1$ and $C_2$ individually, obtaining a connected component for each, say $c_1$ and $c_2$. Further, the diameters of $c_1$ and $c_2$ are less than $2\sqrt{2} \varepsilon$.
    \end{proof}
    
   The following lemma guarantees that if there are samples in $S_{R-\varepsilon}^{R+\varepsilon}(p)$ that are within $\varepsilon$ of a triangle, the diameter test fails.
   
    \begin{lemma}\label{lem:diamtri}
        Let $L_1, L_2 \subset \mathbb{R}^n$ be two rays originating at the same point $z$ with the angle $\alpha$ between in the interval 
        \begin{equation*}
            \left [ \frac{\pi}{6} , \pi \right ), 
        \end{equation*}
        and fix $R \geq 14\varepsilon >0$. Let $T$ be the set between $L_1$ and $L_2$. Take $p \in \mathbb{R}^n$ with $d( T, p ) \leq \varepsilon$ and $\lVert p - x \rVert \leq \frac{R-\varepsilon}{2}$, and $P \subset \mathbb{R}^n$ with $d_H(P, T) \leq \varepsilon$. Then, there exist points $q_1, q_2$ in $P$ with $\lVert q_1 -p \rVert, \lVert q_2 - p \rVert \in [R-\varepsilon, R+\varepsilon]$ such that $\lVert q_1 - q_2 \rVert > 2\sqrt{2} \varepsilon$, and $q_1, q_2$ are path connected. Furthermore, the connected component containing $q_1$ and $q_2$ has diameter bigger than $2\sqrt{2} \varepsilon$.
    \end{lemma}

\section{Local Structures}\label{sec:local}

    To identify the abstract structure of the $2$-complex, the algorithm in \Cref{sec:alg} first partitions the sample $P$ into sets $P_{LM}$, containing samples that are only near locally maximal cells, and $P_{NLM}$, containing samples near cells that are not locally maximal. The decision tree for if a point is in $P_{NLM}$ or $P_{LM}$ is summarised in \Cref{fig:flowchart}. After this, we further partition $P_{LM}$ and $P_{NLM}$ to infer the number of cells and their dimensions, as well as the incidence operator. 
    
    Take an embedded  $2$-complex $|X| \subset \mathbb{R}^n$, fix (an appropriate) $0 < \varepsilon \leq R$ and take $p\in \mathbb{R}^n$ with $d(|X|, p) \leq \varepsilon$. Consider the topological and geometric structure of $|X|$ in a neighbourhood of $p$, beginning with $B_{R}(p) \cap |X|$. If $B_{R}(p) \cap |X|$ is disconnected, we restrict to the connected component $C_p$ containing $\operatorname{proj}_{|X|}(p)$. Then, we consider $\partial B_{R}(p) \cap C_p$. Let $\text{proj}_{|X|}(p)$ be the projection of $p$ to $|X|$, and let $\sigma_p$ be the cell containing $\text{proj}_{|X|}(p)$. If $\sigma_p$ is locally maximal and $d(|\partial \sigma_p |, p) > R$, then $\partial B_{R}(p) \cap C_p$ has one of the following structures:
    
    \begin{enumerate}
        \item $\partial B_{R}(p) \cap C_p$ is empty, in which case $\sigma_p$ is a locally maximal vertex,
        \item $\partial B_{R}(p) \cap C_p$ is a pair of antipodal points, in which case $\sigma_p$ is a locally maximal 1-cell,
        \item $\partial B_{R}(p) \cap C_p$ is homotopic to $\mathcal{S}^1$ lying in a plane, in which case $\sigma_p$ is a 2-cell.
    \end{enumerate}

    The above structures consist of two parts: we examine the topological structure of $\partial B_{R}(p) \cap C_p$, and then look at its geometry. If $p$ is within $R$ of some cell $\tau_p$ (possibly $ \tau_p = \sigma_p$) which is not locally maximal, then either the topological structure or the geometric structure is not one of the above cases. As such, we use a two-step process to decide if a given sample $p$ is within $R$ of some not locally maximal cell $\tau_p$: first, we examine the topological structure of $\partial B_{R} (p) \cap C_p$ by looking at its homology, and then if necessary, we consider its geometric structure. We let 
    
    \begin{equation*}
        \mathcal{H}_{\bullet}(p) := H_{\bullet} \left( \partial B_{R}(p) \cap C_p \right).
    \end{equation*}
    
    As we are restricting ourselves to  $2$-complexes, we focus on $\mathcal{H}_{0}(p)$ and $\mathcal{H}_{1}(p)$.

    \begin{definition}[Local homology signature]\label{def:localsig-complex}
        Let $|X| \subset \mathbb{R}^n$ be an embedded $2$-complex, and fix $R>\varepsilon>0$. Take a point $p \in \mathbb{R}^n$ with $d(p, |X|) \leq \varepsilon$. The \emph{local homology signature} of $|X|$ at $p$ is 
        \begin{equation*}
            \sig(p) := \left(|\mathcal{H}_{0}(p)|, |\mathcal{H}_{1}(p)|\right).
        \end{equation*}
    \end{definition}
    
    In the above cases, the local homology signature of $|X|$ at $p$ is as follows.

    \begin{enumerate}
        \item $\sig(p) = (0,0)$,
        \item $\sig(p) = (2,0)$,
        \item $\sig(p) = (1,1)$.
    \end{enumerate}
    and so if $\sig(p)$ is not equal to $(0,0), (2,0)$ or $(1,1)$, then $p$ is within $R$ of a cell $\tau_p$ which is not locally maximal. If $\sig(p)$ is $(0,0)$ then $p$ is within $\varepsilon$ of a degree $0$ vertex. Unfortunately, if $\sig(p)$ is either $(2,0)$ or $(1,1)$, we need to examine the geometric structure of $\partial B_{R}(p) \cap C_p$. When $\sig(p) = (2,0)$, we can distinguish between the case where $\sigma_p$ is a locally maximal 1-cell and where $\sigma_p$ is a vertex of degree $2$ as follows: let the two points in $\partial B_R(p) \cap C_p$ be $c_1$ and $c_2$. If $\sigma_p$ is a 1-cell, then $\angle c_1 p c_2 = \pi$, and other $\angle c_1 p c_2 \neq \pi$. When $\sig(p) = (1,1)$ we need to distinguish between if $\sigma_q$ is a 2-cell, and if $\sigma_p$ is in the boundary of $2$-cells. We can do so by checking if $\partial B_{R}(p) \cap C_p$ is contained in a plane: if it is, then $\sigma_p$ is a $2$-cell, if not $\sigma_p$ is either an edge or a vertex that is not locally maximal. 
    
    Recall that we are working with an $\varepsilon$-sample $P$ of the embedded  $2$-complex $|X|$ instead of $|X|$. We want to approximate $\sig(p)$ with $P$. As $P$ is an $\varepsilon$-sample, we can approximate $\partial B_{R}(p) \cap C_p$ by first considering the structure of $B_{R+\varepsilon}(p) \cap P$, then the structure of $S_{R-\varepsilon}^{R+\varepsilon}(p) \cap P$. Before we define the $(\varepsilon, R)$-local structure of $P$ at $p$ (\Cref{def:localhomsig}), we need the following notation.

    \begin{definition}\label{def:inclusion-rank}
        Let $P \subset \mathbb{R}^n$ be a finite set of points. Then, $\operatorname{rk}_{k}^{\delta,\gamma}(P)$ is the rank of the map on the $k^{th}$ homology groups induced by the inclusion $P^{\delta} \hookrightarrow P^{\gamma}$.
    \end{definition}
    We can now formally define the $(\varepsilon, R)$-local structure of $P$ at $p$.
    
    \begin{definition}[$(\varepsilon, R)$-local homology signature]\label{def:localhomsig}
        Let $P \subset \mathbb{R}^n$ be an $\varepsilon$-sample of an embedded  $2$-complex $|X|$, and fix $R \geq 14 \varepsilon$. Let $C_p^{\frac{3\varepsilon}{2}}$ be samples in the same connected component of threshold graph $\mathfrak{G}  _{3\varepsilon}\left( B_{R+\varepsilon}(p) \cap P \right)$ as $p$. The \emph{$(\varepsilon, R)$-local homology signature} $\sig_{\varepsilon, R}(p)$ of $P$ at a sample $p$ is
        \begin{equation*}
                \sig_{\varepsilon, R}(p) := \left( \operatorname{rk}_{0}^{\frac{3\varepsilon}{2},\frac{7\varepsilon}{2}}\left( S_{R-\varepsilon}^{R+\varepsilon}(p) \cap C^{\frac{3\varepsilon}{2}}_p \right), \operatorname{rk}_{1}^{\frac{3\varepsilon}{2},\frac{7\varepsilon}{2}}\left( S_{R-\varepsilon}^{R+\varepsilon}(p) \cap C^{\frac{3\varepsilon}{2}}_p \right)\right).
        \end{equation*}
    \end{definition}

    We now define the types of local structures, beginning with \emph{maximal} local structures.
    
    \begin{definition}[Maximal $(\varepsilon,R)$-local structure]\label{def:maxstruct}
        Let $P$ be an $\varepsilon$ sample of a linearly embedded  $2$-complex $|X|$ and fix $R \geq 14\varepsilon$. Let $C^{\frac{3 \varepsilon}{2}}_p$ be the set of samples in the same connected component of $\left( B_{R+\varepsilon}(p) \cap P \right)^{\frac{3\varepsilon}{2}}$ as $p$. We say the \emph{$(\varepsilon, R)$-local structure of $P$ at $p$ is maximal} if any of the following hold:
        
        \begin{enumerate}
            \item $\sig_{\varepsilon, R}(p) = (0,0)$, in which case we say that the \emph{$(\varepsilon, R)$-local structure of $P$ at $p$ is maximal of dimension 0},
            \item $\sig_{\varepsilon, R}(p) = (2,0)$, and the two connected components $c_1, c_2$ of $\left( S_{R-\varepsilon}^{R+\varepsilon}(p) \cap C^{\frac{3\varepsilon}{2}}_p \right)^{\frac{3\varepsilon}{2}}$ have diameters less than $5\varepsilon$ and mid-points $q_1$ and $q_2$ such that 
                \begin{equation*}
                    \langle q_1 -p, q_2 - p\rangle \leq -R^2 + 2 R \varepsilon + 7 \varepsilon^2,
                \end{equation*}
                in which case we say that the \emph{$(\varepsilon, R)$-local structure of $P$ at $p$ is maximal of dimension 1},
            \item $\sig_{\varepsilon, R}(p) = (1,1)$, and for all $q_1 \in S_{R-\varepsilon}^{R+\varepsilon}(p) \cap P, \exists q_2 \in S_{R-\varepsilon}^{R+\varepsilon}(p) \cap P$ with 
                \begin{equation*}
                    \lVert q_2 - q_1 \rVert < 2\sqrt{R^2- \varepsilon^2} - \left( 1+ \sqrt{2} \right)\varepsilon.
                \end{equation*}
                 in which case we say that the \emph{$(\varepsilon, R)$-local structure of $P$ at $p$ is maximal of dimension 2},
        \end{enumerate}
    \end{definition}

    Next, we define \emph{not maximal} $(\varepsilon,R)$-local stuctures.
    
    \begin{definition}[Not maximal $(\varepsilon, R)$-local structure]\label{def:notmaxstruct}
         Let $P$ be an $\varepsilon$ sample of a linearly embedded  $2$-complex $|X|$ and fix $R \geq 14\varepsilon$. Let $C^{\frac{3\varepsilon}{2}}_p$ be the set of samples in the same connected component of $\check{\mathcal{C}}_{\frac{3\varepsilon}{2}}\left( S_{R+\varepsilon}(p) \cap P \right)$ as $p$. We say that \emph{the $(\varepsilon, R)$-local structure of $P$ at $p \in P$ is not maximal} if any of the following hold:
        
        \begin{enumerate}
            \item $\sig_{\varepsilon, R}(p) = (n,0)$ for some $n \in \mathbb{Z}_{\geq 0}, \, n \neq 0, 2$,
                
            \item $\sig_{\varepsilon, R}(p) = (1,n)$ for some $n \in \mathbb{Z}_{\geq 0}, \, n \neq 1$,
            \item $\sig_{\varepsilon, R}(p) = (2,0)$ and letting two connected components of 
                \begin{equation*}
                    \left( S_{R - \varepsilon}^{R + \varepsilon}(p) \cap C^{\frac{3\varepsilon}{2}}_p\right)^{\frac{3 \varepsilon}{2}}
                \end{equation*}
                be $c_1, c_2$, either $\operatorname{max}\left\{ \mathcal{D}(c_1), \mathcal{D}(c_2)\right\} \leq2 \sqrt{2}\varepsilon$ and letting mid-points of $c_1 , c_2$ be $q_1, q_2$
                \begin{equation*}
                    \langle q_1 -p, q_2 - p\rangle > -R^2 + 2 R \varepsilon + 7 \varepsilon^2,
                \end{equation*}
           \item $\sig_{\varepsilon, R}(p) = (1,1)$ and there exists $q_1 \in P \cap S_{R-\varepsilon}^{R+\varepsilon}$ such that for all $q_2 \in P \cap S_{R-\varepsilon}^{R+\varepsilon}$
                \begin{equation*}
                    \lVert q_2 - q_1 \rVert < 2\sqrt{R^2- \varepsilon^2} - \left( 1+ \sqrt{2} \right)\varepsilon.
                \end{equation*}
        \end{enumerate}
    \end{definition}

    Having defined the two classes of $(\varepsilon, R)$-local structures, we can define our initial partition.
    
    \begin{definition}[$P_{LM}$ and $P_{NLM}$]\label{def:PNMPLM}
        Let $P$ be an $\varepsilon$-sample of an embedded  $2$-complex $|X|$. We partition $P$ into two sets $P_{LM}$ and $P_{NLM}$ defined as 
        \begin{align*}
            P_{LM} &:= \{ p \in P \mid \text{ the $(\varepsilon, R)$-local structure at of $P$ at $p$ is maximal.} \} \\
            P_{NLM} &:= \{ p \in P \mid\text{ the $(\varepsilon, R)$-local structure of $P$ at $p$ is not maximal.} \} 
        \end{align*}
    \end{definition}
    
    \begin{remark}
        For all $p \in P$, $P$ either has maximal $(\varepsilon, R)$-local structure at $p \in P$ or it does not. Hence, the partitioning of $P$ into $P_{LM}$ and $P_{NLM}$ defined in \Cref{def:PNMPLM} is disjoint.
    \end{remark}

    Recall that the samples we are working with can contain noise, and we use the homology of $\check{\mathcal{C}}_{\frac{3\varepsilon}{2}}\left( S_{R-\varepsilon}^{R+\varepsilon}(p) \cap C^{\frac{3\varepsilon}{2}}_p\right)$ in the definition of $(\varepsilon, R)$-local structure. Hence, we place assumptions on $|X|$ to ensure that we correctly detect when samples are near cells that are not locally maximal. We place assumptions on the distances between any two vertices $u$ and $v$, the distance between an edge $\overline{uw}$ and a vertex $v \neq u,w$, the angle between any pair of edges with a common boundary vertex. Additionally, we place assumptions on the dihedral angle between any two 2-cells which have common boundary components. So that we can infer the incidence operator, we will require an upper bound on the relationship between $R$ and $\varepsilon$, and so we also restrict out choice of $R$ in terms of $\varepsilon$. We use the following notation in the decicision flow chart (\Cref{fig:flowchart}):
    \begin{align*}
            \beta &= -R^2 + 2 R \varepsilon + 7 \varepsilon^2, \\
            \gamma &= 2\sqrt{R^2- \varepsilon^2} - \left( 1+ \sqrt{2} \right)\varepsilon.
    \end{align*}

    \begin{figure}[h]
        \centering
        \begin{tikzpicture}[scale=0.5,every node/.style={scale=0.6},node distance = 4.2cm]
            \node [cloud] (init) {$p \in P$};
            \node [block, above of= init] (00){$\sig_{\varepsilon,R}(p)=(0,0)$};
            \node [block, below left of= init] (20) {$\sig_{\varepsilon,R}(p)=(2,0)$};
            \node [block, below right of= init] (11) {$\sig_{\varepsilon,R}(p)=(1,1)$};
            \node [block, left of= init] (n0) {$\sig_{\varepsilon,R}(p)=(n,0)$};
            \node [block, right of= init] (1n) {$\sig_{\varepsilon,R}(p)=(1,n)$};
            \node [decision, above of = 00] (00lm) {$p$ has a maximal $(\varepsilon, R)$-local structure of dimension $0$};
            \node [decision, above left of = n0] (n0nlm) {$P$ has non-maximal $(\varepsilon, R)$-local structure};
            \node [decision, above right of = 1n] (1nnlm) {$P$ has non-maximal $(\varepsilon, R)$-local structure};
            \node [block, below of = 20] (20diam1) {$\operatorname{max}\{diam(c_1), $ $diam(c_2)\} \leq 2\sqrt{2}  \varepsilon$};
            \node [block, left of = 20] (20diam2) {$\operatorname{max}\{diam(c_1), $ $diam(c_2)\} > 2\sqrt{2} \varepsilon$};
            \node [block, right of = 11] (11geom1) {$\forall q_1 \exists q_2$ $\lVert q_2 - q_1 \rVert \geq \gamma$};
            \node[block, below of = 11] (11geom2) {$\exists q_1 \forall q_2$ $\lVert q_2 - q_1 \rVert < \gamma$};
            \node [block, below of = 20diam1] (20geom1) {$\langle q_1 - p, q_2 - p \rangle > \beta $};
            \node [block, left of = 20diam1] (20geom2) {$\langle q_1 - p, q_2 - p \rangle \leq \beta$};
            \node [left of = 20diam2] (20diam2-1) {};
            \node [decision, below left of = 20diam2] (20nlm) {$p$ has a non-maximal $(\varepsilon, R)$-local structure};
            \node [left of = 20geom1] (20geom1-1) {};
            \node [decision, left of = 20geom1-1] (20lm) {$p$ has a maximal $(\varepsilon, R)$-local structure of dimension 1};
            \node [decision, below of = 11geom2] (11nlm) {$p$ has a non-maximal $(\varepsilon, R)$-local structure};
            \node [right of = 1n] (11geome) {};
            \node [decision, below right of = 11geom1] (11lm) {$p$ has a maximal $(\varepsilon, R)$-local structure of dimension 2};
            \path [line] (init) -- (00);
            \path [line] (init) -- (20);
            \path [line] (init) -- (11);
            \path [line] (init) -- (n0);
            \path [line] (init) -- (1n);
            \path [line] (00) -- (00lm);
            \path [line] (n0) -- (n0nlm);
            \path [line] (1n) -- (1nnlm);
            \path [line] (20) -- (20diam1);
            \path [line] (11) -- (11geom1);
            \path [line] (11) -- (11geom2);
            \path [line] (11geom2) -- (11nlm);
            \path [line] (11geom1) -- (11lm);
            \path [line] (20geom2) -- (20nlm);
            \path [line] (20geom1) -- (20lm);
            \path [line] (20diam1) -- (20geom1);
            \path [line] (20) -- (20diam2);
            \path [line] (20diam2) -- (20nlm);
            \path [line] (20diam1) -- (20geom2);
        \end{tikzpicture}
        \caption{Flow chart for determining if the $(\varepsilon,R)$-local structure of $P$ at $p$ is maximal or not. If maximal, what the dimension is.}\label{fig:flowchart}
    \end{figure}
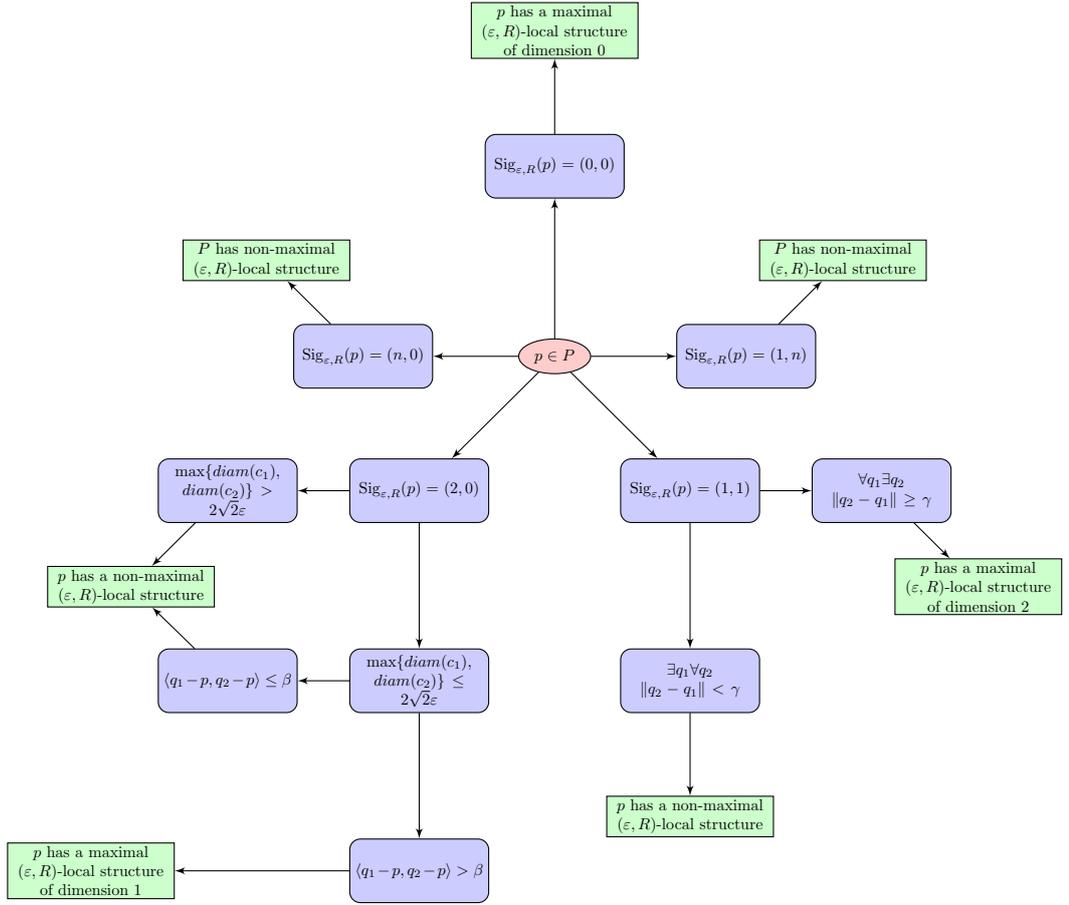
    
    To increase the readability of this article, we define the following functions.
 
    \begin{definition}
        Fix $R > 14 \varepsilon >0$. We define the following functions:
        \begin{enumerate}
            \item 
                \begin{dmath*} 
                    \Psi_1(\varepsilon,R) = \arccos \left ( \frac{\left(\frac{R}{2}-\varepsilon\right)^2 - 18 \varepsilon^2}{\left(\frac{R}{2} - \varepsilon \right)^2}\right) \geq  \arccos\left(\frac{(R-\varepsilon)^2-18\varepsilon^2}{(R-\varepsilon)^2}\right)+2\arcsin\left(\frac{2\varepsilon}{(R-\varepsilon)}\right)
                \end{dmath*}
            \item 
                \begin{dmath*} 
                    \Psi_2(\varepsilon,R) = \pi-\arctan\left(\frac{R+3\varepsilon}{6\varepsilon}\right) +\arcsin\left(\frac{R^2-4R\varepsilon-9\varepsilon^2}{(R+\varepsilon)\sqrt{R^2+6R\varepsilon+34\varepsilon^2}}\right)  
                \end{dmath*}
            \item 
                \begin{align*} 
                    \Psi_3(\varepsilon,R) = \AngleBoundTriangles
                \end{align*}
        \end{enumerate}
    \end{definition}
    
        To improve intuition of these functions, \Cref{fig:Psi_1,fig:Psi_2,fig:Phi_2} provide graphs of them. Note they are effectively a function of $\frac{R}{\varepsilon}$ as they are invariant to scaling both $R$ and $\varepsilon$ by the same amount.
	
	\begin{figure}
	     \includegraphics[width=300pt]{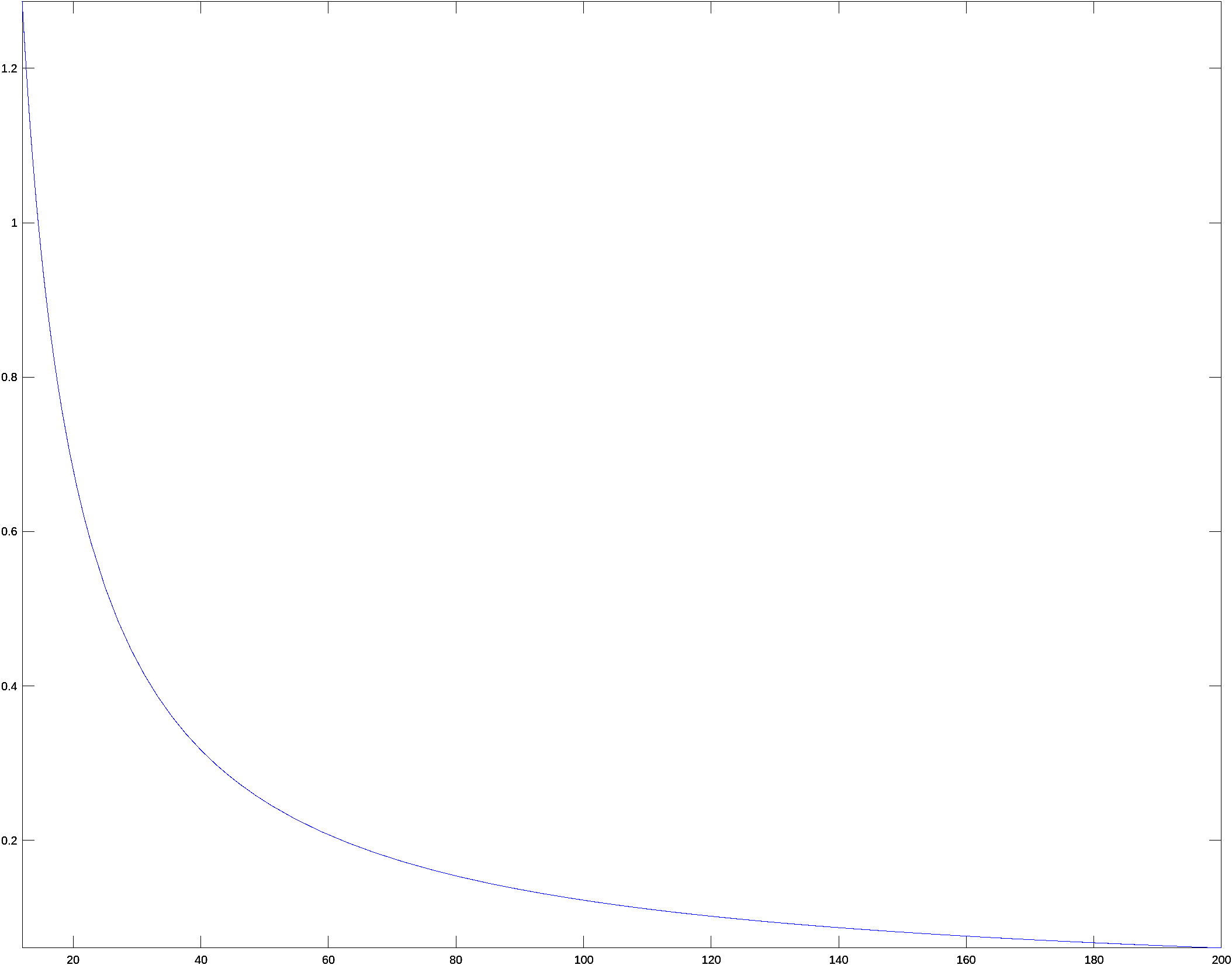}
	     \caption{Graph of $\Psi_1\left(1, \frac{R}{\varepsilon}\right)$.}\label{fig:Psi_1}
	\end{figure}
	
        \begin{figure}
	     \includegraphics[width=300pt]{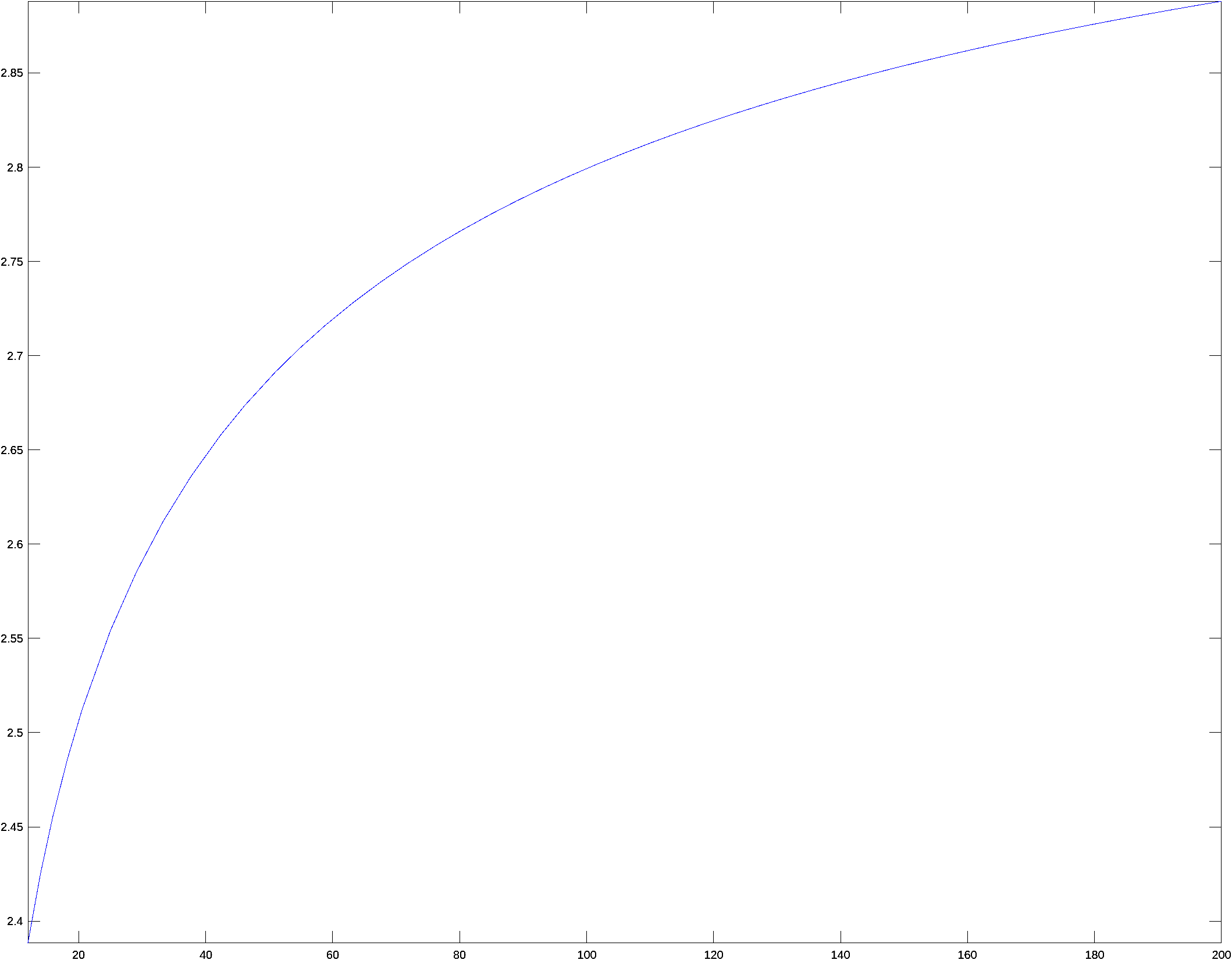}
	      \caption{Graph of $\Psi_2\left(1, \frac{R}{\varepsilon} \right)$.}\label{fig:Psi_2}
	\end{figure}

        \begin{figure}
	     \includegraphics[width=300pt]{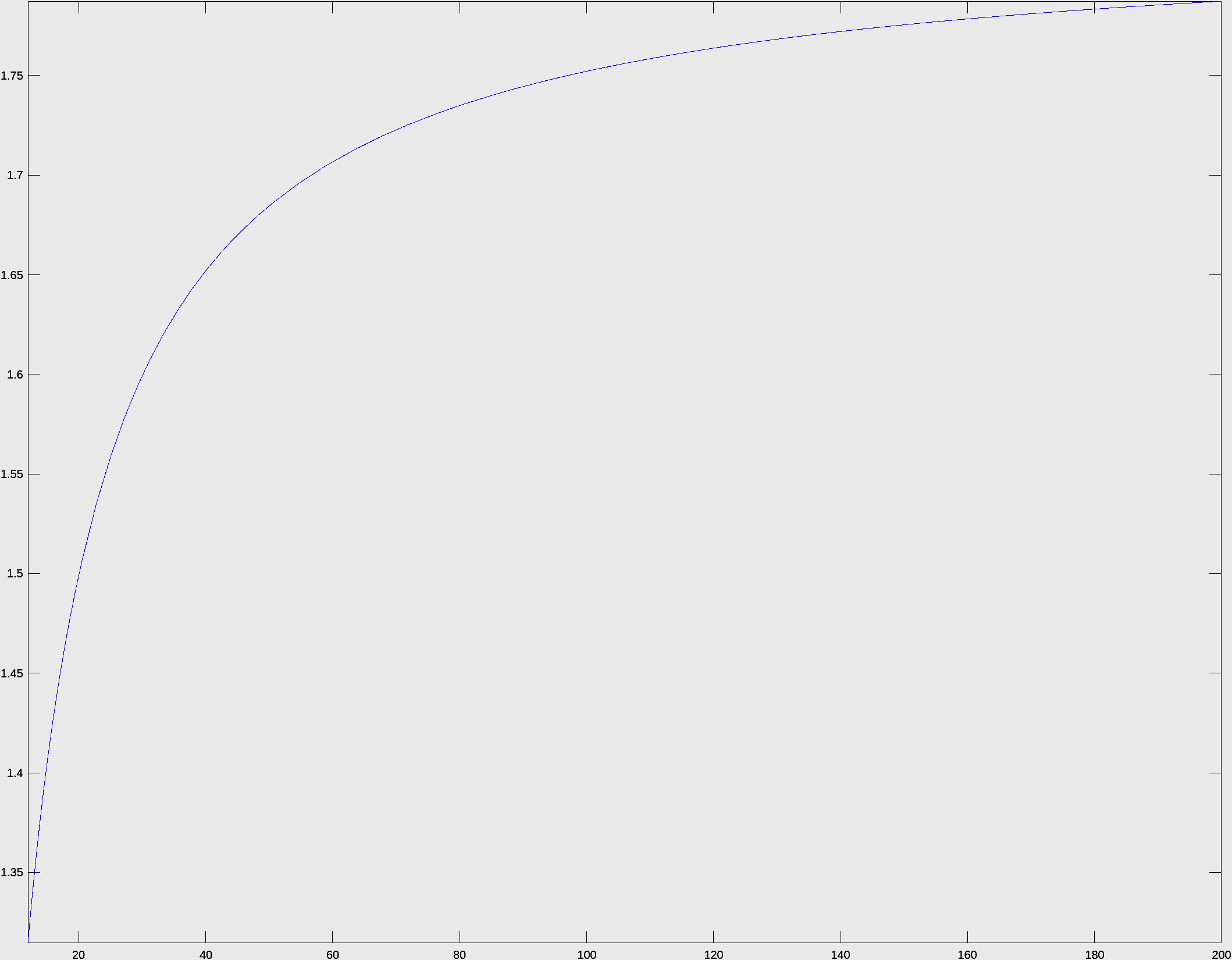}
	     \caption{Graph of $\Psi_3\left(1, \frac{R}{\varepsilon}\right)$.}\label{fig:Phi_2}
	\end{figure}

    We now state the assumptions we place on $|X|$.

    \begin{assumption}\label{as:embed}
        Fix $ R \geq 14 \varepsilon > 0$. We restrict to embedded  2-complexes $|X| = (X, \pi)$ which satisfy the following.
        
        \begin{enumerate}
            \item For all vertices $u,v$, \[\lVert u - v \rVert > 6(R+ \varepsilon).\]
            \item For a vertex $v$ and edge $\overline{uw}$ with $v \neq u,w$, \[d(\overline{uw}, v) > 6(R+\varepsilon).\]
            \item For a vertex $v$ and a triangle $\triangle uwx$ with $v \neq u, w, x$, \[d(\triangle uwx, v) > 6(R+\varepsilon).\]
            \item For an edge $\overline{uv}$ and a triangle $\triangle wxy$ with $v,u \neq w, x, y$, \[d(\triangle wxy, \overline{uv}) > 6(R+\varepsilon).\]
            \item For any triangle $\triangle u v w$, \[ \angle u v w, \angle v w u , \angle w u v \geq \frac{\pi}{6}.\]
            \item For any pair of edges $\overline{uv}, \overline{xy}$ with no common vertex, \[d(\overline{uv}, \overline{xy}) > 6(R+\varepsilon).\]
            \item For any triangles $\triangle uwv, \triangle xyz$, \[ d(\triangle uwv, \triangle xyz) > 6(R+\varepsilon).\]
            \item For any pair of edges $\overline{uv}, \overline{wv}$, \[\angle uvw \geq \Psi_1(\varepsilon, R).\]
            \item For all degree 2 vertices $v$ with edges $\overline{uv}, \overline{wv}$ and no triangle $\triangle u v w$, \[\angle uvw \leq \Psi_2 (\varepsilon, R).\]     
            \item For any pair of triangles $\triangle uvw_1, \triangle uvw_2$, the dihedral angle between them is bounded below by $\Psi_1(\varepsilon, R)$.
            \item For any pair of triangles $\triangle uvw_1, \triangle uvw_2$, with $\overline{uv}$ of degree 2, the dihedral angle between them is bounded above by $\Psi_2(\varepsilon, R)$.
            \item For any triangle $\triangle ww v w_2$ and edge $\overline{uv}$ the angle between $\overline{uv}$ and and ray $L$ in $\triangle w_1 v w_2$  at $v$ is bounded below by $\Psi_1(\varepsilon, R)$ and the radius of the largest circle inscribed by $\triangle u v w$ is at least $2R + 3 \varepsilon$.
            \item For any vertex $v$ such that 
                \begin{equation*}
                    |H_0\left(B_R(v) \cap |X|\right)|=1, \text{ and } |H_1\left(B_R(v) \cap |X|\right)|=1,
                \end{equation*}
                the angle between any two rays $L_1, L_2 \in |X|$ at $v$ is bounded above $\Psi_3(\varepsilon,R)$.
        \end{enumerate}
    \end{assumption}

    \begin{remark}
        The reasons behind some of the conditions in \Cref{as:embed} are relatively clear, while others are a bit more obscure. In particular, the roles of conditions 11 and 12 are not immediately clear. Condition 12 allows us to detect the vertex $v$ in our algorithms. In particular, it is used in \Cref{prop:nlmvertex} show that we obtain $\sig_{\varepsilon, R} = (n, \bullet), \,\ n \geq 2$. Condition 13 allows us to detect which topologically looks similar to an edge of degree 2 or a triangle, and so we place restrictions on the formation of the \emph{cone}, potentially with \emph{fins}, so that we can detect the vertex (\Cref{prop:nlmvertex}). This condition is equivalent to bounding the angle at $v$ of the convex hull which contains the triangles with vertex $v$.
    \end{remark}

    The following Propositions provide us with `regions' near locally maximal $i$-cells $\sigma$ (for $i =0,\, 1, \, 2$), where we can guarantee that at any sample in this region, the $(\varepsilon, R)$-local structure of $P$ at $p$ is maximal of dimension $i$.

    We begin with the region around a locally maximal vertex.
    
    \begin{proposition}\label{prop:lmvertex}
        Let $v$ be a vertex of $|X| \subset \mathbb{R}^n$, which is locally maximal, and let $P$ be an $\varepsilon$-sample of $|X|$. Then, for all $p \in P$ with $\lVert p - v \rVert \leq 4 \varepsilon$, the $(\varepsilon,R)$-local structure of $P$ at $p$ is maximal of dimension $0$.
    \end{proposition}
    
    \begin{proof}
        As $v$ is locally maximal, it is not in the boundary of any other cell, and from \Cref{as:embed} for all vertices $u \neq v$, $\lVert u - v \rVert > 6(R+\varepsilon)$, for all edges $\overline{uw}$ with $v \neq u,w$, 
        \[
            d(\overline{uv}, v) > 6(R+\varepsilon),
        \]
        and for all triangles $\triangle uwx$ with $v \neq u, w, x$, 
        \[
            d(\triangle uwx, v) > 6(R+\varepsilon).
        \]
        
        Hence, any sample $p \in P$ within $4\varepsilon$ of $v$ is within $\varepsilon$ of $v$. Thus, $\left( B_{R+\varepsilon}(p) \cap P\right)^{\frac{3\varepsilon}{2}}$ consists of a single connected component, and $S_{R-\varepsilon}^{R+\varepsilon}(p) \cap P = \emptyset$. 
        
        Thus, $S_{R-\varepsilon}^{R+\varepsilon}(p) \cap P$, $\sig_{\varepsilon, R}(p) = (0,0)$, and the $(\varepsilon, R)$-local structure of $P$ at $p$ is maximal of dimension $0$.
    \end{proof}

    Next, we bound the region near a locally maximal edge.
    
    \begin{proposition}\label{prop:lmedge}
        Let $\overline{uv}$ be an edge of $|X| \subset \mathbb{R}^n$, which is locally maximal, and let $P$ be an $\varepsilon$-sample of $|X|$. Then, for all $p \in P$ with $d(\overline{uv},p) \leq \varepsilon$, and $\lVert p - u \rVert, \lVert p- v \rVert \geq \frac{3R}{2}+\varepsilon$, the $(\varepsilon,R)$-local structure of $P$ at $p$ is maximal of dimension $1$.
    \end{proposition}
    
    \begin{proof}
        By \Cref{as:embed}, for any vertex $w \neq u,v$ 
        \[
            d (\overline{uv}, w) > 6(R+\varepsilon),
        \]
        for any edge $\overline{wx}$, with $w,x \neq u, v$, 
        \[
            d(\overline{uv}, \overline{wx}) > 6(R+\varepsilon),
        \]
        for any triangle $\triangle wxy$, with $w,x,y \neq u, v$, 
        \[
            d(\triangle wxy, \overline{uv}) > 6(R+\varepsilon),
        \]
        and so the connected component $C^{\frac{3\varepsilon}{2}}_p$ of $\left(B_{R+\varepsilon}(p) \cap P\right)^{\frac{3\varepsilon}{2}}$ which contains $p$, contains only points $q \in P$ with $d(q, \overline{uv}) \leq \varepsilon$.
        
        Hence, $\check{\mathcal{C}}_{\frac{3\varepsilon}{2}} \left ( S_{R-\varepsilon}^{R + \varepsilon}(p) \cap C^{\frac{3\varepsilon}{2}}_p \right)$ consists of two connected components, $c_1$ and $c_2$. By \Cref{lem:diamline}, the diameters of $c_1$ and $c_2$ are less than $5\varepsilon$. Let $x_1$ and $x_2$ be the centroids of $c_1$ and $c_2$. Then, applying Lemma 2.1 in \cite{fods-graph}, 
        \[
            \langle x_1 -p, x_2 -p \rangle \leq -R^2 + 2R \varepsilon + 7 \varepsilon^2,        
        \]
        so the $(\varepsilon, R)$-local structure of $P$ at $p$ is maximal of dimension 1.
    \end{proof}

    Finally, we bound the region near (locally maximal) triangles.
    
    \begin{proposition}\label{prop:lmtriangle}
        Let $\triangle uvw$ be an triangle of $|X| \subset \mathbb{R}^n$, and let $P$ be an $\varepsilon$-sample of $|X|$. Then, for all $p \in P$ with $d(\triangle uvw, p) \leq \varepsilon$, and $d(\partial \triangle uvw, p) \geq \frac{3R}{2} + \varepsilon$, the $(\varepsilon,R)$-local structure of $P$ at $p$ is maximal of dimension $2$.
    \end{proposition}
    
    \begin{proof}
        From \Cref{as:embed}, for all triangles $\triangle xyz$, with $x,y,z \neq u,v,w$, 
        \[ 
            d(\triangle uwv, \triangle xyz) > 6(R+\varepsilon),
        \]
        and hence the connected component $C^{\frac{3\varepsilon}{2}}_p$ of $\check{\mathcal{C}}_{\frac{3\varepsilon}{2}}\left( B_{R+\varepsilon}(p) \cap P\right)$ containing $p$, consists only of samples $q \in P$ with $d(q, \triangle uvw) \leq \varepsilon$, as the angle between triangles is bounded below (\Cref{as:embed}). 
        
        First, we need to show that $\sig_{\varepsilon, R}(p) = (1,1)$, after which \Cref{lem:antipodal-flat} implies that for all $q_1 \in S_{R-\varepsilon}^{R+\varepsilon}(p) \cap P$, there exists $q_2 \in S_{R-\varepsilon}^{R+\varepsilon}(p) \cap P$ such that 
        
        \begin{align*}
           \lVert q_2 - q_1 \rVert   & \geq 2\sqrt{R^2 - \varepsilon^2} - (1 + \sqrt{2})\varepsilon.
        \end{align*}
        As $d(\partial \triangle uvw, p) > \frac{3R}{2} + \varepsilon$, we have the following inclusions
        \begin{align*}
            S_{R-\varepsilon}^{R+\varepsilon}(p) \cap \triangle uvw &\hookrightarrow \left( S_{R-\varepsilon}^{R+\varepsilon}(p) \cap P \right)^{\frac{3\varepsilon}{2}} &\\
            &\hookrightarrow  \left( S_{R-\varepsilon}^{R+\varepsilon}(p) \cap \triangle uvw \right )^{\frac{5\varepsilon}{2}} \\&
            \hookrightarrow \left( S_{R-\varepsilon}^{R+\varepsilon}(p) \cap P \right )^{\frac{7\varepsilon}{2}}\\
            &\hookrightarrow \left( S_{R-\varepsilon}^{R+\varepsilon}(p) \cap \triangle uvw \right )^{\frac{9\varepsilon}{2}}.
        \end{align*}
        
        By the bounds in \Cref{as:embed} on the distances between a triangle and cells not in its boundary, the weak feature size of $S_{R-\varepsilon}^{R+\varepsilon}(p) \cap \triangle uvw$ is greater than $5\varepsilon$, and so the inclusion maps induce isomorphisms 
        
        \begin{equation*}
            H_{\bullet}\left( S_{R-\varepsilon}^{R+\varepsilon}(p) \cap \triangle uvw \right) \cong H_{\bullet}\left( \left( S_{R-\varepsilon}^{R+\varepsilon}(p) \cap \triangle uvw \right )^{\frac{5\varepsilon}{2}} \right) \cong H_{\bullet}\left( \left( S_{R-\varepsilon}^{R+\varepsilon}(p) \cap \triangle uvw \right )^{\frac{9\varepsilon}{2}} \right).
        \end{equation*}
        
        The above homology factors through $\left( S_{R-\varepsilon}^{R+\varepsilon}(p) \cap P \right)^{\frac{3\varepsilon}{2}}$ and $\left( S_{R-\varepsilon}^{R+\varepsilon}(p) \cap P \right)^{\frac{7\varepsilon}{2}}$ so we have 
        \begin{align*}
            \operatorname{rk}_{\bullet}^{\frac{3\varepsilon}{2}, \frac{5\varepsilon}{2}}\left(S_{R-\varepsilon}^{R+\varepsilon}(p) \cap P \right) &= \left| H_{\bullet} \left( S_{R-\varepsilon}^{R+\varepsilon}(p) \cap \triangle uvw  \right) \right|,
        \end{align*}
        and as 
        \begin{align*}
            |H_{0}\left( S_{R-\varepsilon}^{R+\varepsilon}(p) \cap \triangle uvw \right)| = 1, \, & |H_{1}\left( S_{R-\varepsilon}^{R+\varepsilon}(p) \cap \triangle uvw \right)| = 1,
        \end{align*}
        it follows that $\sig_{\varepsilon, R}(p) = (1,1)$. Now we apply \Cref{lem:antipodal-flat} and conclude that the $(\varepsilon, R)$-local structure of $P$ at $p$ is maximal of dimension 2. 
    \end{proof}

    Now, we obtain the regions around not locally maximal $i$-cells $\sigma$ ($i = 0,\, 1$) in which we can guarantee that the $(\varepsilon, R)$-local structure of $P$ at a sample $p$ in this region is not locally maximal. Again, we begin with non-locally maximal vertices.

    \begin{remark}
        As we have restricted our considerations to $2$-complexes, every triangle $\sigma$ is locally maximal; hence, we need only to consider vertices and edges that are not locally maximal.
    \end{remark}
    
    \begin{proposition}\label{prop:nlmvertex}
        Let $v$ be a vertex of $|X| \subset \mathbb{R}^n$, which is not locally maximal, and let $P$ be an $\varepsilon$-sample of $|X|$. Then, for all $p \in P$ with 
        \begin{equation*}
            \lVert p - v \rVert \leq \frac{R}{2} -2\varepsilon, 
        \end{equation*}
        the $(\varepsilon,R)$-local structure of $P$ at $p$ is not maximal.
    \end{proposition}
    
    \begin{proof}
        There are several cases we need to consider, which we can classify by the homology of $\partial B_{R}(v) \cap |X|$:
        
        \begin{enumerate}
            \item $\left|H_0\left(\partial B_{R}(v) \cap |X|\right)\right| = n,\, \left|H_1\left(\partial B_{R}(v) \cap |X|\right)\right| = 0, \, n \neq 2$,
            \item $\left|H_0\left(\partial B_{R}(v) \cap |X|\right)\right| = 2,\, \left|H_1\left(\partial B_{R}(v) \cap |X|\right)\right| = 0$,
            \item $\left|H_0\left(\partial B_{R}(v) \cap |X|\right)\right| = 1,\, \left|H_1\left(\partial B_{R}(v) \cap |X|\right)\right| = 1$,
            \item $\left|H_0\left(\partial B_{R}(v) \cap |X|\right)\right| = 1,\, \left|H_1\left(\partial B_{R}(v) \cap |X|\right)\right| = n, \, n \geq 2$.
        \end{enumerate}

        In each of these cases, the following argument holds. Let $C_p$ be the connected component of $B_{R+\varepsilon}(p) \cap |X|$ which contains the projection of $p$ to $|X|$, and let $C^{\frac{3\varepsilon}{2}}_p$ be the connected component of $\left(S_{R-\varepsilon}^{R+\varepsilon}(p) \cap P \right)^{\frac{3\varepsilon}{2}}$. As $P$ is a $\varepsilon$-sample of $|X|$, we have the following inclusions
        \begin{align*}
            S_{R-\varepsilon}^{R+\varepsilon}(p) \cap \triangle uvw &\hookrightarrow \left( S_{R-\varepsilon}^{R+\varepsilon}(p) \cap P \right)^{\frac{3\varepsilon}{2}} &\\
            &\hookrightarrow  \left( S_{R-\varepsilon}^{R+\varepsilon}(p) \cap \triangle uvw \right )^{\frac{5\varepsilon}{2}} \\&
            \hookrightarrow \left( S_{R-\varepsilon}^{R+\varepsilon}(p) \cap P \right )^{\frac{7\varepsilon}{2}}\\
            &\hookrightarrow \left( S_{R-\varepsilon}^{R+\varepsilon}(p) \cap \triangle uvw \right )^{\frac{9\varepsilon}{2}}.
        \end{align*}

        By the bounds in \Cref{as:embed} on 
        \begin{itemize}
            \item the angle betwen edges at a common vertex,
            \item the distance between vertices,
            \item the angles between triangles with a common vertex or edge,
            \item the distance between vertices and cells they do not intersect with,
        \end{itemize}
        the weak feature size of $S_{R-\varepsilon}^{R+\varepsilon}(p) \cap C_p^{\frac{3\varepsilon}{2}}$ is greater than $5\varepsilon$, and we have the following isomorphism on homology induced by the inclusions above

        \begin{equation*}
            H_{\bullet}\left( S_{R-\varepsilon}^{R+\varepsilon}(p) \cap |X| \right) \cong H_{\bullet}\left( \left( S_{R-\varepsilon}^{R+\varepsilon}(p) \cap |X|\right )^{\frac{5\varepsilon}{2}} \right) \cong H_{\bullet}\left( \left( S_{R-\varepsilon}^{R+\varepsilon}(p) \cap |X|\right )^{\frac{9\varepsilon}{2}} \right).
        \end{equation*}
        
        The above homology factors through $\left( S_{R-\varepsilon}^{R+\varepsilon}(p) \cap P \right)^{\frac{3\varepsilon}{2}}$ and $\left( S_{R-\varepsilon}^{R+\varepsilon}(p) \cap P \right)^{\frac{7\varepsilon}{2}}$ so we have 
        \begin{align*}
            \operatorname{rk}_{\bullet}^{\frac{3\varepsilon}{2}, \frac{7\varepsilon}{2}}\left(S_{R-\varepsilon}^{R+\varepsilon}(p) \cap P \right) &= \left| H_{\bullet} \left( S_{R-\varepsilon}^{R+\varepsilon}(p) \cap |X| \right) \right|.
        \end{align*}

        As $\lVert p - v \rVert \leq \frac{R}{2} - 2 \varepsilon$, we have
        \begin{equation*}
            \left| H_{\bullet} \left( S_{R-\varepsilon}^{R+\varepsilon}(p) \cap |X| \right) \right| = \left| H_{\bullet} \left( \partial B_R(v) \cap |X| \right) \right|,
        \end{equation*}
        giving 
        \begin{equation*}
            \operatorname{rk}_{\bullet}^{\frac{3\varepsilon}{2}, \frac{7\varepsilon}{2}}\left(S_{R-\varepsilon}^{R+\varepsilon}(p) \cap P \right) = \left| H_{\bullet} \left( \partial B_R(v) \cap |X| \right) \right|.
        \end{equation*}
        
        \medskip 
        
        \noindent {\bf Case 1:} $\left|H_0\left(\partial B_{R}(v) \cap |X|\right)\right| = n,\, \left|H_1\left(\partial B_{R}(v) \cap |X|\right)\right| = 0, \, n \neq 2$ \\

        By the above, we have $\sig(p) = (n,0),\, n \neq 2$, and so the $(\varepsilon, R)$-local structure of $P$ at $p$ is not maximal.
        
        \medskip 
        
        \noindent {\bf Case 2:} $\left|H_0\left(\partial B_{R}(v) \cap |X|\right)\right| = 2,\, \left|H_1\left(\partial B_{R}(v) \cap |X|\right)\right| = 0$\\
        
        By the above, we have $\sig(p) = (2,0)$. Let $C^{2\varepsilon}_{p}$ be the connected component of $\check{\mathcal{C}}_{\frac{3\varepsilon}{2}}\left( S_{R-\varepsilon}^{R+\varepsilon}(p) \cap P \right)$ containing $p$. 
        
        Assume that $v$ is a face of some triangle $\triangle u v w$. Then by the bounds placed on angles between edges, and distances between edges without a common face, edges and vertices which are not faces, and vertices and triangles they are not a face of (see \Cref{as:embed}), and \Cref{lem:diamtri} at least one connected component in $\left( S_{R-\varepsilon}^{R+\varepsilon}(p) \cap C^{\frac{3\varepsilon}{2}}_p\right)^{\frac{3\varepsilon}{2}}$ has a diameter greater than $2\sqrt{2}\varepsilon$. Thus, the $(\varepsilon, R)$-local structure of $P$ at $p$ is not maximal.
        
        If $v$ is only the face of edges, then by the bounds placed on angles between edges, and distances between edges without a common face, edges and vertices which are not faces, and vertices and triangles they are not a face of (see \Cref{as:embed}), both connected components come from two edges $uv$ and $wv$, Lemma 2.1, in \cite{fods-graph} and \Cref{lem:diamline} give that the $(\varepsilon, R)$-local structure of $P$ at $p$ is not maximal.
        
        \medskip
        
        {\bf Case 3:} $\left|H_0\left(\partial B_{R}(v) \cap |X|\right)\right| = 1\, \left|H_1\left(\partial B_{R}(v) \cap |X|\right)\right| = 1$\\
        
        Again, we have $\sig(p) = (1,1)$ so there are at least three triangles having $v$ as a common vertex. Let $p_X$ be the closest point in $|X|$ to $p$, and let $x_1 \in \partial B_R(p) \cap |X|$ be colinear with $v$ and $p_X$, then there is $q_1 \in S_{R-\varepsilon}^{R+\varepsilon} \cap P$ with $\lVert q_1 - x_1 \rVert \leq \varepsilon$. 
        
        Now take any $ q_2 \in S_{R-\varepsilon}^{R+\varepsilon} \cap P$, and let $x_2$ be the point in $|X| \cap \partial B_R(p)$ closest to $q_2$. Then from \Cref{lem:sample-sphere-distance} \[\lVert q_2 - x_2 \rVert \leq \sqrt{2} \varepsilon.\] Consider the rays $L_1, L_2$ from $v$ through $x_1, x_2$ respectively, and assume $d(p, L_1) \leq \varepsilon$.  
        
        \begin{figure}[h!]
            \centering
            \begin{tikzpicture}[x=1.5pt,y=1.5pt]
                \draw (25,52.5) node [anchor = south] {$x_2$};
                \draw (0,0) node [anchor = south east] {$v$};
                \draw (29.5,-40) node [anchor = north] {$x_1$};
                \draw (29.5,-40) -- (0,0);
                \draw (25,52.5) -- (0,0);
                \draw (25,52.5) -- (29.5,-40);
                \draw (10,-15) node [anchor = east] {$p_X$};
            \end{tikzpicture}
            \caption{$d(q_2, H_2) \leq \varepsilon$}
        \end{figure}
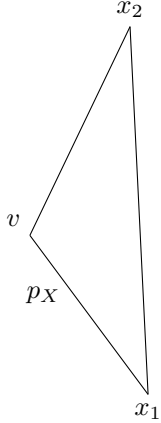
        We have 
        \begin{align*}
            \lVert x_1 - v \rVert & = \lVert x_1 - p_X \rVert + \lVert p_X - v \rVert \leq \frac{3R}{2} - 2 \varepsilon, \\
            \lVert x_2 - v \rVert & \leq R + \varepsilon,
        \end{align*}
        and so 
        \begin{align*}
            \lVert x_2 - x_1 \rVert & = \lVert x_2 - v \rVert^2 + \lVert x_1 - v \rVert^2  - 2 \lVert x_2 - v \rVert \lVert x_1 - v \rVert^2 \cos \angle x_1 v x_2 \\
                                    &\leq \left( \frac{3R}{2} - 2 \varepsilon \right)^2 + (R+\varepsilon)^2 - \left( \frac{3R}{2} - 2 \varepsilon \right) (R+\varepsilon) \cos x_1 v x_2 . \\
        \end{align*}
        By condition 13 in \Cref{as:embed} the angle between them is bounded above by $\Psi_3(\varepsilon,R)$, so 
        \begin{align*}
            \lVert x_2 - x_1 \rVert & \leq 2\sqrt{R^2 - \varepsilon^2} - (1 + \sqrt{2}) \varepsilon,
        \end{align*}
        and so 
        \[ \lVert q_2 - q_1 \rVert \leq 2\sqrt{R^2 - \varepsilon^2} - (2 + 2\sqrt{2}) \varepsilon.\]
        
        Thus, the $(\varepsilon, R)$-local structure of $P$ at $p$ is not maximal.
        
        \medskip
        
        {\bf Case 4:} $\left|H_0\left(\partial B_{R}(v) \cap |X|\right)\right| = 1,\, \left|H_1\left(\partial B_{R}(v) \cap |X|\right)\right| = n, \, n \geq 2$\\

        By the argument at the start of this proof, $\sig(p) = (1, n ), \, n \geq 2$ and so the $(\varepsilon,R)$-local structure of $P$ at $p$ is not maximal.
    \end{proof}

    Next, we bound the region near edges that are not locally maximal.
    
    \begin{proposition}\label{prop:nlmedge}
         Let $\overline{u v}$ be an edge of $|X| \subset \mathbb{R}^n$, which is not locally maximal, and let $P$ be an $\varepsilon$-sample of $|X|$. Then, for all $p \in P$ with $d(\overline{u v}, p)\leq \frac{R}{2} - 2\varepsilon$, the $(\varepsilon,R)$-local structure of $P$ at $p$ is not maximal.
    \end{proposition}
    
    \begin{proof}
        If an edge $\overline{uv}$ is not locally maximal, then there is at least one triangle $\triangle u v w$.
        
        We consider 3 cases:
        
        \begin{enumerate}
            \item there is a unique triangle $\triangle u v w$ with $\overline{u v}$ in the boundary, 
            \item there are exactly two triangles $\triangle u v w_1$ and $\triangle u v w_2$ with $\overline{u v}$ in their boundaries, 
            \item there are three or more triangles $\triangle u v w_1, \triangle u v w_2$ and $\triangle u v w_3$ with $\overline{u v}$ in their boundaries. 
        \end{enumerate}

        Recall that we restrict our attention to the connected components $C_p, C_p^{\frac{3\varepsilon}{2}}$ of $S_{R-\varepsilon}^{R+\varepsilon}(p) \cap |X|$ and $\left(S_{R-\varepsilon}^{R+\varepsilon}(p) \cap P\right)^{\frac{3\varepsilon}{2}}$ which contains $p$.
        
        By the bounds in \Cref{as:embed} on 
        \begin{itemize}
            \item the angle betwen edges at a common vertex,
            \item the distance between edges that do not have a common face,
            \item the angles between triangles with a common edge,
            \item the distance between edges and cells they do not intersect with,
        \end{itemize}
        the weak feature size of $C_p$ is greater than $5\varepsilon$. Hence by the same argument as at the start of the poof of \Cref{prop:nlmvertex},
        
        \begin{equation*}
            \sig_{\varepsilon, R}(p) = \left(\left|H_0\left(\partial B_{R}(m) \cap |X|\right)\right| ,\, \left|H_1\left(\partial B_{R}(m) \cap |X|\right)\right| \right).
        \end{equation*}

        Thus, in cases 1 and 3, we get $\sig(p) = (1,0)$ and $\sig(p) = (1,n)$ for $n \geq 3$ respectively.

        In case 2, we get $\sig(p) = (1,1)$, and so need to check the geometric condition. By \Cref{lem:antipodal-not-flat}, there is a $q_1 \in S_{R-\varepsilon}^{R+\varepsilon}(p) \cap P$ such that for all $q_2 \in S_{R-\varepsilon}^{R+\varepsilon}(p) \cap P$ 
        \begin{equation*}
            \lVert q_2 - q_1 \rVert < 2\sqrt{R^2- \varepsilon^2} - (1 + \sqrt{2})\varepsilon, 
        \end{equation*} 
        and so the $(\varepsilon, R)$-local structure of $P$ at $p$ is not maximal.

        Hence, in all 3 cases, the $(\varepsilon, R)$ local structure of $P$ at $p$ is not maximal. 
    \end{proof}
    
\section{2-Complex Algorithm and Correctness}\label{sec:alg}

    In this section, we present a set of algorithms, which together, recover the structure of $X$ from an $\varepsilon$-sample $P$ of an embedding $(X, \Theta) \subset \mathbb{R}^n$. \Cref{thm:recover} states that given an $\varepsilon$-sample $P$ of an embedded $2$ complex $|X|= (X, \Theta_X) \subset \mathbb{R}^n$ satisfying \Cref{as:embed}, we can recover the structure of $X$ using this algorithm. There is a sequence of lemmas (\Cref{lem:lmedge-disconnected,lem:lmedge-connected,lem:lmtriangle-disconnected,lem:lmtriangle-share-1-vertex,lem:lmtriangle-share-1-edge-with-vertices,lem:lmtriangle-share-2-vertices,lem:share-edge-with-vertex-and-vertex,lem:share-edge-with-vertices-and-vertex,lem:share-three-vertex,lem:share-edge-with-vertex-and-2-vertices,lem:share-edge-with-2-vertices,lem:share-2-edge-with-vertex-and-vertex,lem:share-edge-vertex-vertex-vertex,lem:share-edge-with-vertex-edge-vertex-vertex,lem:share-edge-vertex-vertex-vertex,lem:share-edge-edge-vertex-vertex-vertex,lem:share-edge-edge-edge-vertex-vertex-vertex}), which culminates in the `big theorem' (\Cref{thm:recover}). The proofs of the lemmas are in Appendix \ref{sec:correctness-lems-proofs}.
    
    The algorithm partitions $P$ into $P_{LM}$ and $P_{NLM}$, such that for each $p \in P_{LM}$ the $(\varepsilon, R)$-local structure of $P$ at $p$ is maximal, and for each $p \in P_{NLM}$ the $(\varepsilon, R)$-local structure of $P$ at $p$ is not maximal. We then detect the number of vertices, the number of edges, the number of triangles and the incidence operator. To obtain $P_{LM}$ and $P_{NLM}$, we use  \[ \Delta_{\varepsilon, R}: P \to \{0,1\}, \] see \Cref{alg:localstruct}.
    
    Let $\mathcal{C}_p$ be the samples $q \in P$ in the connected component containing $p$ in the threshold graph 
    \begin{equation*}
        \mathcal{G}_p = \mathfrak{G}_{3\varepsilon}\left(B_{R+\varepsilon}(p) \cap P \right)
    \end{equation*}    
    with $\lVert q - p \rVert \in [R- \varepsilon, R+ \varepsilon]$.
    In the definitions of $(\varepsilon, R)$-local structure (\Cref{def:maxstruct,def:notmaxstruct}), we used 
    \begin{equation*}
        \operatorname{rk}_{\bullet}^{\frac{3\varepsilon}{2}, \frac{7\varepsilon}{2}}\left( S_{R- \varepsilon}^{R+\varepsilon}(p) \cap P \right),
    \end{equation*}
    which by the Nerve Lemma (Corollary 4G.3 \cite{hatcher}) is equal to the rank, $\mathcal{RK}_{\bullet}$, of the map
    \begin{equation*}
        H_{\bullet} \left( \check{\mathcal{C}}_{\frac{3 \varepsilon}{2}}\left( S_{R-\varepsilon}^{R+\varepsilon}(p) \cap \mathcal{C}_p\right)\right) \rightarrow H_{\bullet} \left( \check{\mathcal{C}}_{\frac{7 \varepsilon}{2}}\left( S_{R-\varepsilon}^{R+\varepsilon}(p) \cap \mathcal{C}_p\right)\right) 
    \end{equation*}
    induced by the inclusion
    \begin{equation*}
        \check{\mathcal{C}}_{\frac{3 \varepsilon}{2}}\left( S_{R-\varepsilon}^{R+\varepsilon}(p) \cap P\right)\hookrightarrow \check{\mathcal{C}}_{\frac{7 \varepsilon}{2}}\left( S_{R-\varepsilon}^{R+\varepsilon}(p) \cap P\right).
    \end{equation*}
    Hence, $\Delta_{\varepsilon, R}(p)$ returns $0$ if the $(\varepsilon,R)$-local structure of $P$ at $P$ is not maximal, and returns $1$ if it is maximal. Then,  \[ P_{NLM}= \Delta_{\varepsilon, R}^{-1}(0) \] and \[P_{NLM} = \Delta_{\varepsilon, R}^{-1}(1).\]

    \begin{remark}
        We can appeal to the Nerve Lemma, as the balls used in the construction of $\check{\mathcal{C}}_{\frac{3 \varepsilon}{2}}\left( S_{R-\varepsilon}^{R+\varepsilon}(p) \cap \mathcal{C}_p\right)$ and $\check{\mathcal{C}}_{\frac{7 \varepsilon}{2}}\left( S_{R-\varepsilon}^{R+\varepsilon}(p) \cap \mathcal{C}_p\right)$ lead us to \emph{good covers} of $\left( S_{R- \varepsilon}^{R+\varepsilon}(p) \cap P \right)^{\frac{3\varepsilon}{2}}$ and $\left( S_{R- \varepsilon}^{R+\varepsilon}(p) \cap P \right)^{\frac{7\varepsilon}{2}}$ respectively. To see that these covers satisfy the `every non-empty intersection is contractible' condition required to be a good cover, note that we are using the $\check{\text{C}}$hech complex, rather than the Viertoris-Rips complex. Combining this with the linearity of the embedding and the assumptions placed on both $\varepsilon$ and $R$, we have covers that satisfy the Nerve Lemma.
    \end{remark}

    \begin{algorithm}\caption{$\Delta_{\varepsilon, R}(p)$}\label{alg:localstruct}
    		\KwData{An $\varepsilon$-dense sample $P$ of an embedded $2$-complex $|X|$, a point $p \in P$.} \KwResult{0 if the $(\varepsilon, R)$-local structure of $P$ at $p$ is not maximal, \newline 1 if the $(\varepsilon, R)$-local structure of $P$ at $p$ is maximal.}
    		\Begin{
      			$\mathcal{G}_p\longleftarrow \{ q \in P \mid \|p-q\| \le R+\varepsilon\}$\;
				connect $q, q' \in \mathcal{G}_p$ if $\|q-q'\| \le 3 \varepsilon$\;
      			$\mathcal{C}_p \longleftarrow \{q \in \mathcal{G}_p \mid q \text{ is path connected to } p \text{ in }  \mathcal{G}_p\}$\;
                remove $q \in \mathcal{C}_p$ if $\lVert p - q \rVert \geq  R -\varepsilon$\;
                \If{$\mathcal{RK}_0 =0$ and $\mathcal{RK}_1 =0$}{\Return{1}} 
            	\ElseIf{$\mathcal{RK}_0 = 1$ and $\mathcal{RK}_1 \neq 1$}{\Return{0}}
            	\ElseIf{$\mathcal{RK}_0 =1$ and $\mathcal{RK}_1 =1$}{
                    \If{$\forall q_1,q_2\in C_p, \, \exists q_0$ such that $\lVert q_1 - q_0 \rVert, \, \lVert q_2 - q_0 \rVert ,\, \lVert q_2 - q_1 \rVert \in [\sqrt{3(R^2-\varepsilon^2)}, \sqrt{3} R]$}{\Return{1}}
                    \Else{\Return{0}}
                }
                \ElseIf{$\mathcal{RK}_0 =2$ and $\mathcal{RK}_1 =0$}
                {
                    \If{$\operatorname{max}\left\{ \mathcal{D}(c_1), \mathcal{D}(c_2)\right\} \leq 5\varepsilon$}
                    {
                        \If{$\langle q_1 -p, q_2 -p\rangle > -R^2 +2 R\varepsilon - 7\varepsilon^2$}{\Return{1}}
                        \Else{\Return{0}}
                    }
                    \Else{\Return{0}}
                }
                \ElseIf{$\mathcal{RK}_0 =n, \, n \neq 0,1,2$ and $\mathcal{RK}_1 =0$}{\Return{0}}
                }
 		\end{algorithm}

    After we have $P_{LM}$, we use the function \[\mathfrak{D}_{\varepsilon, R}(p): P_{LM} \to \{ \, 0, \, 1, \, 2 \, \},\] see \Cref{alg:localstructdim} to determine what dimension of $(\varepsilon,R)$-local structure each sample in $P_{LM}$ has. 

    \begin{algorithm}\caption{$\mathfrak{D}_{\varepsilon, R}(p)$}\label{alg:localstructdim}
    		\KwData{An $\varepsilon$-dense sample $P$ of an embedded $2$-complex $|X|$, a point $p \in P$ such that the $(\varepsilon, R)$-local structure of $P$ at $p$ is maximal.} \KwResult{0 if the $(\varepsilon, R)$-local structure of $P$ at $p$ is maximal of dimension 0, \newline 1 if the $(\varepsilon, R)$-local structure of $P$ at $p$ is maximal of dimension 1, \newline 2 if the $(\varepsilon, R)$-local structure of $P$ at $p$ is maximal of dimension 2.}
    		\Begin{
                $\mathcal{G}_p\longleftarrow \{ q \in P \mid \|p-q\| \le R+\varepsilon\}$\;
				connect $q, q' \in \mathcal{G}_p$ if $\|q-q'\| \le 3 \varepsilon$\;
      			$\mathcal{C}_p \longleftarrow \{q \in \mathcal{G}_p \mid q \text{ is path connected to } p \text{ in }\mathcal{G}_p\}$\;
                remove $q \in C_p$ if $\|p-q\| \le R -\varepsilon$\;
                \If{$\mathcal{RK}_0 =0$ and $\mathcal{RK}_1 =0$}{\Return{0}}
                \ElseIf{$\mathcal{RK}_0 =2, \, n \neq 0,1,2$ and $\mathcal{RK}_1 =0$}
                {\Return{1}}
            	\ElseIf{$\mathcal{RK}_0 =1, \, n \neq 0,1,2$ and $\mathcal{RK}_1 =1$}{\Return{2}}
            }
 		\end{algorithm}

    Recall that our end goal is to learn the combinatorial structure of $X$. We begin by learning the number of triangles, locally maximal edges, and locally maximal vertices. Consider the following three subsets of $P_{LM}$:
    \begin{align*}
        P_{LM,2}    &= \left \{ p \in P_{LM}\, \mid \, \mathfrak{D}_{\varepsilon, R}(p) = 2 \right \}, \\
        P_{LM,1}    &= \left \{ p \in P_{LM}\, \mid \, \mathfrak{D}_{\varepsilon, R}(p) = 1 \right \}, \\
        P_{LM,0}    &= \left \{ p \in P_{LM}\, \mid \, \mathfrak{D}_{\varepsilon, R}(p) = 0 \right \}.
    \end{align*}
    
    When partitioning $P$ into $P_{LM}$ and $P_{NLM}$, there is a \emph{grey} region where a sample $p$ could be in either of these two sets. This presents a problem for learning the combinatorics of $X$ from the partitioning $P_{LM}$ and $P_{NLM}$. We can overcome this, by \emph{cleaning} $P_{LM}$. In particular, we clean $P_{LM,2}$ and $P_{LM,1}$.

    We begin by introducing the notion of a connected component of $\check{\mathcal{C}}_{\frac{3\varepsilon}{2}}\left( P_{LM,1} \right)$ \emph{spanning} an edge, and then introduce the notion of a connected component of $\check{\mathcal{C}}_{\frac{3\varepsilon}{2}}\left( P_{LM,2} \right)$ \emph{spanning} a triangle.

    \begin{definition}[Spanning an edge]\label{def:2comps:spanedge}
        We say a connected component of $\check{\mathcal{C}}_{\frac{3\varepsilon}{2}}(P_{LM,1})$ \emph{spans} a locally maximal edge $\overline{uv}$ if it contains a sample $p$ within $\varepsilon$ of the midpoint of $\overline{uv}$.
    \end{definition}

     \begin{definition}[Spanning a triangle]\label{def:spantri}
        We say a connected component of $\check{\mathcal{C}}_{\frac{3\varepsilon}{2}}(P_{LM,2})$ \emph{spans} a triangle $\triangle uvw$ if it contains a sample $p$ within $\varepsilon$ of the midpoint of $\triangle uvw$.
    \end{definition}

    We require some geometric conditions on when a connected component spans an edge or a triangle. For an edge, we will use the diameter of the connected component as a condition. 
    
    \begin{proposition}\label{prop:spanningedge}
        A connected component $C$ of $\check{\mathcal{C}}_{\frac{3\varepsilon}{2}}(P_{LM,1})$ \emph{spans} a locally maximal edge $\overline{uv}$ if and only if $\mathcal{D}(C) \geq \frac{3R}{2} - 2\varepsilon$.
    \end{proposition}
    
    Unfortunately, it is not immediately clear that such a test is suitable for detecting components that span triangles. For instance, consider a complex which consists of a single triangle, its three edges, and the three required vertices. While heuristically, it is unlikely to occur, the sampling could lead to 2 connected components  $C_1, C_2 \in \check{\mathcal{C}}_{\frac{3\varepsilon}{2}}\left(P_{LM,2}\right)$: one which is far away from the boundary of the triangle, and one that is \emph{surrounded} by points in $P_{NLM}$, both with large diameters. In fact, the one we wish to say is spanning, say $C_1$, will have a smaller diameter than the other one, $C_2$. Note, however, that as $C_2$ does not contain a sample $p$ near the midpoint of $\triangle uvw$, if $\mathcal{D}(C_1) \leq \mathcal{D}(C_2)$, then $C_2$ contains a non-contractible loop. However, a sample $p \in P$ near the midpoint $m_{\triangle u v w}$ of a triangle $\triangle u v w$ is not near any samples $q \notin P_{LM,2}$, and so we can exploit this fact to obtain a geometric test.

    \begin{proposition}\label{prop:spanningtri}
        A connected component $C$ of $\check{\mathcal{C}}_{\frac{3\varepsilon}{2}}$ \emph{spans} a triangle $\triangle uvw$ if and only if there is a point $p \in C$ such that 
        \begin{equation*}
            B_{\frac{R}{2} +  \varepsilon}(p) \cap P \subset P_{LM,2}.
        \end{equation*}
    \end{proposition}
    
    We now have geometric conditions for determining if a connected component of $\check{\mathcal{C}}_{\frac{3\varepsilon}{2}}\left(P_{LM,2}\right)$/$\check{\mathcal{C}}_{\frac{3\varepsilon}{2}}\left(P_{LM,1}\right)$ spans a triangle/edge respectively. Next, show that the locally maximal vertices of $X$ are in bijection with connected components of $\check{\mathcal{C}}_{\frac{3\varepsilon}{2}}\left(P_{LM,0}\right)$, the locally maximal edges of $X$ are in bijection with the spanning connected components of $\check{\mathcal{C}}_{\frac{3\varepsilon}{2}}(P_{LM,1})$, and that the triangles of $X$ are in bijection with the spanning connected components of $\check{\mathcal{C}}_{\frac{3\varepsilon}{2}}(P_{LM,2})$.

    We begin with the locally maximal vertices.
    
    \begin{proposition}\label{prop:equivlmvert}
        The connected components of $\check{\mathcal{C}}_{\frac{3\varepsilon}{2}} \left(P_{LM,0}\right)$ are in bijection with the set $V_{LM}$ of locally maximal vertices of $X$.
    \end{proposition}

    Next, we show that the edge spanning components are in bijection with the locally maximal edges. 
    
    \begin{proposition}\label{prop:equivlmspanedge}
        The \emph{spanning} components of $\check{\mathcal{C}}_{\frac{3\varepsilon}{2}} \left(P_{LM,1}\right)$ are in bijection with the set $E_{LM}$ of locally maximal edges of $X$.
    \end{proposition}

    Finally, we show that the spanning components of $\check{\mathcal{C}}_{\frac{3\varepsilon}{2}} \left ( P_{LM,2} \right )$ are in bijection with the triangles of $X$.

    \begin{proposition}\label{prop:equivlmspantri}
        The \emph{spanning} components of $\check{\mathcal{C}}_{\frac{3\varepsilon}{2}} \left(P_{LM,2}\right)$ are in bijection with the set $T$ of triangles in $X$.
    \end{proposition}
    
    Having identified the locally maximal cells $X_{LM}$ of $X$, we could learn the combinatorial structure of $X$ by identifying the structure of $X_{NLM}$ from $P_{NLM}$, and combining this with what we know about $X_{LM}$ from $P_{LM}$. The process in \cite{fods-graph} could be applied, but this requires the existence of some $\widetilde{\varepsilon}$ such that $P_{NLM}$ is a $\widetilde{\varepsilon}$-sample of $X_{NLM}$ satisfying Assumptions 1 in \cite{fods-graph} This would impose stricter assumptions than \Cref{as:embed}, but after ensuring these new assumptions are satisfied, works out of the box. 
    
    To avoid placing stricter assumptions on $|X|$, we use the idea of \emph{witness points} to discover the combinatorics. For each sample $p \in P_{NLM}$, we can examine the spanning connected components $C_{LM}$ of $\check{\mathcal{C}}_{\frac{3\varepsilon}{2}}\left(P_{LM,1}\right)$ and $\check{\mathcal{C}}_{\frac{3\varepsilon}{2}}\left(P_{LM,2}\right)$ such that $C_{LM} \cap B_{R+3\varepsilon}\left(p\right) \neq \emptyset$. In particular, we can use $\mathfrak{D}_{\varepsilon, R}(q)$ for some $q \in C_{LM}$, to determine of what dimension the local structure is maximal. If there is a $q$ in $C_{LM} \cap S_{R-\varepsilon}^{R+\varepsilon}\left(p\right)$ such that $\mathfrak{D}_{\varepsilon, R}(q) = 1$, then $p$ is near a vertex. 
    
    If there are no connected components $C_{LM}$ which are $(\varepsilon, R)$-locally maximal of dimension 1, then $p$ only \emph{witnesses} samples $q \in P_{LM}$ such that the $(\varepsilon, R)$-local structure of $P$ at $q$ is maximal of dimension 2. Hence, we need to understand the combinatorics of $|X|\setminus \left(E_{LM}\cup V_{LM}\right)$ where $E_{LM}$ is the set of locally maximal edges and $V_{LM}$ the set of locally maximal vertices. 

    In \Cref{as:embed}, we assumed that for any triangle $\triangle u v w$, \[ \angle u v w, \angle v w u , \angle w u v \geq \frac{\pi}{6}.\] This means that for any sample $p \in P_{NLM}$ with $d(\partial \triangle u v w, p ) < R+ \varepsilon$ for some $\triangle u v w$, there is some sample $q \in P_{LM,2}$ with $d(\triangle u v w, q) \leq \varepsilon$ and $d(\partial \triangle u v w, p ) \geq R+ \varepsilon$, such that $\lVert q - p \rVert \leq \frac{2\sqrt{2}(R+2\varepsilon)}{\sqrt{3}-1}$. Further, $q$ is in a triangle spanning component $\mathcal{T}$.

    Similarly, for any sample $p \in P_{NLM}$ with $d(\partial \overline{uv}, p) < \frac{3R}{2}+ \varepsilon$ for some edge $\overline{uv}$, there is a sample $q \in P_{LM,1}$ with $d(\partial \overline{uv}, p) \geq \frac{3R}{2} + \varepsilon$ such that $\lVert q - p \rVert \leq \frac{2\sqrt{2}(R+2\varepsilon)}{\sqrt{3}-1}$. Further, $q$ is in an edge spanning component $\mathcal{E}$.

    This leads us to say a sample $p \in P_{NLM}$ \emph{witnesses} a spanning connected component $\mathcal{C}$ if \[B_{\frac{2\sqrt{2}(R+2\varepsilon)}{\sqrt{3}-1}}(p) \cap \mathcal{C} \neq \emptyset.\] For ease of reading, we set $\kappa  = \frac{2\sqrt{2}}{\sqrt{3}-1}$.

    \begin{definition}[Witnessing a spanning component]\label{def:witness-spanning}
        Let $P$ be an $\varepsilon$-sample $P$ of an embedded $2$-complex $|X|$ satisfying \Cref{as:embed}. Then a sample $p \in P_{NLM}$ \emph{witnesses} an edge/triangle spanning component if 
        \begin{equation*}
            B_{\kappa(R+\varepsilon)}(p) \cap \mathcal{C} \neq \emptyset.
        \end{equation*}
    \end{definition}

    To determine the final combinatorial structure of $X$, we look at the local neighbourhood of each $p \in P_{NLM}$ and look at both
    \begin{align*}
        B_{(R+2\varepsilon)\kappa}(p) &\cap \check{\mathcal{C}}_{\frac{3\varepsilon}{2}} \left( P_{LM,1} \right) \\
        B_{(R+2\varepsilon)\kappa}(p) &\cap \check{\mathcal{C}}_{\frac{3\varepsilon}{2}} \left( P_{LM,2} \right). \\
    \end{align*}

    If 
    
    \begin{equation*}
         B_{(R+2\varepsilon)\kappa}(p) \cap \check{\mathcal{C}}_{\frac{3\varepsilon}{2}} \left( P_{LM,1} \right) \neq \emptyset
    \end{equation*}

    then we know that $p$ is \emph{near} a vertex, and the spanning components $\mathcal{E}$ of $\check{\mathcal{C}}_{\frac{3\varepsilon}{2}} \left( P_{LM,1} \right)$ that $p$ witnesses, share a boundary vertex. Further, if 

    \begin{equation*}
         B_{(R+2\varepsilon)\kappa}(p) \cap \check{\mathcal{C}}_{\frac{3\varepsilon}{2}} \left( P_{LM,2} \right) \neq \emptyset
    \end{equation*}

    as well, then there are spanning components of $\check{\mathcal{C}}_{\frac{3\varepsilon}{2}} \left ( P_{LM,2}\right)$ that $p$ witnesses, which have a vertex in common with the edges.

    If only 
    
    \begin{equation*}
         B_{(R+2\varepsilon)\kappa}(p) \cap \check{\mathcal{C}}_{\frac{3\varepsilon}{2}} \left( P_{LM,2} \right) \neq \emptyset
    \end{equation*}

    we examine how many spanning components $\mathcal{T}$ are seen by $p$, as well as if samples $p \in P_{NLM}$ that witness $\mathcal{T}$, also witness any other spanning components $\mathcal{T}'$. We use this information to partition $P_{NLM}$ into $\{P_i\}$ in \Cref{alg:pnlmpart}, with a final clean of the partitions, to account for some special cases. As $R \leq 16 \varepsilon$, for all $p \in P_{NLM}$ there is some spanning connected component $\mathcal{C}$ such that $B_{\frac{R+\varepsilon}{\kappa}}(p) \cap \mathcal{C} \neq \emptyset$.

    We then label each component $P_i$ as follows, from \Cref{alg:P1class,alg:P2class}:
    \begin{itemize}
        \item $-1$ if $P_i$ corresponds to $2$ vertices,
        \item $0$ if $P_i$ corresponds to a vertex, 
        \item $1$ if $P_i$ corresponds to a vertex and an edge,
        \item $2$ if $P_i$ corresponds to two vertices and an edge, 
        \item $3$ if $P_i$ corresponds to just an edge, 
        \item $4$ if $P_i$ corresponds to two edges and a vertex, 
        \item $5$ if $P_i$ corresponds to three edges and two vertices, 
        \item $6$ if $P_i$ corresponds to three edges and a vertex, 
        \item $7$ if $P_i$ corresponds to three edges and three vertices,
        \item $8$ if $P_i$ corresponds to three edges,
        \item $9$ if $P_i$ corresponds to two edges,
    \end{itemize}
    
    \begin{algorithm}\caption{Spanning triangle components}\label{alg:spanningtri}
        \KwData{Parameters $\varepsilon, R$ and $P_{LM,1}$.} \KwResult{The set of triangle spanning components.}
            
            \Begin{
                Initialise empty set $T$\;
                Let $C$ be the set of connected components of  $\check{\mathcal{C}}_{\frac{3\varepsilon}{2}}\left(P_{LM,2}\right)$\;
                \For{$\mathcal{T} \in C$}{
                    \If{$\exists p \in \mathcal{C}$ such that $B_{R/2 +\varepsilon}(p) \cap P \subset P_{LM,2}$}{
                        Add $\mathcal{T}$ to $T$\;
                    }
                }
                \Return $T$
            }
    \end{algorithm}

    \begin{algorithm}\caption{Spanning edge components}\label{alg:spanningedge}
        \KwData{Parameters $\varepsilon, R$ and $P_{LM,1}$.} \KwResult{The set of triangle spanning components.}
            
            \Begin{
                Initialise empty set $E$\;
                Let $C$ be the set of connected components of  $\check{\mathcal{C}}_{\frac{3\varepsilon}{2}}\left(P_{LM,2}\right)$\;
                \For{$\mathcal{E} \in C$}{
                    \If{$\mathcal{D}(\mathcal{T}) \geq \frac{3R}{2}-2\varepsilon$}{
                        Add $\mathcal{E}$ to $E$\;
                    }
                }
                \Return $E$
            }
    \end{algorithm}
    
    \begin{algorithm}\caption{Partitioning $P_{NLM}$}\label{alg:pnlmpart}
        \KwData{An $\varepsilon$-dense sample $P$ of an embedded $2$-complex $|X|$, partitioned into $P_{NLM}, \, P_{LM,0}, \, P_{LM,1}, \, P_{LM,2}$.} \KwResult{A partition $\{ P_i \}$ of $P_{NLM}$, and for each $P_i$, two sets $S_E(P_i), \, S_T(P_i)$.}
                
    	\Begin{
            For each $p \in P_{NLM}$, find all the edge spanning components $\mathcal{E}$ such that $\mathcal{E} \cap B_{(R+2\varepsilon)\kappa}(p) \neq \emptyset$, and place them in $S_E(p)$\;
            Find all the triangle spanning components $\mathcal{T}$ such that $\mathcal{T} \cap B_{(R+2\varepsilon)\kappa}(p) \neq \emptyset$, and place them in $S_T(p)$\;
            Partition $P_{NLM}$ into $\{ \, P_i \, \}$ such that for each $p , q \in P_i$, $S_E(p) = S_E(q)$ and $S_T(p) = S_T(q)$\;
            Assign $S_E(P_i)$ and $S_T(P_i)$ to each $P_i$\;
            \For{$P_i$ and $P_j$ with $S_E(P_j) \subseteq S_E(P_i)$ and $S_T(P_j) \subseteq S_T(P_i)$}{ 
                \If{$S_E(P_j), S_T(P_j) \neq \emptyset$}{
                    Merge $P_j$ into $P_i$ with labels $S_E(P_i), \, S_T(P_i)$\;
                } \ElseIf{$|S_T(P_j)|\geq 2$ and $\forall p \in P_j$ such that $\sig_{\varepsilon,R}(p) =(n,0), \, n \in \mathbb{Z}_{\geq 0}$}{
                    Merge $P_j$ into $P_i$           with labels $S_E(P_i), \, S_T(P_i)$\;
                }
            }
            \Return $\{ P_i \}$, and $S_E(P_i), \, S_T(P_i)$ for each $P_i$
        }
 	\end{algorithm}

    \begin{algorithm}\caption{Order $\{ P_i \}$}\label{alg:Piorder}
    	\KwData{An $\varepsilon$-dense sample $P$ of an embedded $2$-complex $|X|$, partition $\{P_i\}$ of $P_{NLM}$ with two sets $S_E(P_i), \, S_T(P_i)$ for each $P_i$ and partitions of $P_{LM,0}, \, P_{LM,1}, \, P_{LM,2}$.} \KwResult{Two sets $P^1, P^2 \subset \{ \, P_i \, \}$.}
    
    	\Begin{
                Initialise empty $P^1$ and $P^2$\;
               \For{$P_i \in \{ P_i \}$}{
                    \If{$S_E(P_i) \neq \emptyset$}{Add $P_i$ to $P^1$}
                    \ElseIf{$\exists p \in P_i$ such that $\sig(p) \neq (1,n)$}{Add $P_i$ to $P^1$}
                    \ElseIf{$|S_T(P_i)| \neq 1$}{Add $P_i$ to $P^1$}
                    \Else{Add $P_i$ to $P^2$}
               }
               \Return $P^1, \, P^2$
        }
    \end{algorithm}

  \begin{algorithm}\caption{Classification of $P^1$}\label{alg:P1class}
    	\KwData{An $\varepsilon$-dense sample $P$ of an embedded $2$-complex $|X|$, $P^1$, and partitions of $P_{NLM}, \, P_{LM,0}, \, P_{LM,1}, \, P_{LM,2}$.} \KwResult{A labeled list $C$, where the label for $P_i$ is $-1$ if $P_i$ corresponds to $2$ vertices, $0$ if $P_i$ corresponds to a vertex, $1$ if $P_i$ corresponds to a vertex and an edge, $2$ if $P_i$ corresponds to two vertices and an edge, $3$ if $P_i$ corresponds to just an edge.}
    
    	\Begin{
            Initialise empty list $C$\;
            \For{$P_i \in P^1$}{
                \If{$|S_E(P_i)| = 1$ and $S_T(P_i) = \emptyset$}{
                    \If{$\mathcal{E} \notin S_E(P_j) \forall P_j \neq P_i$ }{
                        Add $P_i$ to $C$ with label $-1$\;
                    } \ElseIf{$\exists P_j \neq P_i$ such that $\mathcal{E} \in S_E(P_j)$}{
                        Add $P_i$ to $C$ with label 0\;
                    }
                } \ElseIf{$S_E(P_i) \neq \emptyset$}{
                    Add $P_i$ to $C$ with label 0\;
                } \Else{
                    \For{$\mathcal{T} \in S_T(P_i)$}{
                        Let $LN(\mathcal{T}) = \{ P_k \mid \mathcal{T} \in S_T(P_k)\}$
                    }
                    Let $N(P_i) = \bigcap_{\mathcal{T} \in S_T(P_i)} LN(\mathcal{T})$\;
                    \If{$N(P_i) = \{ P_i, P_k \}$}{
                        Add $P_i$ to $C$ with label $1$\;
                        Add $P_k$ to $C$ with label $0$, unless $P_k$ is already in $C$\;
                    } \ElseIf {$N(P_i) = \{ P_i, P_k, P_l \}$}{
                        Add $P_i$ to $C$ with label $3$\;
                        Add $P_k$ to $C$ with label $0$, unless $P_k$ is already in $C$\;
                        Add $P_l$ to $C$ with label $0$, unless $P_l$ is already in $C$\;
                    }
                }
            }
            \If{$\exists P_i \in P^1 \setminus C$}{
                Add $P_i$ to $C$ with label $2$\;
            }
            \Return $C$
        }
    \end{algorithm}

    \begin{algorithm}\caption{Classification of $P^2$}\label{alg:P2class}
    	\KwData{An $\varepsilon$-dense sample $P$ of an embedded $2$-complex $|X|$, $P^2$, and partitions of $P_{NLM}, \, P_{LM,0}, \, P_{LM,1}, \, P_{LM,2}$, a labelled list $C$ obtained from \Cref{alg:P1class}.} \KwResult{A labelled list $C$.}
    
    	\Begin{
            \For{$P_i \in P^2$}{
                \If{ $P_i \notin C$}{
                    Let $LN  = \{ P_k \mid \mathcal{T} \in S_T(P_k)\}$\;
                    \If {$LN \cap P^2  = \{ P_i, P_k, P_l \}$}{
                            Add $P_i, P_k, P_l$ to $C$ with label $3$\;
                    } \ElseIf{$LN \cap P^2  = \{ P_i, P_k \}$}{
                        Add $P_i$ to $C$ with label $3$\;
                        Add $P_l$ to $C$ with label $4$\;
                    } \ElseIf { $LN \cap P^2 = \{ P_i\}$}{
                        \If{$LN= \{ P_i \}$}{
                            Add $P_i$ to $C$ with label $7$\;
                        } \ElseIf {$LN= \{ P_i, P_k \}$ and $P_k$ has label $0$}{
                            Add $P_i$ to $C$ with label $5$\;
                        } \ElseIf {$LN= \{ P_i, P_k \}$ and $P_k$ has label $2$}{
                                Add $P_i$ to $C$ with label $4$\;
                        } \ElseIf {$LN= \{ P_i, P_k, P_l \}$ and $P_k$ has label $0$, $P_l$ label $1$}{
                                Add $P_i$ to $C$ with label $4$\;
                        } \ElseIf {$LN= \{ P_i, P_k, P_l \}$ and $P_k$ has label $1$, $P_l$ label $2$}{
                                Add $P_i$ to $C$ with label $3$\;
                        } \ElseIf {$LN= \{ P_i, P_k, P_l \}$ and $P_k$ has label $0$, $P_l$ label $0$}{
                                Add $P_i$ to $C$ with label $6$\;
                        } \ElseIf {$LN= \{ P_i, P_k, P_l, P_j \}$ and $P_k, P_l, P_j$ have label $0$}{
                                Add $P_i$ to $C$ with label $8$\;
                        } \ElseIf {$LN= \{ P_i, P_k, P_l, P_j, P_m \}$ and $P_k, P_l, P_j$ have label $0$ and $P_m$ has label $3$}{
                                Add $P_i$ to $C$ with label $9$\;
                        }
                    }
                }
            }
            \Return $C$
        }
    \end{algorithm}

    \begin{algorithm}[H]\caption{Number of triangles, edges and vertices.}\label{alg:numcells}
    	\KwData{An $\varepsilon$-dense sample $P$ of an embedded $2$-complex $|X|$, partitions of $P_{NLM}, \, P_{LM,0}, \, P_{LM,1}, \, P_{LM,2}$ and the labelled list $C$ from \Cref{alg:P2class}.} \KwResult{The triangles, edges, and vertices in $X$.}
                
    	\Begin{
            Initialise an empty weighted graph $B$\;
            $\forall$ spanning components $\mathcal{T}$ of $P_{LM,2}$, add weight $2$ node to $B$, labelled with $\mathcal{T}$\;
            $\forall$ spanning components $\mathcal{E}$ of $P_{LM,1}$, add weight $1$ node to $B$, labelled with $\mathcal{E}$\;
            $\forall$ components $\mathcal{V}$ of $P_{LM,0}$, add weight $0$ node to $B$, labelled with $\mathcal{V}$\;
            \For{$P_i \in C$}{
                \If{$P_i$ has label $-1$}{
                    Add $2$ weight $0$ nodes to $B$, labelled with $P_i$\;
                }\ElseIf{$P_i$ has label $0$}{
                    Add weight $0$ node to $B$, labelled with $P_i$\;
                } \ElseIf{$P_i$ has label $1$}{
                    Add $2$ weight $0$ nodes to $B$, labelled with $P_i$\;
                    Add weight $1$ node to $B$, labelled with $P_i$\;
                } \ElseIf{$P_i$ has label $2$}{
                    Add weight $0$ node to $B$, labelled with $P_i$\;
                    Add weight $1$ node to $B$, labelled with $P_i$\;
                } \ElseIf{$P_i$ has label $3$}{
                    Add two weight $0$ nodes to $B$, labelled with $P_i$\;
                    Add weight $1$ node to $B$, labelled with $P_i$\;
                } \ElseIf{$P_i$ has label $4$}{
                    Add weight $1$ node to $B$, labelled with $P_i$\;
                } \ElseIf{$P_i$ has label $5$}{
                    Add weight $0$ node to $B$, labelled with $P_i$\;
                    Add two weight $1$ nodes to $B$, labelled with $P_i$\;
                } \ElseIf{$P_i$ has label $6$}{
                    Add two weight $0$ nodes to $B$, labelled with $P_i$\;
                    Add three weight $1$ nodes to $B$, labelled with $P_i$\;
                } \ElseIf{$P_i$ has label $7$}{
                    Add three weight $0$ nodes to $B$, labelled with $P_i$\;
                    Add three weight $1$ nodes to $B$, labelled with $P_i$\;
                }
            }
        }
    \end{algorithm}

    The following lemmas together show that \Cref{alg:pnlmpart,alg:P1class,alg:P2class} correctly partition $P_{NLM}$ and label the partitions $P_i$ appropriately.

    \begin{lemma}\label{lem:lmedge-disconnected}
        Let $\overline{uv}$ be a locally maximal edge of $X$, such that $u,v$ are only faces of $\overline{uv}$. Then, there is a unique partition $P_1$ of $P_{NLM}$ which witnesses $\mathcal{E}$, where $\mathcal{E}$ is the edge spanning component corresponding to $\overline{uv}$. Further, $P_1$ is assigned label $-1$ by \Cref{alg:P1class,alg:P2class}.
    \end{lemma}

   \begin{lemma}\label{lem:lmedge-connected}
        Let $\overline{uv}$ be a locally maximal edge of $X$, such that $u$ and/or $v$ is the face of some locally maximal cell $\sigma \in X, \, \sigma \neq \overline{uv}$. Then, there are partitions $P_1, P_2$ of $P_{NLM}$, which witness $\mathcal{E}$, where $\mathcal{E}$ is the edge spanning component corresponding to $\overline{uv}$. Further, $P_1$ and $P_2$ are assigned label $0$ by \Cref{alg:P1class,alg:P2class}.
   \end{lemma}

    \begin{lemma}\label{lem:lmtriangle-disconnected}
        Let $\triangle u v w$ be a triangle of $X$, such that for all locally maximal cells $\sigma \in X$ with $\sigma \neq \triangle u v w$, we have 

        \begin{equation*}
            u, \, v, \, w \notin \sigma.
        \end{equation*}

        Then, there is a unique partition $P_1$ of $P_{NLM}$ which witness $\mathcal{T}$, where $\mathcal{T}$ is the edge spanning component corresponding to $\triangle u v w$.  Further, $P_1$ is given label $7$ by \Cref{alg:P1class,alg:P2class}.
    \end{lemma}

    \begin{lemma}\label{lem:lmtriangle-share-1-vertex}
        Let $\triangle u v w$ be a triangle of $X$, such that there is some locally maximal cell $\sigma \in X$ with $\sigma \neq \triangle u v w$, such that $v \in \sigma$, without loss of generality, and for all locally maximal $\tau \in X, \, \tau \neq \sigma, \triangle u v w$, either $\triangle u v w \cap \tau = v$ or $\triangle u v w \cap \tau = \emptyset$.

        Then, there are exactly two partitions $P_1, P_2$ of $P_{NLM}$ which witness $\mathcal{T}$, where $\mathcal{T}$ is the edge spanning component corresponding to $\triangle u v w$.  Further, $P_1$ is given label $0$ and $P_2$ label $5$ by \Cref{alg:P1class,alg:P2class}.
    \end{lemma}

    \begin{lemma}\label{lem:lmtriangle-share-1-edge-with-vertices}
        Let $\triangle u v w$ be a triangle of $X$, such that there is some locally maximal cell $\sigma \in X$ with $\sigma \neq \triangle u v w$, such that $v \in \sigma$, without loss of generality, and for all locally maximal $\tau \in X, \, \tau \neq \sigma, \triangle u v w$, either $\triangle u v w \cap \tau = \overline{uv}$ or $\triangle u v w \cap \tau = \emptyset$.

        Then, there are exactly two partitions $P_1, P_2$ of $P_{NLM}$ which witness $\mathcal{T}$, where $\mathcal{T}$ is the edge spanning component corresponding to $\triangle u v w$.  Further, $P_1$ is given label $0$ and $P_2$ label $5$ by \Cref{alg:P1class,alg:P2class}.
    \end{lemma}

    \begin{lemma}\label{lem:lmtriangle-share-2-vertices}
        Let $\triangle u v w$ be a triangle of $X$, such that there are some locally maximal cells $\sigma_1 \neq \sigma_2 \in X$ with $\sigma_1, \sigma_2 \neq \triangle u v w$, such that 

        \begin{align*}
            \sigma_1 \cap \triangle u v w &= v \\
            \sigma_2 \cap \triangle u v w &= u \\
        \end{align*}
        and for all other locally maximal cells $\tau \in X$, either

        \begin{enumerate}
            \item $\tau \cap \triangle u v w = v$,
            \item $\tau \cap \triangle u v w = u$,
            \item $\tau \cap \triangle u v w = \emptyset$.
        \end{enumerate}

        Then, there are exactly three partitions $P_1, P_2, P_2$ of $P_{NLM}$ which witness $\mathcal{T}$, where $\mathcal{T}$ is the edge spanning component corresponding to $\triangle u v w$.  Further, $P_1, P_2$ are given label $0$ and $P_3$ label $6$ by \Cref{alg:P1class,alg:P2class}.
    \end{lemma}

    \begin{lemma}\label{lem:share-edge-with-vertex-and-vertex}
        Let $\triangle u v w$ be a triangle of $X$, such that there are some locally maximal cells $\sigma_1 \neq \sigma_2 \in X$ with $\sigma_1, \sigma_2 \neq \triangle u v w$, such that 

        \begin{align*}
            \sigma_1 \cap \triangle u v w &= \overline{uv} \\
            \sigma_2 \cap \triangle u v w &= v \\
        \end{align*}
        
        and for all other locally maximal cells $\tau \in X$, either

        \begin{enumerate}
            \item $\tau \cap \triangle u v w = \overline{uv}$,
            \item $\tau \cap \triangle u v w = v$,
            \item $\tau \cap \triangle u v w = \emptyset$.
        \end{enumerate}
        
        Then, there are exactly three partitions $P_1, P_2, P_2$ of $P_{NLM}$ which witness $\mathcal{T}$, where $\mathcal{T}$ is the edge spanning component corresponding to $\triangle u v w$.  Further, $P_1$ has label $0$, $P_2$ label $1$ and $P_3$ label $4$ by \Cref{alg:P1class,alg:P2class}.
    \end{lemma}

    \begin{lemma}\label{lem:share-edge-with-vertices-and-vertex}
        Let $\triangle u v w$ be a triangle of $X$, such that there are some locally maximal cells $\sigma_1 \neq \sigma_2 \in X$ with $\sigma_i \neq \triangle u v w$ and $\sigma_i \neq \sigma_j$ for $i \neq j$, such that 

        \begin{align*}
            \sigma_1 \cap \triangle u v w &= \overline{uv} \\
            \sigma_2 \cap \triangle u v w &= w \\
        \end{align*}
        
        and for all other locally maximal cells $\tau \in X$, either

        \begin{enumerate}
            \item $\tau \cap \triangle u v w = \overline{uv}$,
            \item $\tau \cap \triangle u v w = w$,
            \item $\tau \cap \triangle u v w = \emptyset$.
        \end{enumerate}
        
        Then, there are exactly three partitions $P_1, P_2, P_2$ of $P_{NLM}$ which witness $\mathcal{T}$, where $\mathcal{T}$ is the edge spanning component corresponding to $\triangle u v w$.  Further, $P_1$ has label $0$, $P_2$ label $2$ and $P_3$ label $9$ by \Cref{alg:P1class,alg:P2class}.
    \end{lemma} 

    \begin{lemma}\label{lem:share-three-vertex}
        Let $\triangle u v w$ be a triangle of $X$, such that there are some locally maximal cells $\sigma_1, \sigma_2, \sigma_3 \in X$ with $\sigma_i \neq \triangle u v w$ and $\sigma_i \neq \sigma_j$ for $i \neq j$, such that 

        \begin{align*}
            \sigma_1 \cap \triangle u v w &= u \\
            \sigma_2 \cap \triangle u v w &= v \\
            \sigma_3 \cap \triangle u v w &= w \\
        \end{align*}
        
        and for all other locally maximal cells $\tau \in X$, either

        \begin{enumerate}
            \item $\tau \cap \triangle u v w = u$,
            \item $\tau \cap \triangle u v w = v$,
            \item $\tau \cap \triangle u v w = w$,
            \item $\tau \cap \triangle u v w = \emptyset$.
        \end{enumerate}
        
        Then, there are exactly four partitions $P_1, P_2, P_3, P_4$ of $P_{NLM}$ which witness $\mathcal{T}$, where $\mathcal{T}$ is the edge spanning component corresponding to $\triangle u v w$.  Further, $P_1, P_2$ and $P_3$ are labelled with $0$ and $P_4$ with $8$ by \Cref{alg:P1class,alg:P2class}.
    \end{lemma}

    \begin{lemma}\label{lem:share-edge-with-vertex-and-2-vertices}
        Let $\triangle u v w$ be a triangle of $X$, such that there are some locally maximal cells $\sigma_1, \sigma_2, \sigma_3 \in X$ with $\sigma_i \neq \triangle u v w$ and $\sigma_i \neq \sigma_j$ for $i \neq j$, such that 

        \begin{align*}
            \sigma_1 \cap \triangle u v w &= \overline{uv}\\
            \sigma_2 \cap \triangle u v w &= v \\
            \sigma_3 \cap \triangle u v w &= w \\
        \end{align*}
        
        and for all other locally maximal cells $\tau \in X$, either

        \begin{enumerate}
            \item $\tau \cap \triangle u v w = \overline{uv}$,
            \item $\tau \cap \triangle u v w = v$,
            \item $\tau \cap \triangle u v w = w$,
            \item $\tau \cap \triangle u v w = \emptyset$.
        \end{enumerate}
        
        Then, there are exactly four partitions $P_1, P_2, P_3, P_4$ of $P_{NLM}$ which witness $\mathcal{T}$, where $\mathcal{T}$ is the edge spanning component corresponding to $\triangle u v w$.  Further, $P_1$ is labelled with $1$, $P_2, P_3$ with $0$ and $P_4$ with $9$ by \Cref{alg:P1class,alg:P2class}.
    \end{lemma}

    \begin{lemma}\label{lem:share-edge-with-2-vertices}
        Let $\triangle u v w$ be a triangle of $X$, such that there are some locally maximal cells $\sigma_1, \sigma_2, \sigma_3 \in X$ with $\sigma_i \neq \triangle u v w$ and $\sigma_i \neq \sigma_j$ for $i \neq j$, such that 

        \begin{align*}
            \sigma_1 \cap \triangle u v w &= \overline{uv}\\
            \sigma_2 \cap \triangle u v w &= u \\
            \sigma_3 \cap \triangle u v w &= v \\
        \end{align*}
        
        and for all other locally maximal cells $\tau \in X$, either

        \begin{enumerate}
            \item $\tau \cap \triangle u v w = \overline{uv}$,
            \item $\tau \cap \triangle u v w = u$,
            \item $\tau \cap \triangle u v w = v$,
            \item $\tau \cap \triangle u v w = \emptyset$.
        \end{enumerate}
        
        Then, there are exactly four partitions $P_1, P_2, P_3, P_4$ of $P_{NLM}$ which witness $\mathcal{T}$, where $\mathcal{T}$ is the edge spanning component corresponding to $\triangle u v w$.  Further, $P_1$ is labelled with $3$, $P_2, P_3$ with $0$ and $P_4$ with $4$ by \Cref{alg:P1class,alg:P2class}.
    \end{lemma}

    \begin{lemma}\label{lem:share-2-edge-with-vertex-and-vertex}
        Let $\triangle u v w$ be a triangle of $X$, such that there are some locally maximal cells $\sigma_1, \sigma_2, \sigma_3 \in X$ with $\sigma_i \neq \triangle u v w$ and $\sigma_i \neq \sigma_j$ for $i \neq j$, such that 

        \begin{align*}
            \sigma_1 \cap \triangle u v w &= \overline{uv}\\
            \sigma_2 \cap \triangle u v w &= \overline{vw} \\
            \sigma_3 \cap \triangle u v w &= v \\
        \end{align*}
        
        and for all other locally maximal cells $\tau \in X$, either

        \begin{enumerate}
            \item $\tau \cap \triangle u v w = \overline{uv}$,
            \item $\tau \cap \triangle u v w = \overline{vw}$,
            \item $\tau \cap \triangle u v w = v$,
            \item $\tau \cap \triangle u v w = \emptyset$.
        \end{enumerate}
        
        Then, there are exactly four partitions $P_1, P_2, P_3, P_4$ of $P_{NLM}$ which witness $\mathcal{T}$, where $\mathcal{T}$ is the edge spanning component corresponding to $\triangle u v w$.  Further, $P_1$ is labelled with $0$, $P_2, P_3$ with $1$, and $P_3$ with $3$ by \Cref{alg:P1class,alg:P2class}.
    \end{lemma}

    \begin{lemma}\label{lem:share-edge-vertex-vertex-vertex}
        Let $\triangle u v w$ be a triangle of $X$, such that there are some locally maximal cells $\sigma_1, \sigma_2, \sigma_3, \sigma_4 \in X$ with $\sigma_i \neq \triangle u v w$ and $\sigma_i \neq \sigma_j$ for $i \neq j$, such that 

        \begin{align*}
            \sigma_1 \cap \triangle u v w &= u\\
            \sigma_2 \cap \triangle u v w &= v \\
            \sigma_3 \cap \triangle u v w &= w \\
            \sigma_4 \cap \triangle u v w &= \overline{uv}
        \end{align*}
        
        and for all other locally maximal cells $\tau \in X$, either

        \begin{enumerate}
            \item $\tau \cap \triangle u v w = u$,
            \item $\tau \cap \triangle u v w = v$,
            \item $\tau \cap \triangle u v w = w$,
            \item $\tau \cap \triangle u v w = \overline{uv}$
            \item $\tau \cap \triangle u v w = \emptyset$.
        \end{enumerate}
        
        Then, there are exactly five partitions $P_1, P_2, P_3, P_4, P_5$ of $P_{NLM}$ which witness $\mathcal{T}$, where $\mathcal{T}$ is the edge spanning component corresponding to $\triangle u v w$.  Further, $P_1, P_2, P_3$ are labelled with $0$, and $P_4$ with $8$ by \Cref{alg:P1class,alg:P2class}.
    \end{lemma}

    \begin{lemma}\label{lem:share-edge-with-vertex-edge-vertex-vertex}
        Let $\triangle u v w$ be a triangle of $X$, such that there are some locally maximal cells $\sigma_1, \sigma_2, \sigma_3, \sigma_4 \in X$ with $\sigma_i \neq \triangle u v w$ and $\sigma_i \neq \sigma_j$ for $i \neq j$, such that 

        \begin{align*}
            \sigma_1 \cap \triangle u v w &= u\\
            \sigma_2 \cap \triangle u v w &= v \\
            \sigma_3 \cap \triangle u v w &= \overline{vw} \\
            \sigma_4 \cap \triangle u v w &= \overline{uv}
        \end{align*}
        
        and for all other locally maximal cells $\tau \in X$, either

        \begin{enumerate}
            \item $\tau \cap \triangle u v w = u$,
            \item $\tau \cap \triangle u v w = v$,
            \item $\tau \cap \triangle u v w = w$,
            \item $\tau \cap \triangle u v w = \overline{uv}$
            \item $\tau \cap \triangle u v w = \overline{vw}$
            \item $\tau \cap \triangle u v w = \emptyset$.
        \end{enumerate}
        
        Then, there are exactly five partitions $P_1, P_2, P_3, P_4, P_5$ of $P_{NLM}$ which witness $\mathcal{T}$, where $\mathcal{T}$ is the edge spanning component corresponding to $\triangle u v w$.  Further, $P_1, P_2$ are labelled with $0$, $P_3$ with $1$, and $P_4, P_5$ with $3$ by \Cref{alg:P1class,alg:P2class}.
    \end{lemma}

    \begin{lemma}\label{lem:share-edge-edge-vertex-vertex-vertex}
        Let $\triangle u v w$ be a triangle of $X$, such that there are some locally maximal cells $\sigma_1, \sigma_2, \sigma_3, \sigma_4, \sigma_5 \in X$ with $\sigma_i \neq \triangle u v w$ and $\sigma_i \neq \sigma_j$ for $i \neq j$, such that 

        \begin{align*}
            \sigma_1 \cap \triangle u v w &= u\\
            \sigma_2 \cap \triangle u v w &= v \\
            \sigma_3 \cap \triangle u v w &= w \\
            \sigma_4 \cap \triangle u v w &= \overline{uv}\\
            \sigma_5 \cap \triangle u v w &= \overline{vw}\\
        \end{align*}
        
        and for all other locally maximal cells $\tau \in X$, either

        \begin{enumerate}
            \item $\tau \cap \triangle u v w = u$,
            \item $\tau \cap \triangle u v w = v$,
            \item $\tau \cap \triangle u v w = w$,
            \item $\tau \cap \triangle u v w = \overline{uv}$
            \item $\tau \cap \triangle u v w = \overline{vw}$
            \item $\tau \cap \triangle u v w = \emptyset$.
        \end{enumerate}
        
        Then, there are exactly five partitions $P_1, P_2, P_3, P_4, P_5, P_6$ of $P_{NLM}$ which witness $\mathcal{T}$, where $\mathcal{T}$ is the edge spanning component corresponding to $\triangle u v w$.  Further, $P_1, P_2, P_3$ are labelled with $0$, $P_4, P_5, P_6$ with $3$ by by \Cref{alg:P1class,alg:P2class}.
    \end{lemma}

    \begin{lemma}\label{lem:share-edge-edge-edge-vertex-vertex-vertex}
        Let $\triangle u v w$ be a triangle of $X$, such that there are some locally maximal cells $\sigma_1, \sigma_2, \sigma_3, \sigma_4, \sigma_5, \sigma_6 \in X$ with $\sigma_i \neq \triangle u v w$ and $\sigma_i \neq \sigma_j$ for $i \neq j$, such that 

        \begin{align*}
            \sigma_1 \cap \triangle u v w &= u\\
            \sigma_2 \cap \triangle u v w &= v \\
            \sigma_3 \cap \triangle u v w &= w \\
            \sigma_4 \cap \triangle u v w &= \overline{uv}\\
            \sigma_5 \cap \triangle u v w &= \overline{vw}\\
            \sigma_6 \cap \triangle u v w &= \overline{uw}
        \end{align*}
        
        and for all other locally maximal cells $\tau \in X$, either

        \begin{enumerate}
            \item $\tau \cap \triangle u v w = u$,
            \item $\tau \cap \triangle u v w = v$,
            \item $\tau \cap \triangle u v w = w$,
            \item $\tau \cap \triangle u v w = \overline{uv}$,
            \item $\tau \cap \triangle u v w = \overline{vw}$,
            \item $\tau \cap \triangle u v w = \overline{uw}$,
            \item $\tau \cap \triangle u v w = \emptyset$.
        \end{enumerate}
        
        Then, there are exactly six partitions $P_1, P_2, P_3, P_4, P_5, P_6$ of $P_{NLM}$ which witness $\mathcal{T}$, where $\mathcal{T}$ is the edge spanning component corresponding to $\triangle u v w$.   Further, $P_1, P_2, P_3$ are labelled with $0$, $P_4, P_5, P_6$ with $3$ by \Cref{alg:P1class,alg:P2class}.
    \end{lemma}
    
    \begin{theorem}\label{thm:recover}
        Let $P$ be an $\varepsilon$-sample of an embedded $2$-complex $|X| \subset \mathbb{R}^n$ satisfying \Cref{as:embed}, and let $B$ be the graph obtained from \Cref{alg:numcells}. 

        Then, we can complete $B$ to be the incidence graph of $X$, to recover the abstract structure.
    \end{theorem}
    
    \begin{proof}
        From \Cref{prop:equivlmvert,prop:equivlmspanedge,prop:equivlmspantri}, we correctly identify the locally maximal components of $X$. It remains to show that we correctly learn the number of not locally maximal cells, and the incidence relationship. 
        
        For a locally maximal edge, we need to identify two vertices as its faces. To do so, we must identify which partition(s) of $P_{NLM}$ correspond to these vertices. 
    
        Take a spanning edge component $\mathcal{E}$. Then there is some locally maximal edge $\overline{uv}$ corresponding to $\mathcal{E}$. There are two cases to consider:
            
        \begin{enumerate}
            \item[A:] $\overline{uv}$ is disconnected from every other part of $X$, 
            \item[B:] $\overline{uv}$ is not disconnected every other part of $X$. 
        \end{enumerate}
            
        \underline{Case A:} From \Cref{prop:lmvertex,prop:lmedge,prop:lmtriangle,prop:nlmedge,prop:nlmvertex} and \Cref{as:embed}, there is a single partition $P_i \subset P_{NLM}$ which contains points $p$ such that $\mathcal{E} \cap B_{(R+\varepsilon)/\kappa + 3 \varepsilon}(p) \neq \emptyset$. Hence, $P_i$ contains samples $p$ such that either $\lVert v -p \rVert \leq \frac{3R}{2} + \varepsilon$ or $\lVert u - p \rVert \leq \frac{3R}{2} + \varepsilon$, and $P_i$ corresponds to $u$ and $v$. In this case, $P_i$ is labelled with $-1$ in \Cref{alg:P1class}. This occurs only when $\overline{uv}$ is disconnected from the rest of $|X|$; hence, we infer the two boundary vertices.
    
        \underline{Case B:} As $\overline{uv}$ is not disconnected, there is some locally maximal cell $\sigma \in X, \, \sigma \neq \overline{uv}$ such that either $u$ or $v$ is a vertex of $\sigma$. Without loss of generality, let $v \in \sigma$. For the vertices $u$ and $v$ let the set of locally maximal faces they see be $S(u)$ and $S(v)$, respectively.  As $X$ is a $2$-complex, and $\overline{uv}$ a locally maximal edge, $\sigma \notin S(u)$. Hence, there are two partitions, $P_u, P_v$, which correspond to the vertices $u$ and $v$, respectively. In this case, $P_u$ and $P_v$ are labelled with $0$ in \Cref{alg:P1class}.
    
        We now need to examine how we identify the faces of triangles. 
        
        For a triangle spanning component $\mathcal{T}$, let $\mathcal{P}_{\mathcal{T}}$ be the set of partitions $P_i$ of $P_{NLM}$ such that $d(\mathcal{T}, P_i) \leq 3 \varepsilon$. There are a few cases we need to consider to ensure we correctly recover the structure of $X$:

        \begin{enumerate}
            \item $|\mathcal{P}_{\mathcal{T}}|=1$,
            \item $|\mathcal{P}_{\mathcal{T}}|=2$,
            \item $|\mathcal{P}_{\mathcal{T}}|=3$,
            \item $|\mathcal{P}_{\mathcal{T}}|=4$,
            \item $|\mathcal{P}_{\mathcal{T}}|=5$,
            \item $|\mathcal{P}_{\mathcal{T}}|=6$.
        \end{enumerate}

        Let the weight $2$ node labelled with $\mathcal{E}$ be $t$.

        \underline{Case 1 $|\mathcal{P}_{\mathcal{T}}| = 1$:}
        Let $P_1$ be the single partition in $\mathcal{P}_{\mathcal{T}}$.  

        This can only occur if the triangle $\triangle u v w$ corresponding to $\mathcal{T}$ does not share any faces with another cell. Then, $P_1$ corresponds to three edges and three vertices and is correctly labelled with $7$ by \Cref{alg:P1class,alg:P2class}. Let the corresponding weight $1$ nodes of $B$ be $e_1, e_2, e_3$ and the weight $0$ nodes be $v_1,v_2,v_3$. We add an edge between  $t$ and $e_1,e_2,e_3, v_1,v_2,v_3$ and between the following pairs:
        
        \begin{equation*}
            (e_1, v_1), (e_1, v_2), (e_2, v_2), (e_2,v_3), (e_3,v_3), (e_3,v_1).
        \end{equation*}

        \underline{Case 2 $|\mathcal{P}_{\mathcal{T}}| = 2$:}
        Let $\mathcal{P}_{\mathcal{T}}=\{P_1, P_2\}$. 

        This can only occur if the triangle $\triangle u v w$ corresponding to $\mathcal{T}$ shares either a vertex, or an edge and two vertices with other triangles or locally maximal edges. Thus, either $P_1$ is labelled with $0$ and $P_2$ with $5$, or $P_1$ is labelled with $2$ and $P_2$ with $4$ by \Cref{alg:P1class,alg:P2class}.
        
        If $P_1$ has label $0$ and $P_2$ has label $5$, we find the weight $0$ node $v_1$ with label $P_1$ and the three weight $1$ nodes $e_1,e_2,e_3$ and two weight $0$ nodes $v_2, v_3$ with label $P_2$. Then, we add an edge between $t$ and each of $e_1,e_2,e_3, v_1,v_2,v_3$ and between the following pairs:
        
        \begin{equation*}
            (e_1, v_1), (e_1, v_2), (e_2, v_2), (e_2,v_3), (e_3,v_3), (e_3,v_1).
        \end{equation*}
        
        If $P_1$ has label $2$ and $P_2$ has label $4$, we find the weight $1$ note $e_1$ and two weight $0$ node $v_1,v_2$ with label $P_1$, the two weight $1$ nodes $e_2,e_3$ and one weight $0$ nodes $v_3$ with label $P_2$. We add an edge between $t$ and each of $e_1,e_2,e_3, v_1,v_2,v_3$ and between the following pairs:

        \begin{equation*}
            (e_1, v_1), (e_1, v_2), (e_2, v_2), (e_2,v_3), (e_3,v_3), (e_3,v_1).
        \end{equation*}

        \underline{Case 3 $|\mathcal{P}_{\mathcal{T}}| = 3$:}
        Let $\mathcal{P}_{\mathcal{T}}=\{P_1, P_2, P_3\}$. 

        This can only occur if the triangle $\triangle u v w$ corresponding to $\mathcal{T}$ shares either two vertices, or two vertices and an edge with other triangles or locally maximal edges. Thus, either $P_1$ and $P_2$ are labelled with $0$ and $P_2$ with $6$; or $P_1$ is labelled with $0$, $P_2$ with $1$ and $P_3$ with $4$; or $P_1$ is labelled $0$, $P_2$ with $2$ and $P_3$ with $9$.

        If $P_1, P_2$ have label $0$ and $P_3$ has label $6$, we find the weight $0$ node $v_1$ with label $P_1$, the weight $0$ node $v_2$ with label $P_2$, the three weight $1$ nodes $e_1,e_2,e_3$ and the weight $0$ node $v_3$ with label $P_3$. Then add an edge between $t$ and each of $e_1,e_2,e_3, v_1,v_2,v_3$ and between following pairs: 
        
        \begin{equation*}
            (e_1, v_1), (e_1, v_2), (e_2, v_2), (e_2,v_3), (e_3,v_3), (e_3,v_1).
        \end{equation*}

        If $P_1$ has label $0$, $P_2$ label $1$ and $P_3$ label $4$, we find the weight $0$ node $v_1$ with label $P_1$, the weight $0$ node $v_2$ with label $P_2$, weight $1$ node $e_1$ with label $P_2$, the weight $0$ node $v_3$ with label $P_3$, and the two weight $1$ nodes $e_2, e_3$ with label $P_3$. Then add an edge between $t$ and each of $e_1,e_2,e_3, v_1,v_2,v_3$ and between the following pairs: 

        \begin{equation*}
            (e_1, v_1), (e_1, v_2), (e_2, v_2), (e_2,v_3), (e_3,v_3), (e_3,v_1).
        \end{equation*}

        If $P_1$ has label $0$, $P_2$ label $2$ and $P_3$ label $9$, we find the weight $0$ node $v_1$ with label $P_1$, the weight $0$ node $v_2$ and weight $1$ node $e_1$ with label $P_2$, and the weight $1$ nodes $e_2, e_3$ and weight $0$ node $v_3$ with label $P_3$. Then add an edge between $t$ and each of $e_1,e_2,e_3, v_1,v_2,v_3$ and between the following pairs: 

        \begin{equation*}
            (e_1, v_1), (e_1, v_2), (e_2, v_2), (e_2,v_3), (e_3,v_3), (e_3,v_1).
        \end{equation*}

        \underline{Case 4 $|\mathcal{P}_{\mathcal{T}}| = 4$:} Let $\mathcal{P}_{\mathcal{T}} = \{ P_1, P_2, P_3, P_4 \}$.

        This can only occur if the triangle $\triangle u v w$ corresponding to $\mathcal{T}$ shares three vertices, or three vertices and an edge, or three vertices and two edges with other triangles or locally maximal edges. Thus, either $P_1, P_2$ and $P_3$ are labelled with $0$ and $P_4$ with $8$; or $P_1$ is labelled with $1$, $P_2, P_3$ with $0$ and $P_3$ with $9$; or $P_1$ with $3$, $P_2, P_3$ with $0$ and $P_4$ with $4$; or $P_1$ is labelled with $0$, $P_2, P_3$ with $1$, and $P_3$ with $3$ by \Cref{alg:P1class,alg:P2class}.

        If $P_1, P_2, P_3$ have label $0$ and $P_4$ has label $8$, find the weight $0$ node $v_1$ with label $P_1$, weight $0$ node $v_2$ with label $P_2$, weight $0$ node $v_3$ with label $P_3$, and the three weight $1$ nodes $e_1, e_2, e_3$ with label $P_4$. Then add an edge between $t$ and each of $e_1,e_2,e_3, v_1,v_2,v_3$ and between the following pairs: 

        \begin{equation*}
            (e_1, v_1), (e_1, v_2), (e_2, v_2), (e_2,v_3), (e_3,v_3), (e_3,v_1).
        \end{equation*}

        If $P_1$ has label $1$, $P_2, P_3$ have label $0$, and $P_4$ has label $9$, find the weight $0$ node $v_1$ and weight $1$ node $e_1$ with label $P_1$, weight $0$ node $v_2$ with label $P_2$, weight $0$ node $v_3$ with label $P_3$, and the two weight $1$ nodes $e_2, e_3$ with label $P_4$. Then add an edge between $t$ and each of $e_1,e_2,e_3, v_1,v_2,v_3$ and between the following pairs: 

        \begin{equation*}
            (e_1, v_1), (e_1, v_2), (e_2, v_2), (e_2,v_3), (e_3,v_3), (e_3,v_1).
        \end{equation*}

        If $P_1$ has $3$, $P_2, P_3$ label $0$ and $P_4$ label $4$;, find the weight $1$ node $e_1$ with label $P_1$, weight $0$ node $v_1$ with label $P_2$, weight $0$ node $v_2$ with label $P_3$, and the two weight $1$ nodes $e_2$ and weight $0$ node $e_3$ with label $P_4$. Then add an edge between $t$ and each of $e_1,e_2,e_3, v_1,v_2,v_3$ and between the following pairs: 

        \begin{equation*}
            (e_1, v_1), (e_1, v_2), (e_2, v_2), (e_2,v_3), (e_3,v_3), (e_3,v_1).
        \end{equation*}
        
        If $P_1$ has label $0$, $P_2, P_3$ have label $1$, and $P_4$ has label $3$, find the weight $0$ node $v_1$ with label $P_1$, weight $0$ node $v_2$ and weight $1$ node $e_1$ with label $P_2$, weight $0$ node $v_3$ and weight $1$ node $e_3$  with label $P_3$, and the two weight $1$ nodes $e_2$ with label $P_4$. Then add an edge between $t$ and each of $e_1,e_2,e_3, v_1,v_2,v_3$ and between the following pairs: 

        \begin{equation*}
            (e_1, v_1), (e_1, v_2), (e_2, v_2), (e_2,v_3), (e_3,v_3), (e_3,v_1).
        \end{equation*}
        
       \underline{Case 5 $|\mathcal{P}_{\mathcal{T}}| = 5$:}  Let $\mathcal{P}_{\mathcal{T}} = \{ P_1, P_2, P_3, P_4, P_5 \}$.

       This can occur if the triangle $\triangle u v w$ corresponding to $\mathcal{T}$ shares three vertices and two edges; or three vertcies and one edge with other triangles or locally maximal edges. Thus, $P_1, P_2$ are labelled with $0$, $P_3$ with $1$ and $P_4,P_5$ with $3$; or $P_1, P_2, P_3$ are labelled with $0$, $P_4$ with $3$ and $P_5$ with $9$ by \Cref{alg:P1class,alg:P2class}.

       If $P_1, P_2$ are labelled with $0$, $P_3$ with $1$ and $P_4,P_5$ with $3$ we find the weight $0$ node $v_1$ with label $P_1$, find the weight $0$ node $v_2$ with label $P_2$, find the weight $1$ node $e_1$ and weight $0$ node $v_3$ with label $P_3$, find the two weight $1$ nodes $e_2, e_3$ with label $P_4$, and the two weight $1$ nodes $e_2, e_3$ with label $P_5$. Then add an edge between $t$ and each of $e_1,e_2,e_3, v_1,v_2,v_3$ and between $e_i$ with label $P_i$ and $v_j$ with label $P_j$ if $d(P_i, P_j)$.

       If $P_1, P_2, P_3$ are labelled with $0$, $P_4$ with $3$ and $P_5$ with $9$ we find the weight $0$ node $v_1$ with label $P_1$, find the weight $0$ node $v_2$ with label $P_2$, find the weight $0$ node $v_3$ with label $P_3$, find the weight $1$ node $e_1$ with label $P_4$, and the two weight $1$ nodes $e_2, e_3$ with label $P_5$. Then add an edge between $t$ and each of $e_1,e_2,e_3, v_1,v_2,v_3$ and between $e_i$ with label $P_i$ and $v_j$ with label $P_j$ if $d(P_i, P_j)$.

    \underline{Case 6 $|\mathcal{P}_{\mathcal{T}}| = 6$:}  Let $\mathcal{P}_{\mathcal{T}} = \{ P_1, P_2, P_3, P_4, P_5, P_6 \}$.
        
        This can only occur if the triangle $\triangle u v w$ corresponding to $\mathcal{T}$ shares three vertices and two edges, or three vertices and three edges with other triangles or locally maximal edges. In either case, $P_1, P_2, P_3$ are labelled with $0$, $P_4, P_5, P_6$ with $3$ by \Cref{alg:P1class,alg:P2class}.

       So we find the weight $0$ node $v_1$ with label $P_1$, find the weight $0$ node $v_2$ with label $P_2$, find the weight $0$ node $v_3$ with label $P_3$, find the weight $1$ node $e_1$ with label $P_4$, the weight $1$ node $e_2$ with label $P_5$, and the weight $1$ node $e_3$ with label $P_6$. Then add an edge between $t$ and each of $e_1,e_2,e_3, v_1,v_2,v_3$ and between $e_i$ with label $P_i$ and $v_j$ with label $P_j$ if $d(P_i, P_j)$.

    In each of these $6$ cases, we have connected the weight $2$ node $t$ corresponding to the cell $\tau$ to each weight $1$ node $e$ corresponding to an edge $\sigma_e$ of $\tau$, as well as to each weight $0$ node $v$ corresponding to a vertex $\sigma_v$ of $\tau$. Further, in the process, we also connect the weight $1$ node $e$ and weight $0$ node $v$ if $\sigma_v$ is a vertex of $\sigma_e$.

    We have shown that the weight $2$ nodes of $B$ correspond bijectively to the triangles of $X$, the weight $1$ nodes of $B$ correspond bijectively to the edges of $X$, and the weight $0$ nodes of $B$ correspond bijectively to the vertices of $X$. We have also shown that for any pair of nodes $n_1, n_2$ with corresponding cells $\sigma_1, \sigma_2$, there an edge between them if and only if $\sigma_1 \subset \sigma_2$ or $\sigma_2 \subset \sigma_1$.

    Hence, $B$ is the incidence graph of $X$.
    \end{proof}

    In this article, we have presented a method for learning the abstract structure $X$ underlying an embedded $2$-simplicial complex $|X| = (X, \Theta)$  (satisfying \Cref{as:embed}) from an $\varepsilon$-sample $P$. For abstract $2$-complexes, modelling the embedding is future work. In particular, to modelling embeddings that are not linear or where we allow for cells of dimension $2$, which are not triangles (along the lines of CW-complexes), we need to develop the process for learning the faces of locally maximal cells further. 

\section{Future directions}
    There are several natural paths for the work in this article to be extended. In particular, removing the assumption that the maximal dimension of a cell in the complex is $2$ is a direct next step. It is also natural to consider how to modify the algorithm to allow for non-linear embeddings, in particular using semi-algebraic sets, as well as what happens when the noise is not assumed to be Hausdorff. These directions form a sort of 'orthogonal' basis for future research, as they can be thought of as independent problems, but when combined present a rather significant development towards learning stratified spaces.  

\backmatter

\bmhead{Acknowledgments}
This work was undertaken during the author's PhD, which was supported by Australian Federal Government Grant, 2019-2022, `Stratified Space Learning'. The author would like to thank Kate Turner, Chris Williams, Jonathan Spreer, Stephan Tillmann, Vanessa Robins, Vigleik Angeltveit, Martin Helmer, and James Morgan for very helpful discussions.

\section*{Declarations}

\subsection*{Competing interests}
    The work in this paper was undertaken during my PhD at the Australian National University and the University of Sydney, which was supported by an Australian Federal Government Grant, 2019-2022, \emph{Stratified Space Learning}.

\begin{appendices}

\section{Proofs of Geometric Lemmas}\label{sec:geom-lems-proofs}

\begin{proof-app}[Proof of \Cref{lem:sample-sphere-distance}]
        Consider $S_{R-\varepsilon}^{R+\varepsilon}(p) \cap L$ say $C$. Consider a point $q \in S_{R-\varepsilon}^{R+\varepsilon}(p)$ with $d(L,q) \leq \varepsilon$. Let $q_L$ be the projection of $q$ to $L$, $p_L$ the projection of $p$ to $L$.

        There are two cases we need to consider, 
        \begin{enumerate}
            \item $\lVert x - q_L \rVert \geq \lVert q_L - p_L \rVert$,
            \item $\lVert x - q_L \rVert < \lVert q_L - p_L \rVert$.
        \end{enumerate}

        \begin{figure}[h]
            \centering
            \begin{subfigure}[b]{0.4\textwidth}
            \begin{tikzpicture}[scale=0.5, x=1pt, y=1pt]
                \draw (0,0) -- (200,0);
                \draw (10,-20) node [anchor = north] {$p$}; 
                \draw (90,20) node [anchor = south] {$q$};
                \draw (90,0) node [anchor = north, yshift=-8pt] {$q_L$};
                \draw (180,0) node [anchor = north] {$x$};
                \draw (10,0) node [anchor = south] {$p_L$};
                \draw (10,-20) -- (180,0);
                \draw (10,-20) -- (90,20);
                \draw (180,0) -- (90,20);
                \draw (90,20) -- (90,0);
                \draw (10,-20) -- (10,0);
            \end{tikzpicture}
            \caption{When $p$ and $q$ are on the same side of $x$}\label{fig:sample-sphere-distance-larger}
            \end{subfigure}
            \hspace{1cm}
            \begin{subfigure}[b]{0.4\textwidth}
            \begin{tikzpicture}[scale=0.5, x=1pt, y=1pt]
                \draw (0,0) -- (200,0);
                \draw (10,-20) node [anchor = north] {$p$}; 
                \draw (180,20) node [anchor = south] {$q$};
                \draw (180,0) node [anchor = north] {$q_L$};
                \draw (70,0) node [anchor = south] {$x$};
                \draw (10,0) node [anchor = south] {$p_L$};
                \draw (10,-20) -- (70,0);
                \draw (10,-20) -- (180,20);
                \draw (70,0) -- (180,20);
                \draw (180,20) -- (180,0);
                \draw (10,-20) -- (10,0);
            \end{tikzpicture}
           \caption{When $p$ and $q$ are on different sides of $x$}\label{fig:sample-sphere-distance-smaller}
           \end{subfigure}
            \caption{}
            \label{fig:sample-sphere-distance}
        \end{figure}
        
        We begin with case 1.

        We want to bound $\lVert x - q \rVert$. Note that 
        \begin{align*}
            \lVert q-x \rVert^2 &= \lVert q-q_L \rVert^2 + \lVert q_L - x \rVert^2,\\
            \lVert q_L - x \rVert &= \lVert p_L - x \rVert - \lVert p_L - q_L \rVert, \\
            \lVert p_L - q_L \rVert^2 & = \lVert q - p \rVert^2 - \left( \lVert p-p_L \rVert + \lVert q-q_L \rVert \right)^2,\\
            \lVert p_L - x \rVert^2 &= \lVert x - p \rVert^2 - \lVert p_L - p \rVert^2.
        \end{align*}
        Hence, 
        \begin{align*}
            &\lVert q-x \rVert^2  \\
                                    &= \lVert q - q_L \rVert^2 + \left( \lVert p_L - x \rVert - \lVert p_L - q_L \rVert \right)^2 \\
                                    &= \lVert q - q_L \rVert^2 + \left( \sqrt{\lVert q - p \rVert^2 - \lVert p_L - p \rVert^2} - \sqrt{\lVert q - p \rVert^2 - \left( \lVert p-p_L \rVert + \lVert q-q_L \rVert \right)^2}\right)^2 \\
                                    &= \lVert q - q_L \rVert^2 + \\
                                    & \, \left( \sqrt{\lVert q - p \rVert^2 - \lVert p_L - p \rVert^2} - \sqrt{\lVert q - p \rVert^2 - \lVert p-p_L \rVert^2 - \left(\lVert q-q_L \rVert^2 + \lVert p-p_L \rVert \lVert q-q_L \rVert\right)}\right).
        \end{align*}
        \newcommand{\A}{\lVert q - p \rVert^2 - \lVert p-p_L \rVert^2}
        \newcommand{\B}{\lVert q-q_L \rVert^2 + \lVert p-p_L \rVert \lVert q-q_L \rVert}
        Let 
        \begin{align*}
            A &= \A, \\
            B &= \B.
        \end{align*}
        As
        \begin{align*}
            \lVert q - p \rVert     &\leq  R, \\          
            \lVert p - p_L \rVert   &\leq \frac{R}{2},\\
            \lVert q - q_L \rVert   &\leq \varepsilon,
        \end{align*}
        we have
        \begin{align*}
            A &>(R-\varepsilon)^2 - \varepsilon^2 \\
            B &< 3 \varepsilon^2
        \end{align*}
        and  so $A > \frac{4B}{3}$.
        Then
        \begin{align*}
            \frac{A B}{3}           & > \frac{4B^2}{9} \\
            A^2 - AB                &> A^2 - \frac{4 AB}{3} + \frac{4B^2}{9} \\
            \sqrt{A(A-B)}           &> A - \frac{2B}{3} \\
            - 2 \sqrt{A(A-B)}       &< -2 A + \frac{4B}{3} \\
            2A - B- 2\sqrt{A(A-B)}  &< \frac{B}{3} \\
            \left(\sqrt{A} - \sqrt{ A - B} \right)^2 &< \frac{B}{3}
        \end{align*}

        Recall $A > \frac{4B}{3}$, thus
    
        \begin{align*}
             \lVert q-x \rVert^2    &= \lVert q - q_L \rVert + \left( \sqrt{A} - \sqrt{A-B} \right)^2 \\
                                    &\leq \varepsilon^2 + \frac{B}{3} \\
                                    & \leq 2 \varepsilon
        \end{align*}

        A similar calculation in case 2 gives a smaller bound, so \[\lVert q- x \rVert \leq \sqrt{2} \varepsilon.\]
    \end{proof-app}

    \begin{proof-app}[Proof of \Cref{lem:antipodal-flat}]
        First, let $p_H$ be the projection of $p$ to $H$, and note that $\lVert p_H -p \rVert \leq \varepsilon$. Take $q_1 \in S_{R-\varepsilon}^{R+\varepsilon}(p)\cap P$. Let $x_1$ be the point in $\partial B_R(p) \cap H$ closest to $q_1$, and $q_H$ the projection of $q_1$ to $H$. Note that $p_H, q_H, x_1$ are co-linear, lying on the ray $L$ from $p_H$, and $\lVert q_1 - q_H \rVert \leq \varepsilon$. By \Cref{lem:sample-sphere-distance}, $\lVert q_1 - x_1 \rVert \leq \sqrt{2} \varepsilon.$
        
        As $H \cap \partial B_R(p)$ is a circle with radius $\sqrt{R^2 - \lVert p_H -p \rVert^2}$, there is a point $x_2 \in H \cap \partial B_R(p)$ such that $\lVert x_2 - x_1 \rVert = 2 \sqrt{R^2 - \lVert p_H -p \rVert^2}$. As $d_H(p,H) \leq \varepsilon$, we have 
        \begin{equation*}
            \lVert x_2 - x_1 \rVert \geq 2\sqrt{R^2 - \varepsilon^2},
        \end{equation*}
        and as $d_H(P,H) \leq \varepsilon$, there is $q_1 \in P$ with $\lVert q_1 - x_1 \rVert \leq \varepsilon$. Hence
        \begin{equation*}
            \Vert q_2 - q_1 \rVert \geq 2\sqrt{R^2- \varepsilon^2} - (1+\sqrt{2})\varepsilon.
        \end{equation*}
    \end{proof-app}

    \begin{proof-app}[Proof of \Cref{lem:antipodal-not-flat}]
        First, let $H_1'$ be the half plane containing $H_1$ with bounding line $L'$ such that $D(L,L') = \varepsilon$, $p_H$ be the projection of $p$ onto $H_1'$ and $p_L$ the projection of $p$ to $L$. Then take $x_1 \in H_1$ such that $\lVert p - x_1 \rVert = R$ and $p_H, p_L$ and $x_1$ are co-linear. Take $q_1 \in P$ with $\lVert q_1 - x_1 \rVert \leq \varepsilon$, so $q_1 \in S_{R-\varepsilon}^{R+\varepsilon}(p) \cap P$. 
        
        Let $q_2$ be a point in $S_{R-\varepsilon}^{R+\varepsilon}(p) \cap P$. There are two cases to consider: $d(q_2, H_1') \leq \varepsilon$ and $d(q_2, H_2) \leq \varepsilon$.
        
        If $d(q_2, H_1') \leq \varepsilon$, take $x_2 \in \partial B_R(p) \cap H_1'$  such that $x_2, p_H$ and the projection of $q_2$ to $H_1'$ are co-linear. Then by \Cref{lem:sample-sphere-distance} $\lVert q_2 - x_2 \rVert \leq \sqrt{2} \varepsilon$. 
        
        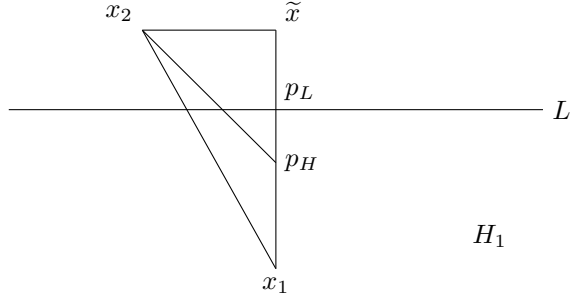
\begin{figure}[h!]
            \centering
            \begin{tikzpicture}[x=1pt,y=1pt]
                \draw (-100,0) -- (100,0); 
                \draw (100,0) node [anchor =west] {$L$};
                \draw (80,-40) node [anchor = north] {$H_1$};
                \draw (0,0) node [anchor = south west] {$p_L$};
                \draw (0,-60) node [anchor = north] {$x_1$};
                \draw (-50, 30) node [anchor = south east] {$x_2$};
                \draw (0,-20) node [anchor = west] {$p_H$};
                \draw (-50,30) -- (0,-20);
                \draw (0,0) -- (0,-60);
                \draw (0,-60) -- (-50,30);
                \draw (-50,30) -- (0,30);
                \draw (0,30) -- (0,0);
                \draw (0,30) node [anchor = south west] {$\widetilde{x}$};
            \end{tikzpicture}
            \caption{Understanding the behaviour of points near the common boundary of two half-planes}\label{fig:flat-angle}
        \end{figure}

        Consider the triangle formed by $x_1, p_h, x_2$.  
        By assumption,
        \begin{align*}
            \lVert \widetilde{x} - p_H \rVert &< \frac{R}{2} < R -7 \varepsilon, \\
            \lVert x_2 - p_H \rVert = \lVert x_1 - p_H \rVert &\leq R.
        \end{align*}
        Let $\widehat{R} = \sqrt{R^2 - \lVert p_H - p\rVert^2}$.
        Then
        \begin{align*}
            \lVert \widetilde{x} - p_H \rVert &< \widehat{R} \\
            \lVert \widetilde{x} - p_H \rVert &< \widehat{R} - 6 \varepsilon \\
            2\widehat{R} \lVert \widetilde{x} - p_H \rVert &< 2 \widehat{R}^2 - 12\widehat{R} \varepsilon \\
            2\widehat{R}^2 + 2 \widehat{R} \lVert \widetilde{x} -p_H \rVert &< 4 \widehat{R}^2 - 4(1+\sqrt{2}) \widehat{R}\varepsilon + (1+\sqrt{2}) \varepsilon \\
            2 \widehat{R}^2 + 2 \widehat{R} \left( \frac{\lVert \widetilde{x} - p_H \rVert}{\widehat{R}}\right) &< \left( 2\widehat{R} - (1 + \sqrt{2})\varepsilon \right)^2.
        \end{align*}
        Further, 
        \begin{align*}
            2 \widehat{R}^2 + 2 \widehat{R} \left( \frac{\lVert \widetilde{x} - p_H \rVert}{\widehat{R}}\right)&= \lVert x_1 - p_H \rVert^2 + \lVert x_2 - p_H \rVert^2 + 2 \lVert x_1 -p_H \rVert \lVert x_2 -p_h \rVert \cos \angle x_2 p_h \widetilde{x} \\
            &= \lVert x_1 - p_H \rVert^2 + \lVert x_2 - p_H \rVert^2 - 2 \lVert x_1 -p_H \rVert \lVert x_2 -p_h \rVert \cos \angle x_2 p_h x_1 \\
                &= \lVert x_2 - x_1 \rVert^2,
        \end{align*}
        so
        \begin{align*}
            \lVert x_2 - x_1 \rVert & < \sqrt{R^2 - \lVert p_H - p\rVert^2} - (1 + \sqrt{2})\varepsilon.
        \end{align*}
        which implies 
        \begin{align*}  
            \lVert q_2 - q_1 \rVert &< 2\sqrt{R^2 - \varepsilon^2} - (2+2\sqrt{2}) \varepsilon. 
        \end{align*}

        \begin{figure}[h!]
            \centering
            \begin{tikzpicture}[x=1.5pt,y=1.5pt]
                \draw (25,52.5) node [anchor = south] {$x_2$};
                \draw (-100,-20) -- (100,20);
                \draw (0,0) node [anchor = south east] {$p_L$};
                \draw (29.5,-40) node [anchor = north] {$x_1$};
                \draw (29.5,-40) -- (0,0);
                \draw (25,52.5) -- (0,0);
                \draw (25,52.5) -- (29.5,-40);
                \draw (10,-15) node [anchor = east] {$p_H$};
            \end{tikzpicture}
            \caption{$d(q_2, H_2) \leq \varepsilon$}
        \end{figure}
        Now assume $d(q_2, H_2) \leq \varepsilon$. Let $H_2'$ be the half-plane which contains $H_2$ and has boundary $L'$ with $d(L,L') = \varepsilon$. As $d(q_2, H_2) \leq \varepsilon$, then there is $x_2 \in \partial B_R(p) \cap H'_2$ with $\lVert q_2 - x_2 \rVert \leq \sqrt{2} \varepsilon$. Hence, \[ \lVert x_1 -x_2 \rVert \geq 2\sqrt{R^2- \varepsilon^2} - (2+ 2\sqrt{2}) \varepsilon.\] If $x_2 \in H_2' \setminus H_2$, then by a similar argument to above, \[ \lVert x_1 -x_2 \rVert \geq 2\sqrt{R^2- \varepsilon^2} - (2+ 2\sqrt{2}) \varepsilon.\]
        
        If $x_2 \in H_2 \subsetneq H_2 '$, by the cosine rule we have
        \begin{equation*}
            \lVert x_2 - x_1 \rVert^2 = \lVert x_2 - p_L \rVert^2 + \lVert x_1 - p_L \rVert^2 - 2 \lVert x_2 - p_L \rVert  \lVert x_1 - p_L \rVert \cos \angle x_1 p_L x_2.
        \end{equation*}
        Note $\lVert x_1 -p_L \rVert = \lVert x_1 -p_H  \rVert + \lVert p_H - p_L \rVert$, and $\lVert x_2 - x_1 \rVert$ is bounded above by the case when 
        \begin{align*}
            \angle x_1 p_L x_2 &= \alpha, \\
            \lVert x_2 - p_L \rVert &= R+2 \varepsilon, \\
            \lVert x_1 -p_L \rVert &= \lVert x_1 -p_H  \rVert + \lVert p_H - p_L \rVert = \frac{3R}{2} + \varepsilon.
        \end{align*}
        Hence, we have 
        \begin{align*}
            \lVert x_2 - x_1 \rVert & < (R+2 \varepsilon)^2 + \left( \frac{3R}{2} + \varepsilon \right)^2 - (R+2 \varepsilon) \left( \frac{3R}{2} + \varepsilon \right) \cos \alpha.
        \end{align*}

        By assumption, $\alpha \in \left(0, \Psi(\varepsilon,R)\right)$, and so 
        \begin{align*}
            \lVert x_2 - x_1 \rVert & < 2\sqrt{R^2- \varepsilon^2} - \left( 2+ 2\sqrt{2} \right)\varepsilon,
        \end{align*}
        which implies that 
        \begin{align*}
            \lVert q_2 - q_1 \rVert & < 2\sqrt{R^2- \varepsilon^2} - \left( 1+ \sqrt{2} \right) \varepsilon.
        \end{align*}

        Hence, there is a $q_1 \in S_{R-\varepsilon}^{R+\varepsilon}(p) \cap P$ such that for all $q_2 \in S_{R-\varepsilon}^{R+\varepsilon}(p) \cap P$
         \begin{align*}
            \lVert q_2 - q_1 \rVert & < 2\sqrt{R^2- \varepsilon^2} - \left( 1+ \sqrt{2} \right) \varepsilon.
        \end{align*}
    \end{proof-app}

    \begin{proof-app}[Proof of \Cref{lem:diamline}]
        By \Cref{lem:sample-sphere-distance}, every $q \in S_{R-\varepsilon}^{R+\varepsilon}(p) \cap P$ is with in $\sqrt{2} \varepsilon$ of the point $x$ in $L$ with $\lVert x - p \rVert = R$. Hence, $\left(S_{R-\varepsilon}^{R+\varepsilon}(p) \cap P \right)^{\frac{3\varepsilon}{2}}$ consists of a single connected component and it has diameter less than $2\sqrt{2}\varepsilon$. 
    \end{proof-app}

    \begin{proof-app}[Proof of \Cref{lem:diamtri}]
        As $\lVert p - z \rVert \leq \frac{R-\varepsilon}{2}$, the intersection $S_{R-\varepsilon}^{R+\varepsilon}(p) \cap T$ is not empty, connected, and $\mathcal{H}_1\left( S_{R-\varepsilon}^{R+\varepsilon}(p) \cap T\right) = 0$. Further, the intersections $S_{R-\varepsilon}^{R+\varepsilon}(p) \cap L_1$ and $S_{R-\varepsilon}^{R+\varepsilon}(p) \cap L_2$ are also connected.
        
        Now, let $x_1$ be the point on $L_1$ with $\lVert q_1 - p \rVert = R$ and let $x_2$ be the point on $L_2$ with $\lVert x_2 - p \rVert = R$. As $S_{R-\varepsilon}^{R+\varepsilon}(p) \cap T$ is path connected, $x_1$ and $x_2$ are path connected in $T$. 
        
        \begin{figure}
            \centering
            \begin{tikzpicture}[x=1pt,y=1pt]
                \draw (0,0) node [anchor = east] {$p$};
                \draw (50,-50) node [anchor = north] {$x_1$};
                \draw (50,50) node [anchor = south] {$x_2$};
                \draw (20,0) node [anchor = east] {$z$};
                \draw (20,0) node [anchor = west] {$\alpha$};
                \draw (0,0) -- (50,-50);
                \draw (0,0) -- (50,50);
                \draw (20,0) -- (50,-50);
                \draw (20,0) -- (50,50);
                \draw (50,50) -- (50,-50);
            \end{tikzpicture}
            \caption{Bounding the diameter of a set of points}
            \label{fig:diamtri}
        \end{figure}
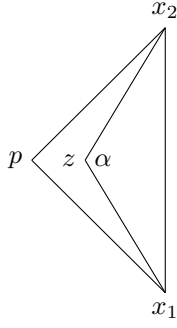
        
        Consider the triangle $\triangle x_1 p x_2$, we have
        
        \begin{align*}
            \lVert x_1 - x_2 \rVert^2   &= \lVert x_1 - z\rVert^2 + \lVert x_2 -z \rVert^2 - 2 \lVert x_1 - z \rVert \lVert q_2 - z \rVert \cos \alpha \\
                                        &\geq \left( R - \frac{R-\varepsilon}{2}\right)^2 + \left( R - \frac{R-\varepsilon}{2}\right)^2 - 2\left( R - \frac{R-\varepsilon}{2}\right)^2 \cos \alpha \\
                                        &= 2 \left( \frac{R+\varepsilon}{2}\right)^2(1 - \cos \alpha).
        \end{align*}     

        Now, as $d_H(P, T) \leq \varepsilon$, there are points $q_1, q_2 \in P$ with \[\lVert q_1 - x_1\rVert, \lVert q_2 - x_2 \rVert \leq \varepsilon.\] Then by the triangle inequality
        
        \begin{align*}
            \lVert q_1 - q_2 \rVert^2  &= 2 \left( \frac{R+\varepsilon}{2}\right)^2(1 - \cos \alpha) - 2 \varepsilon\\
                    & > 2 \sqrt{2} \varepsilon, \, \text{ as } \alpha \in \left [ \frac{\pi}{6}, \pi \right )
        \end{align*} 
    \end{proof-app}

\section{Proof of Correctness Lemmas}\label{sec:correctness-lems-proofs}

\begin{proof-app}[Proof of \Cref{prop:spanningedge}]
        Let $C$ be a connected component of $\check{\mathcal{C}}_{\frac{3\varepsilon}{2}}(P_{LM,1})$ which spans a locally maximal edge $\overline{uv}$, with midpoint $m_{uv}$. Then, there is a sample $p_m \in C$ such that $\lVert p_m - m_{uv} \rVert \leq \varepsilon$. 
        
        To show that $\mathcal{D}(C) \geq \frac{9R}{2}$, we show that there are two points $x_u, x_y \in \overline{uv}$ such that 
        \begin{enumerate}
            \item $\lVert u - x_u \rVert > \frac{3R}{2} + 2 \varepsilon$,
            \item $\lVert v - x_v \rVert > \frac{3R}{2} + 2 \varepsilon$,
            \item $\lVert x_u - x_v \rVert \geq \frac{3R}{2}$.
        \end{enumerate}

        Without loss of generality, we show that $x_u$ exists, and 
       
       \begin{equation*}
            \lVert x_u - m_{uv} \rVert \geq \frac{3R}{4}+\varepsilon.
        \end{equation*} 
        
        By \Cref{as:embed}, $\lVert u - v \rVert \geq 6(R+\varepsilon)$. As $\overline{uv}$ is a line segment, for all $\eta \in [0, \frac{9R}{4} + 3 \varepsilon]$ there is a point $x_{\eta} \in \overline{uv}$ such that $\lVert x_{\eta} - u \rVert = \eta$. Letting $\eta = \frac{3R}{2} + 2\varepsilon$, there is a point, namely $x_u$ such that $\lVert x_u - u \rVert = \frac{3R}{2} + 2\varepsilon$. As $P$ is an $\varepsilon$-sample, there is a sample $p_u$ such that $\lVert x_u - p_u \rVert \leq \varepsilon$, and hence $\lVert p_u - u \rVert > \frac{3R}{2}+\varepsilon$. Thus, the $(\varepsilon,R)$-local structure of $P$ at $p_u$ is maximal of dimension $1$.

        We can repeat this argument for all $\eta \in [\frac{3R}{2}+2 \varepsilon, \frac{9R}{4}+3\varepsilon]$, and obtain a path of points $x_{\eta} \in \overline{uv}$ and samples $p_{\eta} \in P$ connecting $p_u$ to $p_m$. 

        This also holds when we replace $u$ with $v$, and hence we have $p_u$ and $p_v$. Finally, we have

        \begin{align*}
            \lVert p_u - p_v \rVert &\geq \rVert  x_u - x_v \rVert - \lVert p_u - x_u \rVert - \lVert p_v - x_v \rVert \\
                                    &\geq \frac{3R}{2} - 2\varepsilon, \\
        \end{align*}

        and hence $\mathcal{D}(C) \geq \frac{3R}{2} - 2\varepsilon.$
        
        Now, we show that if $\mathcal{D}(C) \geq \frac{3R}{2} - 2 \varepsilon$, then $C$ spans some locally maximal edge.

        If $\mathcal{D}(C) \geq \frac{3R}{2} - 2 \varepsilon$, then there are points $p,q \in C$ with \[\lVert p - q \rVert \geq \frac{3R}{2} - 2 \varepsilon.\] As $P$ is an $\varepsilon$-sample of $|X|$, there are points $x_p, x_q \in |X|$, with 
        
        \begin{align*}
            \lVert x_p - p \rVert,\, \lVert x_q - q \rVert & \leq \varepsilon.
        \end{align*}

        Let $m_{pq}$ be the midpoint of $x_p$ and $x_q$. As $p$ and $q$ are in the same connected component of $\check{\mathcal{C}}_{\frac{3\varepsilon}{2}}\left(P_{LM,1}\right)$, we know there is a sequence of points $\{ q_i \}_{i=0}^m$ with $q_0 = p, \, q_m = q$ and for all $0 < i \leq m$, $\lVert q_i - q_{i-1} \rVert \leq 3\varepsilon$. Again, $P$ is an $\varepsilon$-sample of $|X|$, and as $q_i \in P_{LM,1}, \, \forall 0 \leq i \leq m$, for each $q_i$ there is some $x_i \in |X|$ which is on a locally maximal edge, and $\lVert q_i - x_i \rVert \leq \varepsilon$. From \Cref{as:embed} and \Cref{prop:lmedge}, there is a locally maximal edge, say $\overline{uv}$ such $x_i \in \overline{uv}, \forall 0 \leq i \leq m$. Let the midpoint of $\overline{uv}$ be $x_{uv}$.
        
        We now split into two cases:
        \begin{enumerate}
            \item[I] there is some $i$ such that $x_i = x_{uv}$,
            \item[II] for all $i$ we have $x_i \neq x_{uv}$.
        \end{enumerate}
        
        \underline{Case I:} The connected component $C$ is a spanning connected component, as it contains a sample which is within $\varepsilon$ of the midpoint $x_{uv}$ of the locally maximal edge $\overline{uv}$.

        \underline{Case II:} As no $q_i$ is within $\varepsilon$ of $m_{uv}$, we know that $q_i \, \forall 0\leq i \leq m$ are on the same side of $\overline{uv}$. That is, for all $q_i$, without loss of generality,
        
        \begin{align*}
            \lVert q_i - x_{uv} \rVert \leq \lVert q_i - u \rVert &\geq \frac{3\sqrt{3}}{2} R + 3\varepsilon \\
            \lVert q_m - v \rVert &\geq \frac{3R}{2} + \varepsilon.
        \end{align*}
        
        Further, assume that 
        \begin{equation*}
            \lVert q_0 - m_{uv} \rVert \leq \lVert q_m - x_{uv} \rVert.
        \end{equation*}

        There is another sequence of points $\{x'_j\}_{j=0}^{m'}$ in $\overline{uv}$ with $x'_0 = x_m$ and $x'_{m'}=x_{uv}$, and for $0 < j \leq m'$
        
        \begin{equation*}
            \lVert x'_j - x'_{j-1} \rVert \leq \varepsilon.
        \end{equation*}

        Then, there exists $q'_j \in P$ with 
        \begin{align*}
            \lVert q'_j - x'_j \rVert &\leq \varepsilon \\
            \lVert q'j - q'_{j-1} \rVert &\leq \varepsilon \, \forall 0 < j \leq m' \\
            \lVert q'_j - v \rVert &\geq \frac{3R}{2} + \varepsilon.
        \end{align*}
        
        By \Cref{as:embed} and \Cref{prop:lmedge}, $q'_j \in P_{LM,1}$ for all $ 0 \leq j \leq m'$. Hence, each $q'_j$ is in the same connected component $C$ as $q_m$. 

        Thus, $C$ contains a sample $q'_{m'}$ which is within $\varepsilon$ of the midpoint of the locally maximal edge $\overline{uv}$. Hence, $C$ is a spanning connected component.

        Thus a component $C$ of $\check{\mathcal{C}}_{\frac{3\varepsilon}{2}}\left( P_{LM,1}\right)$ spans a locally maximal edge $\overline{uv}$ if and only if $\mathcal{D}(C) \geq \frac{3R}{2} - 2\varepsilon$.
    \end{proof-app}

    \begin{proof-app}[Proof of \Cref{prop:spanningtri}]
        First, let $C$ be a connected component of $\check{\mathcal{C}}_{\frac{3\varepsilon}{2}}$ which spans some triangle $\triangle u v w$ with midpoint $m$. As $P$ is an $\varepsilon$-sample of $X$, there is a sample $p_m \in P$ with $\lVert p_m - m \rVert \leq \varepsilon$. As the radius of the inscribed circle of $\triangle u v w$ is at least $2R+ 3\varepsilon$, $m$ is at least $2R+ 3\varepsilon$ from $\partial \triangle u v w$. Thus, $d(p_m, \partial \triangle u v w) \geq 2R + 2 \varepsilon$.

        Hence, for all $q \in B_{\frac{R}{2}+2\varepsilon}(p) \cap P$, $d(q, \partial \triangle u v w) \geq \frac{3R}{2} + \varepsilon$, and so $q \in P_{LM,2}$.

        Now, take $p \in P_{LM,2}$ such that $B_{\frac{R}{2} + \varepsilon}(p) \cap P \subset P_{LM,2}$. Then, there is some triangle $\triangle u v w$ with $d(\triangle u v w, p) \leq \varepsilon$. As $p \in P_{LM,2}$, we know that $d(\partial \triangle u v w, p) > \frac{R}{2}- \varepsilon$. By assumption, for all $q \in B_{\frac{R}{2} + \varepsilon}(p) \cap P$, we have $d(\partial \triangle u v w, q) > \frac{R}{2}- \varepsilon$. Recall that $P$ is an $\varepsilon$-sample of $|X|$, so there is a point $x \in X$ such that $\lVert p -x \rVert \leq \varepsilon$. As $\triangle u v w$ is convex, and every $B_{\frac{R}{2} + \varepsilon}(p) \cap P \subset P_{LM,2}$, we have
        \begin{align*}
            d(\partial \triangle u v w, x)  &\geq \frac{R}{2} +2\varepsilon + \frac{R}{2} - 2\varepsilon = R.
        \end{align*}
        Hence, is a point $y \in B_{\frac{R}{2}  + 2\varepsilon}(p) \cap \triangle u v w$ with 
        \begin{align*}
           d(\partial \triangle u v w, y)  &\geq \frac{R}{2} +2\varepsilon.
        \end{align*}
        and a sample $q \in B_{\frac{R}{2}+2\varepsilon}(p) \cap P_{LM,2}$  with $\lVert q - y\rVert \leq \varepsilon$.
        
        Now, we can construct a sequence of points $\{ y_i\}_{i=0}^{m} \subset \triangle u v w$ such that $\lVert y_i - y_{i-1} \rVert \leq \varepsilon$ for $1 \leq i \leq m$, and $y_0 = x, \, y_m =y$. Further, for each $y_i$ there is a $q_i \in P$ with $\lVert q_i - y_i \rVert \leq \varepsilon$, and $q_i \in P_{LM,2}$. Note, that this means $p$ and $q_m$ are in the same connected component $C$ of $\check{\mathcal{C}}_{\frac{3\varepsilon}{2}}\left(P_{LM,2}\right)$.

        Finally, we construct a similar sequence of points $\{ \widetilde{y}_j\}_{j=0}^{\widetilde{m}}$ in $|X|$ from $y$ to $m_{\triangle u v w}$ with $\widetilde{y_0}=y, \, \widetilde{y}_{\widetilde{m}}=m_{\triangle u v w}$. Again, for each $\widetilde{y}_j$, there is a $\widetilde{q}_j \in P$ with $\lVert \widetilde{y}_j - \widetilde{q}_j \rVert \leq \varepsilon$ and $\widetilde{q}_j \in P_{LM,2}$. Hence, the $\widetilde{q}_j$ are in the same connected component of $\check{\mathcal{C}}_{\frac{3\varepsilon}{2}}\left(P_{LM,2}\right)$, and further, this connected component is $C$.
    \end{proof-app}

    \begin{proof-app}[Proof of \Cref{prop:equivlmvert}]
        Let $V_{LM}$ be the set of locally maximal vertices of $X$. Let $v$ be a locally maximal vertex, then by \Cref{prop:lmvertex}, $\forall p \in P$ with $\lVert p - v \rVert \leq 4 \varepsilon$, $p \in P_{LM,0}$. In fact, by \Cref{as:embed}, any $p \in P$ with $\lVert p - v \rVert \leq 4 \varepsilon$ is actually within $\varepsilon$ of $v$. Hence, every $p \in P_{LM,0}$ within $\varepsilon$ of $v$ are in the same connected component of $\check{\mathcal{C}}_{\frac{3\varepsilon}{2}}\left(P_{LM,0}\right)$.

        Now, take a connected component $C$ of $\check{\mathcal{C}}_{\frac{3\varepsilon}{2}}\left(P_{LM,0}\right)$. Each $p \in C$ is within $\varepsilon$ of a locally maximal vertex $v_p$ of $X$. By \Cref{as:embed}, every locally maximal vertex $v$ is at least $5 \varepsilon$ away from any other cell of $X$, and hence $\forall p \in C$, $v_p$ is the same.

        Hence, the connected components of $\check{\mathcal{C}}_{\frac{3\varepsilon}{2}}\left(P_{LM,0}\right)$ correspond bijectively to the locally maximal vertices of $X$.
    \end{proof-app}

\begin{proof-app}[Proof of \Cref{prop:equivlmspanedge}]
       Let $E_{LM} \subset E$ be the set of locally maximal edges in $X$. By \Cref{prop:spanningedge}, a connected component $C$ of $\check{\mathcal{C}}_{\frac{3\varepsilon}{2}} \left(P_{LM,1}\right)$ spans an edge $\overline{uv}$ if and only if it contains a sample $p$ within $\varepsilon$ of the midpoint $m$ of $\overline{uv}$.

        If a connected component $C$ is a spanning component, then there is some locally maximal edge $\overline{uv}$ with midpoint $m$ such that there is a sample $p \in C$ with $\lVert m - p \rVert \leq \varepsilon$.

        For any locally maximal $\overline{uv} \in E_{LM}$ with midpoint $m$, there is some sample $p \in P$ such that $\lVert m -p \rVert \leq \varepsilon$. Then, by \Cref{as:embed,prop:lmedge}, $p \in P_{LM,1}$, and so there is some spanning connected component $C_{\overline{uv}}$ in $\check{\mathcal{C}}_{\frac{3\varepsilon}{2}} \left(P_{LM,1}\right)$.

        Now, consider a locally maximal edge $\overline{uv'}, \, v' \neq v$, and take samples $p,q \in P_{LM,2}$ such that $d(\overline{uv}, p ), \, d(\overline{uv'}, q) \leq \varepsilon$. By \Cref{as:embed}, $\lVert p - q \rVert > 6 \varepsilon$, and so $p$ and $q$ are in different connected components of $\check{\mathcal{C}}_{\frac{3\varepsilon}{2}}(P_{LM,1})$. 
       
        Finally, consider a locally maximal edge $\overline{u'v'}$ such that $\overline{uv}$ and $\overline{u'v'}$ do not have a common vertex. Take samples $p,q \in P_{LM,2}$ such that 
        \begin{equation*}
            d(\overline{uv}, p ), \, d(\overline{u'v'}, q) \leq \varepsilon.
        \end{equation*}
        Again, by \Cref{as:embed}, $\lVert p - q \rVert > 6 \varepsilon$, and so $p$ and $q$ are in different connected components of $\check{\mathcal{C}}_{\frac{3\varepsilon}{2}}(P_{LM,1})$. 
       
        Hence, each connected component $C$ only consists of samples $p$ with $d(\overline{uv}, p) \leq \varepsilon$ for a single locally maximal edge $\overline{uv}$. 

        Thus, the spanning connected components of $\check{\mathcal{C}}_{\frac{3\varepsilon}{2}}\left(P_{LM,2}\right)$ are in bijection with the locally maximal edges of $|X|$.
    \end{proof-app}

\begin{proof-app}[Proof of \Cref{prop:equivlmspantri}]
        From \Cref{prop:spanningtri}, a connected component $C$ of $\check{\mathcal{C}}_{\frac{3\varepsilon}{2}}\left(P_{LM,2}\right)$ spans a triangle $\triangle u v w$ if and only if it contains a sample $p$ within $\varepsilon$ of the midpoint $m$ of $\triangle u v w$.

        As $P$ is a $\varepsilon$-sample of $|X|$, for every $\triangle u v w$ with midpoint $m$, there is a sample $p \in P$ such that $\lVert p - m \rVert \leq \varepsilon$. Hence, there is a spanning connected component $C$ in $\check{\mathcal{C}}_{\frac{3\varepsilon}{2}} \left(P_{LM,2}\right)$.

        Now, consider $C$ a spanning component of $\check{\mathcal{C}}_{\frac{3\varepsilon}{2}} \left(P_{LM,2}\right)$. Then, as $P$ is a $\varepsilon$-sample, there is some $\triangle u v w$ with midpoint $m$ such that there is a sample $p \in C$ with $\lVert p - m \rVert \leq \varepsilon$.

        Consider two triangles $\triangle u v w, \, \triangle u' v' w'$, and take two samples $p, p' \in P_{LM,2}$ with 

        \begin{equation*}
            d(\triangle u v w, p ), \, d( \triangle u' v' w', p') \leq \varepsilon.
        \end{equation*}

        As $p, p' \in P_{LM,2}$, we know that 

        \begin{equation*}
            d(\partial \triangle u v w, p ), \, d(\partial \triangle u' v' w', p') > R + \varepsilon,
        \end{equation*}

        and so by \Cref{as:embed}, $\lVert p - p' \rVert > 6 \varepsilon$.

        Hence, the spanning components of $\check{\mathcal{C}}_{\frac{3\varepsilon}{2}}\left( P_{LM,2}\right)$ are in bijection with the triangles of $X$.
    \end{proof-app}

\begin{proof-app}[Proof of \Cref{lem:lmedge-disconnected}]
        As $\overline{uv}$ is a locally maximal edge, there is a corresponding edge spanning component $\mathcal{E}$. As $u, v$ are not faces of any other cell $\sigma \in X$, by \Cref{as:embed} and \Cref{prop:lmedge,prop:nlmvertex}, the points $p \in P_{NLM}$ which witness $\mathcal{E}$ do not witness any other edge spanning component $\mathcal{E}'$ or any triangle spanning component $\mathcal{T}$. 
        
        Thus, there is a single partition $P_1$ of $P_{NLM}$ which contains all the samples $p$ that witness $\mathcal{E}$. By \Cref{as:embed}, there is no other partition $P_2$ of $P_{NLM}$ that witnesses $\mathcal{E}$. Hence, $P_1$ is assigned label $-1$. 
   \end{proof-app}

\begin{proof-app}[Proof of \Cref{lem:lmedge-connected}]
        As $\overline{uv}$ is a locally maximal edge, there is a corresponding edge spanning component $\mathcal{E}$. Without loss of generality, assume $v$ is the face of some locally maximal cell $\sigma \neq \overline{uv}$.

        By \Cref{as:embed} and \Cref{prop:lmtriangle,prop:lmedge,prop:nlmvertex}, there are samples $p_u, pv_v \in P_{NLM}$ such that 

        \begin{equation*}
            \lVert p_u - u \rVert, \, \lVert p_v - v \rVert \leq \varepsilon.
        \end{equation*}

        Further, there is a spanning connected component $\mathcal{C}$ which $p_v$ also witnesses but $p_u$ does not witness. Hence, there are two partitions $P_v, P_u$ which witness $\mathcal{E}$. By \cref{as:embed} and \Cref{alg:pnlmpart}, there are no other partitions which witness $\mathcal{E}$.

        Hence, both $P_v$ and $P_u$ are labelled with $0$ by \Cref{alg:P1class,alg:P2class}.
    \end{proof-app}

\begin{proof-app}[Proof of \Cref{lem:lmtriangle-disconnected}]
        Let $\mathcal{T}$ be the triangle spanning component that corresponds to $\triangle u v w$. By \Cref{as:embed,prop:lmtriangle,prop:nlmedge,prop:lmvertex}, the samples $p \in P_{NLM}$ that witness $\mathcal{T}$ do not witness any spanning connected component $\mathcal{C} \neq \mathcal{T}$. By \cref{as:embed} and \Cref{alg:pnlmpart} there is a unique connected component $P_1$ that witnesses $\mathcal{T}$.

        As $P$ is an $\varepsilon$-sample of $|X|$, and from Cref{prop:nlmedge,prop:lmvertex}, there are samples $p_u, p_v, p_w, p_{uv}, p_{vw}, p_{uw} \in P_1$ such that 

        \begin{align*}
            \lVert p_u -u \rVert, \, \lVert p_v -v \rVert, \, \lVert p_w -w \rVert &\leq \varepsilon, \\
            d(\overline{uv}, p_{uv}), \, d(\overline{vw}, p_{vw}), \, d(\overline{uw}, p_{uw})  & \leq \varepsilon.
        \end{align*} 

        Hence, $P_1$ is assigned label $7$ by \Cref{alg:P1class,alg:P2class}.
    \end{proof-app}

\begin{proof-app}[Proof of \Cref{lem:lmtriangle-share-1-vertex}]
        Let $\mathcal{T}$ be the triangle spanning component that corresponds to $\triangle u v w$. By \Cref{as:embed} and \Cref{prop:lmtriangle,prop:lmedge,prop:nlmedge,prop:lmvertex}, any spanning connected component $\mathcal{C}$ witnessed by samples $p \in P_{NLM}$ that witness $\mathcal{T}$ corresponds to a locally maximal cell $\tau$ such that $\triangle u v w \cap \tau \neq \emptyset$.

        We need to split into two cases:

        \begin{enumerate}
            \item there is a unique locally maximal cell $\tau \in X$ with $\triangle u v w \cap \tau = v$
            \item there are at least two locally maximal cells $\tau, \sigma \in X, \, \tau \neq \sigma$ with $\triangle u v w \cap \tau = \triangle u v w \cap \sigma = v$.
        \end{enumerate}

        \underline{Case 1:} We assumed there was a unique locally maximal $\tau$ with $\triangle u v w \cap \tau = v$, and hence, by \Cref{prop:equivlmspantri,prop:equivlmspanedge} there is some spanning component $\mathcal{C}_{\tau}$ which corresponds to $\tau$. with  By \Cref{as:embed} and \Cref{prop:lmtriangle,prop:lmedge,prop:nlmedge,prop:lmvertex}, in \Cref{alg:pnlmpart} there is a single partition $P_1$ of $P_{NLM}$ which witnesses $\mathcal{T}$ and $\mathcal{C}_{\tau}$, and there is a unique partition $P_2$ which witnesses just $\mathcal{T}$. Further, $P_1$ is assigned label $0$ and $P_2$ label $5$ by \Cref{alg:P1class,alg:P2class}.
         
        \underline{Case 2:} From our assumptions, there are two locally maximal cells $\tau, \sigma \in X, \, \tau \neq \sigma$ such that 

        \begin{equation*}
             \tau \cap \triangle u v w = v = \sigma \cap \triangle u v w.
        \end{equation*}

        By \Cref{prop:equivlmspantri,prop:equivlmspanedge} there is some spanning component $\mathcal{C}_{\tau}$ which corresponds to $\tau$, and some spanning component $\mathcal{C}_{\sigma}$ which corresponds to $\sigma$. 
        
        By \Cref{as:embed} and from \Cref{alg:pnlmpart}, there is a single partition $P_1$ of $P_{NLM}$ which witnesses $\mathcal{T}, \mathcal{C}_{\tau}, \mathcal{C}_{\sigma}$, and no partitions which witness a subset of these spanning components. This holds, by induction, for any locally maximal cell $\tau' \in X, \tau' \neq \tau, \sigma$ with $\tau' \cap \triangle u v w = v$. Similarly, there is a single partition $P_2$ of $P_{NLM}$ which witnesses only $\mathcal{T}$. Further, $P_1$ is assigned label $0$ and $P_2$ label $5$ by \Cref{alg:P1class,alg:P2class}. 
    \end{proof-app}

\begin{proof-app}[Proof of \Cref{lem:lmtriangle-share-1-edge-with-vertices}]
        Let $\mathcal{T}$ be the triangle spanning component that corresponds to $\triangle u v w$. By \Cref{as:embed} and \Cref{prop:lmtriangle,prop:lmedge,prop:nlmedge,prop:lmvertex}, any spanning connected component $\mathcal{C}$ witnessed by samples $p \in P_{NLM}$ that witness $\mathcal{T}$ corresponds to a locally maximal cell $\tau$ such that $\triangle u v w \cap \tau \neq \emptyset$.

        We need to split into two cases:

        \begin{enumerate}
            \item there is a unique locally maximal cell $\tau \in X$ with $\triangle u v w \cap \tau = \overline{uv}$
            \item there are at least two locally maximal cells $\tau, \sigma \in X, \, \tau \neq \sigma$ with $\triangle u v w \cap \tau =  \triangle u v w \cap \sigma = \overline{uv}$.
        \end{enumerate}

        \underline{Case 1:} We assumed there was a unique locally maximal $\tau$ with $\triangle u v w \cap \tau = \overline{uv}$, and hence, by \Cref{prop:equivlmspantri,prop:equivlmspanedge} there is some spanning component $\mathcal{C}_{\tau}$ which corresponds to $\tau$. with  By \Cref{as:embed},\Cref{prop:lmtriangle,prop:lmedge,prop:nlmedge,prop:lmvertex}, in \Cref{alg:pnlmpart} there is a single partition $P_1$ of $P_{NLM}$ which witnesses $\mathcal{T}$ and $\mathcal{C}_{\tau}$, and there is a unique partition $P_2$ which witnesses just $\mathcal{T}$. Further, $P_1$ is assigned label $1$ and $P_2$ label $4$ by \Cref{alg:P1class,alg:P2class}.
         
        \underline{Case 2:} From our assumptions, there are two locally maximal cells $\tau, \sigma \in X, \, \tau \neq \sigma$ such that 

        \begin{equation*}
             \tau \cap \triangle u v w = \overline{uv} = \sigma \cap \triangle u v w.
        \end{equation*}

        By \Cref{prop:equivlmspantri,prop:equivlmspanedge} there is some spanning component $\mathcal{C}_{\tau}$ which corresponds to $\tau$, and some spanning component $\mathcal{C}_{\sigma}$ which corresponds to $\sigma$. 
        
        By \Cref{as:embed} and from \Cref{alg:pnlmpart}, there is a single partition $P_1$ of $P_{NLM}$ which witnesses $\mathcal{T}, \mathcal{C}_{\tau}, \mathcal{C}_{\sigma}$, and no partitions which witness a subset of these spanning components. This holds, by induction, for any locally maximal cell $\tau' \in X, \tau' \neq \tau, \sigma$ with $\tau' \cap \triangle u v w = v$. Similarly, there is a single partition $P_2$ of $P_{NLM}$ which witnesses only $\mathcal{T}$. Further, $P_1$ is assigned label $1$ and $P_2$ label $4$ by \Cref{alg:P1class,alg:P2class}. 
    \end{proof-app}

\begin{proof-app}[Proof of \Cref{lem:lmtriangle-share-2-vertices}]
        Let $\mathcal{T}$ be the triangle spanning component which corresponds to $\triangle u v w$. Then, the proof is an adaption of the proof of \Cref{lem:lmtriangle-share-1-vertex}. By combining the arguments at the two shared vertices, there are three partitions $P_1, P_2, P_3$ from \Cref{alg:pnlmpart} which witness $\mathcal{T}$, and there are spanning connected components $\mathcal{C}_1, \mathcal{C}_2$ such that $P_1$ witnesses $\mathcal{C}_1$ but not $\mathcal{C}_2$, and $P_2$ witnesses $\mathcal{C}_2$ but not $\mathcal{C}_1$. Further, $P_3$ only witnesses $\mathcal{T}$. Hence, $P_1, P_2$ are labelled with $0$ and $P_3$ with $6$.
    \end{proof-app}

\begin{proof-app}[Proof of \Cref {lem:share-edge-with-vertex-and-vertex}]
        Let $\mathcal{T}$ be the triangle spanning component which corresponds to $\triangle u v w$. Then, the proof is an adaption of the proof of \Cref{lem:lmtriangle-share-1-vertex,lem:lmtriangle-share-1-edge-with-vertices}. By combining the arguments there are three partitions $P_1, P_2, P_3$ from \Cref{alg:pnlmpart} which witness $\mathcal{T}$, and there are spanning connected components $\mathcal{C}_1, \mathcal{C}_2$ such that $P_1$ witnesses $\mathcal{C}_1$ and $\mathcal{C}_2$, and $P_2$ witnesses $\mathcal{C}_2$ but not $\mathcal{C}_1$. Further, $P_3$ only witnesses $\mathcal{T}$. Hence, $P_1$ is labelled with $0$, $P_2$ with $1$ and $P_3$ with $3$.
\end{proof-app}

\begin{proof-app}[Proof of \Cref{lem:share-edge-with-vertices-and-vertex}]
        Let $\mathcal{T}$ be the triangle spanning component which corresponds to $\triangle u v w$. Then, the proof is an adaption of the proof of \Cref{lem:lmtriangle-share-1-vertex,lem:lmtriangle-share-1-edge-with-vertices}. By combining the arguments there are three partitions $P_1, P_2, P_3$ from \Cref{alg:pnlmpart} which witness $\mathcal{T}$, and there are spanning connected components $\mathcal{C}_1, \mathcal{C}_2$ such that $P_1$ witnesses $\mathcal{C}_1$ but not $\mathcal{C}_2$, and $P_2$ witnesses $\mathcal{C}_2$ but not $\mathcal{C}_1$. Further, $P_3$ only witnesses $\mathcal{T}$. Hence, $P_1$ is labelled with $0$, $P_2$ with $2$ and $P_3$ with $9$.
    \end{proof-app}

\begin{proof-app}[Proof of \Cref{lem:share-three-vertex}]
        Let $\mathcal{T}$ be the triangle spanning component which corresponds to $\triangle u v w$. Then, the proof is an adaption of the proof of \Cref{lem:lmtriangle-share-1-vertex,lem:lmtriangle-share-1-edge-with-vertices}. By combining the arguments there are three partitions $P_1, P_2, P_3,P_4$ from \Cref{alg:pnlmpart} which witness $\mathcal{T}$, and there are spanning connected components $\mathcal{C}_1, \mathcal{C}_2, \mathcal{C_3}$ such that $P_1$ witnesses $\mathcal{C}_1$ but not $\mathcal{C}_2, \mathcal{C}_3$, $P_2$ witnesses $\mathcal{C}_2$ but not $\mathcal{C}_1,\mathcal{C}_3$, and $P_2$ witnesses $\mathcal{C}_3$ but not $\mathcal{C}_1,\mathcal{C}_2$. Further, $P_4$ only witnesses $\mathcal{T}$. Hence, $P_1, P_2$ and $P_3$ are labelled with $0$ and $P_4$ with $8$.
    \end{proof-app}

\begin{proof-app}[Proof of \Cref{lem:share-edge-with-vertex-and-2-vertices}]
        Let $\mathcal{T}$ be the triangle spanning component which corresponds to $\triangle u v w$. Then, the proof is an adaption of the proof of \Cref{lem:lmtriangle-share-1-vertex,lem:lmtriangle-share-1-edge-with-vertices}. By combining the arguments there are three partitions 
        \[
            P_1, P_2, P_3,P_4
        \]
        from \Cref{alg:pnlmpart} which witness $\mathcal{T}$, and there are spanning connected components $\mathcal{C}_1, \mathcal{C}_2, \mathcal{C_3}$ such that $P_1$ witnesses $\mathcal{C}_1$ but not $\mathcal{C}_2, \mathcal{C}_3$, $P_2$ witnesses $\mathcal{C}_1, \mathcal{C}_2$ but not $\mathcal{C}_3$, and $P_2$ witnesses $\mathcal{C}_3$ but not $\mathcal{C}_1,\mathcal{C}_2$. Further, $P_4$ only witnesses $\mathcal{T}$. Hence, $P_1, P_2$ and $P_3$ are labelled with $0$ and $P_4$ with $8$.
    \end{proof-app}

\begin{proof-app}[Proof of \Cref{lem:share-edge-with-2-vertices}]
        Let $\mathcal{T}$ be the triangle spanning component which corresponds to $\triangle u v w$. Then, the proof is an adaption of the proof of \Cref{lem:lmtriangle-share-1-vertex,lem:lmtriangle-share-1-edge-with-vertices}. By combining the arguments there are three partitions $P_1, P_2, P_3,P_4$ from \Cref{alg:pnlmpart} which witness $\mathcal{T}$, and there are spanning connected components $\mathcal{C}_1, \mathcal{C}_2, \mathcal{C_3}$ such that $P_1$ witnesses $\mathcal{C}_1$ but not $\mathcal{C}_2, \mathcal{C}_3$, $P_2$ witnesses $\mathcal{C}_1, \mathcal{C}_2$ but not $\mathcal{C}_3$, and $P_2$ witnesses $\mathcal{C}_1, \mathcal{C}_3$ but not $\mathcal{C}_2$. Further, $P_4$ only witnesses $\mathcal{T}$. Hence, $P_1$ is labelled with $3$, $P_2, P_3$ with $0$ and $P_4$ with $4$.
    \end{proof-app}

\begin{proof-app}[Proof of \Cref{lem:share-2-edge-with-vertex-and-vertex}]
        Let $\mathcal{T}$ be the triangle spanning component which corresponds to $\triangle u v w$. Then, the proof is an adaption of the proof of \Cref{lem:lmtriangle-share-1-vertex,lem:lmtriangle-share-1-edge-with-vertices}. By combining the arguments there are four partitions $P_1, P_2, P_3,P_4$ from \Cref{alg:pnlmpart} which witness $\mathcal{T}$, and there are spanning connected components $\mathcal{C}_1, \mathcal{C}_2, \mathcal{C_3}$ such that $P_1$ witnesses $\mathcal{C}_1, \mathcal{C}_2, \mathcal{C}_3$, $P_2$ witnesses $\mathcal{C}_1, \mathcal{C}_2$ but not $\mathcal{C}_3$, and $P_3$ witnesses $\mathcal{C}_1, \mathcal{C}_3$ but not $\mathcal{C}_2$. Further, $P_4$ only witnesses $\mathcal{T}$. Hence, $P_1$ is labelled with $0$, $P_2, P_3$ with $1$, and $P_3$ with $3$.
    \end{proof-app}

\begin{proof-app}[Proof of \Cref{lem:share-edge-vertex-vertex-vertex}]
        Let $\mathcal{T}$ be the triangle spanning component which corresponds to $\triangle u v w$. Then, the proof is an adaption of the proof of \Cref{lem:lmtriangle-share-1-vertex,lem:lmtriangle-share-1-edge-with-vertices}. By combining the arguments there are four partitions $P_1, P_2, P_3, P_4$ from \Cref{alg:pnlmpart} which witness $\mathcal{T}$, and there are spanning connected components $\mathcal{C}_1, \mathcal{C}_2, \mathcal{C_3}, \mathcal{C_3}$ such that $P_1$ witnesses $\mathcal{C}_1$ and not $\mathcal{C}_2, \mathcal{C}_3$, $P_2$ witnesses $\mathcal{C}_1$ and not $\mathcal{C}_2, \mathcal{C}_3$, and $P_3$ witnesses $\mathcal{C}_3$ but not $\mathcal{C}_1, \mathcal{C}_2$. Further, $P_4$ only witnesses $\mathcal{T}$. Hence, $P_1, P_2, P_3$ are labelled with $0$, and $P_4$ with $8$.
    \end{proof-app}

\begin{proof-app}[Proof of \Cref{lem:share-edge-with-vertex-edge-vertex-vertex}]
        Let $\mathcal{T}$ be the triangle spanning component which corresponds to $\triangle u v w$. Then, the proof is an adaption of the proof of \Cref{lem:lmtriangle-share-1-vertex,lem:lmtriangle-share-1-edge-with-vertices}. By combining the arguments, there are five partitions 
        
        \[
            P_1, P_2, P_3, P_4,P_5
        \] 
        
        from \Cref{alg:pnlmpart} which witness $\mathcal{T}$, and there are spanning connected components $\mathcal{C}_1, \mathcal{C}_2, \mathcal{C_3}, \mathcal{C_3}, \mathcal{C_4}$ such that $P_1$ witnesses $\mathcal{C}_1,\mathcal{C}_2, \mathcal{C_4}$ and not $\mathcal{C}_3$, $P_2$ witnesses $\mathcal{C}_2$ and not $\mathcal{C}_1, \mathcal{C}_3, \mathcal{C}_4$, $P_3$ witnesses $\mathcal{C}_2, \mathcal{C}_3$ but not $\mathcal{C}_1, \mathcal{C}_4$, and $P_4$ witnesses $\mathcal{C}_4$ but not $\mathcal{C}_1, \mathcal{C}_2, \mathcal{C}_3$. Further, $P_5$ only witnesses $\mathcal{T}$, and hence $P_4$ only witnesses $\mathcal{T}$. Hence, $P_1, P_2$ are labelled with $0$, $P_3$ with $1$, and $P_4, P_5$ with $3$.
    \end{proof-app}

\begin{proof-app}[Proof of \Cref{lem:share-edge-edge-vertex-vertex-vertex}]
        Let $\mathcal{T}$ be the triangle spanning component which corresponds to $\triangle u v w$. Then, the proof is an adaption of the proof of \Cref{lem:lmtriangle-share-1-vertex,lem:lmtriangle-share-1-edge-with-vertices}. By combining the arguments there are six partitions 
            \[ P_1, P_2, P_3, P_4, P_5, P_6 \]
        from \Cref{alg:pnlmpart} which witness $\mathcal{T}$, and there are spanning connected components $\mathcal{C}_1, \mathcal{C}_2, \mathcal{C}_3, \mathcal{C}_4,\mathcal{C}_5$ such that $P_1$ witnesses $\mathcal{C}_1,\mathcal{C}_4$ and not $\mathcal{C}_2, \mathcal{C}_3, \mathcal{C}_5$, $P_2$ witnesses $\mathcal{C}_2,\mathcal{C}_4, \mathcal{C}_5$ and not $\mathcal{C}_1, \mathcal{C}_3$, $P_3$ witnesses $\mathcal{C}_2, \mathcal{C}_3, \mathcal{C}_5$ but not $\mathcal{C}_1, \mathcal{C}_4$, $P_4$ witnesses $\mathcal{C}_4$ but not $\mathcal{C}_1, \mathcal{C}_2, \mathcal{C}_3, \mathcal{C}_5$, and $P_5$ witnesses $\mathcal{C}_5$ but not $\mathcal{C}_1, \mathcal{C}_2, \mathcal{C}_3, \mathcal{C}_4$. Further, $P_6$ only witnesses $\mathcal{T}$, and hence $P_1, P_2, P_3$ are labelled with $0$, $P_4, P_5, P_6$ with $3$.
    \end{proof-app}

\begin{proof-app}[Proof of \Cref{lem:share-edge-edge-edge-vertex-vertex-vertex}]
        Let $\mathcal{T}$ be the triangle spanning component which corresponds to $\triangle u v w$. Then, the proof is an adaption of the proof of \Cref{lem:lmtriangle-share-1-vertex,lem:lmtriangle-share-1-edge-with-vertices}. By combining the arguments there are six partitions 
            \[P_1, P_2, P_3, P_4, P_5, P_6\]
        from \Cref{alg:pnlmpart} which witness $\mathcal{T}$, and there are spanning connected components $\mathcal{C}_1, \mathcal{C}_2, \mathcal{C}_3, \mathcal{C}_4,\mathcal{C}_5, \mathcal{C}_6$ such that $P_1$ witnesses $\mathcal{C}_1,\mathcal{C}_4, \mathcal{C}_6$ and not $\mathcal{C}_2, \mathcal{C}_3, \mathcal{C}_5$, $P_2$ witnesses $\mathcal{C}_2,\mathcal{C}_4, \mathcal{C}_5$ and not $\mathcal{C}_1, \mathcal{C}_3, \mathcal{C}_6$, $P_3$ witnesses $\mathcal{C}_3, \mathcal{C}_5, \mathcal{C}_6$ but not $\mathcal{C}_1, \mathcal{C}_2\mathcal{C}_4$, $P_4$ witnesses $\mathcal{C}_4$ but not $\mathcal{C}_1, \mathcal{C}_2, \mathcal{C}_3, \mathcal{C}_5, \mathcal{C}_6$, and $P_5$ witnesses $\mathcal{C}_5$ but not $\mathcal{C}_1, \mathcal{C}_2, \mathcal{C}_3, \mathcal{C}_4, \mathcal{C}_6$, and $P_6$ witnesses $\mathcal{C}_6$ but not $\mathcal{C}_1, \mathcal{C}_2, \mathcal{C}_3, \mathcal{C}_4, \mathcal{C}_5$. Hence $P_1, P_2, P_3$ are labelled with $0$, $P_4, P_5, P_6$ with $3$.
    \end{proof-app}




\end{appendices}


\bibliography{bibliography}

\end{document}